\documentclass[11pt,a4paper]{article}

\usepackage[margin=23mm]{geometry}

\usepackage{amssymb,amsmath,amsfonts,amsthm, amscd, bbm, bigints, mathrsfs, helvet, mathtools, hyperref,appendix}
\usepackage{color}
\usepackage{comment}
\usepackage{tikz}
\usepackage{tikz-cd}
\usetikzlibrary{cd}
\usetikzlibrary{matrix,calc}
\usetikzlibrary{decorations.markings,shapes.geometric,shapes.misc}
\usetikzlibrary{decorations.pathreplacing,positioning}
\usetikzlibrary{arrows}
\usetikzlibrary{hobby}
\usetikzlibrary{patterns, fadings}

\definecolor{myPurple}{rgb}{0.5,0.1,0.6}
\definecolor{myOrange}{rgb}{1.0,0.5,0.0}
\definecolor{myRed}{rgb}{1.0,0.0,0.0}
\definecolor{myGreen}{rgb}{0.0,0.5,0.0}
\definecolor{LatexBlue}{rgb}{0.211765,0.227451,0.666667}
\definecolor{myBlue}{rgb}{0.0,0.0,1.0}
\definecolor{myBlack}{rgb}{0.0,0.0,0.0}
\definecolor{myGray}{rgb}{0.3,0.3,0.3}

\theoremstyle{plain}
\newtheorem{theorem}{Theorem}[section]
\newtheorem*{theorem*}{Theorem}
\newtheorem{proposition}[theorem]{Proposition}

\newtheorem*{proposition*}{Proposition}

\theoremstyle{definition}
\newtheorem{definition}[theorem]{Definition}

\newtheorem{remarkx}[theorem]{Remark}

\newenvironment{remark}
  {\pushQED{\qed}\remarkx}
  {\popQED\endremarkx}

\newcommand{\Res}{\mathop{\mathrm{Res}}}

\DeclareMathOperator{\Aut}{Aut}

\DeclareMathOperator{\Ad}{Ad}

\DeclareMathOperator{\Tr}{Tr}

\newcommand{\SimTo}{\xrightarrow{\raisebox{-0.3 em}{\smash{\ensuremath{\text{\tiny$\cong$}}}}}}

\newcommand{\g}{\mathfrak{g}}
\newcommand{\h}{\mathfrak{h}}

\newcommand{\CP}{\mathbb{P}^1}

\newcommand{\sfp}{\mathsf{p}}
\newcommand{\sfq}{\mathsf{q}}

\newcommand{\Lag}{\mathscr{L}}

\newcommand{\sgm}{\sigma}

\newcommand{\omg}{\omega}
\newcommand{\zbar}{\bar{z}}
\newcommand{\taubar}{\bar{\tau}}

\newcommand{\Q}{\mathcal{Q}}

\newcommand{\sfE}{\mathsf{E}}
\newcommand{\sfH}{\mathsf{H}}

\def\sl{\mathfrak{sl}}
\def\id{\textup{id}}
\def\h{\mathfrak{h}}

\def\d{\mathrm{d}}

\newcommand{\lau}[1]{(\kern-.2em( #1 )\kern-.2em)}

\newcommand{\parder}[2]{\frac{\partial #1}{\partial #2}}
\newcommand{\ie}{{\it i.e.}\ }

\def\be{\begin{equation}}
\def\ee{\end{equation}}
\def\bea{\begin{eqnarray}}
\def\eea{\end{eqnarray}}

\def\A{\mathcal{A}}

\def\CC{\mathbb{C}}
\def\CP{\mathbb{C}P^1}

\def\RR{\mathbb{R}}

\def\ZZ{\mathbb{Z}}

\def\M{\mathcal M}

\def\dol{\bar{\partial}}

\def\1{\bm{1}}

\numberwithin{equation}{section}

\begin{document}

\title{\textbf{The 3d mixed BF Lagrangian 1-form\\
\large -- A variational formulation of Hitchin's integrable system --
}\vspace{8pt}}

\author{Vincent Caudrelier$^{1}$\footnote{Corresponding author: v.caudrelier@leeds.ac.uk} \and Derek Harland$^1$ \and Anup Anand Singh$^{1,3}$ \and Beno\^{\i}t Vicedo$^2$}
\date{%
\emph{$^1$School of Mathematics, University of Leeds, Leeds LS2 9JT, UK}\\%
\emph{$^2$Department of Mathematics, University of York, York YO10 5GH, UK}\\%
\emph{$^3$Department of Mathematical Sciences, Loughborough University, Loughborough LE11 3TU, UK}\\[2ex]%
}

\maketitle

\begin{abstract}
We introduce the concept of gauged Lagrangian $1$-forms, extending the notion of Lagrangian $1$-forms to the setting of gauge theories. This general formalism is applied to a natural geometric Lagrangian $1$-form on the cotangent bundle of the space of holomorphic structures on a smooth principal $G$-bundle $\mathcal{P}$ over a compact Riemann surface $C$ of arbitrary genus $g$, with or without marked points, in order to gauge the symmetry group of smooth bundle automorphisms of $\mathcal{P}$. The resulting construction yields a multiform version of the $3$d mixed BF action with so-called type A and B defects, providing a variational formulation of Hitchin’s completely integrable system over $C$. By passing to holomorphic local trivialisations and going partially on-shell, we obtain a \emph{unifying action} for a hierarchy of Lax equations describing the Hitchin system in terms of meromorphic Lax matrices. The cases of genus $0$ and $1$ with marked points are treated in greater detail, producing explicit Lagrangian $1$-forms for the rational Gaudin hierarchy and the elliptic Gaudin hierarchy, respectively, with the elliptic spin Calogero–Moser hierarchy arising as a special subcase.
\end{abstract}

\setcounter{tocdepth}{2}
\tableofcontents

\section{Introduction and summary of main results}

The Hamiltonian formalism has historically played a more prominent role in shaping the modern theory of integrable systems than the Lagrangian one. Indeed, it naturally provides a description of \emph{integrable hierarchies} as commuting Hamiltonian flows and its algebraic nature has been pivotal in the problem of quantisation. However, over the past fifteen years, there has been a resurgence of interest in both applying and further developing Lagrangian methods in the context of integrable systems, driven by two major developments.

\medskip

The first of these is the development of the theory of \emph{Lagrangian multiforms}, first introduced by Lobb and Nijhoff in the setting of integrable lattice models \cite{LN}. Since then, significant progress was made in extending the multiform framework to various realms of integrability: discrete and continuous finite-dimensional systems \cite{YKLN,Su,PS,CDS,CSV}, integrable field theories in $1+1$ dimensions \cite{SuV,SNC1,SNC2,CS1,CS2,PV,CStV2} and $2+1$ dimensions \cite{SNC3,N2}, and semi-discrete systems \cite{SV}. The concept was even extended to non-commuting flows in \cite{CNSV,CH}.

Unlike traditional action functionals in $d$ spacetime dimensions, which capture the dynamics of a single $d$-dimensional integrable field theory (the $d=1$ case corresponding to finite-dimensional systems), a multiform action functional is given by the integral of a Lagrangian $d$-form over a $d$-dimensional \textit{sub}manifold $\Gamma$ in a so-called \emph{multi-time} space of dimension strictly greater than $d$. The resulting action functional $S[q, \Gamma]$ depends not only on the field configurations $q$ but also on the submanifold $\Gamma$. Applying a generalised variational principle to this action produces the so-called \emph{multi-time Euler-Lagrange equations}, encoding the dynamics of several $d$-dimensional integrable field theories from the same hierarchy, together with the so-called \emph{closure relation} which provides a variational analogue of the Poisson involutivity of the Hamiltonians of this hierarchy.

\medskip

The second development concerns the deep relationship between the notion of Lax integrability and certain so-called \emph{mixed holomorphic-topological (HT) gauge theories}. The general philosophy is that the Lax formalism of an integrable field theory on a $d$-dimensional manifold $M$ (again the $d=1$ case corresponds to finite-dimensional systems), depending on a spectral parameter living on a Riemann surface $C$, arises from the equations of motion of a suitable gauge theory on $M \times C$ which is holomorphic along $C$ and (possibly a mixture of holomorphic and) topological along $M$.

The first concrete example, discovered by Costello and further developed more recently with Witten and Yamazaki in \cite{Co1, Co2, CoWY1, CoWY2}, relates integrable lattice models to a $4$-dimensional semi-holomorphic variant of Chern-Simons theory in the presence of line defects. This idea has been extensively studied and generalised over the past five years. Most notably, integrable field theories in $1+1$ dimensions were shown to arise from surface defects in the same $4$-dimensional Chern-Simons theory in the seminal paper \cite{CoY} and many subsequent works, see for instance \cite{DLMV, Sch, BL, FSY1, HL, LaV, CStV1, FSY2, HTC, FSY3, FSY4, LiV, BP, LaW}.
Integrable field theories in higher dimensions arising as reductions of the $4$-dimensional anti-self-dual Yang-Mills equations were similarly derived from various defects in $6$-dimensional holomorphic Chern-Simons theory on twistor space \cite{BS, CCHLT1, CCHLT2}. But the example most relevant for the purpose of the present paper is the relationship between the finite (rational) Gaudin model and the $3$-dimensional mixed holomorphic-topological BF theory with line defects \cite{VW}.

\medskip

It is important to stress that, while both developments are rooted in the Lagrangian formalism, their scopes are very different and each has their own advantages and drawbacks. Indeed, the main purpose of the Lagrangian multiform framework is to encode \textit{hierarchies} of integrable systems, but the only known systematic constructions \cite{CStV2,CDS,CSV} rely on the algebraic machinery of the classical $r$-matrix method. In contrast, the framework of mixed HT gauge theories provides an elegant geometric origin for the spectral parameter and a powerful method for constructing many integrable systems, but it is currently limited to producing individual integrable systems rather than hierarchies. It is therefore very tempting to try to combine the two frameworks.

The primary purpose of the present paper is to take the first step in this direction\footnote{In connection with combining the framework of Lagrangian multiforms with gauge-theoretic ideas, we note the work \cite{MNR}. However, it is based on a different approach from mixed HT gauge theories. Lagrangian multiforms for the Darboux--Kadomtsev--Petviashvili system were derived from a hierarchy of Chern-Simons actions in an infinite-dimensional space of Miwa variables.}, in the setting of finite-dimensional integrable systems. Specifically, we will focus on Hitchin's integrable system \cite{H} and its generalisations with marked points introduced in subsequent works (see, for instance, \cite{Ma, Bot, Ne, ER}) which provide master systems from which prototypical examples of completely integrable finite-dimensional Hamiltonian systems can be derived. Focusing in the introduction on the case without marked points for conciseness, the phase space of the Hitchin system is described by the cotangent bundle $T^\ast \text{Bun}_G(C)$ of the moduli space $\text{Bun}_G(C)$ of holomorphic principal $G$-bundles of a fixed topological type on a Riemann surface $C$ of arbitrary genus $g \geq 2$. The Hitchin map describes a complete set of Poisson-commuting Hamiltonians on this phase space, thereby encoding a finite-dimensional integrable hierarchy.

Our main goal is to construct a Lagrangian $1$-form which encodes the entire integrable hierarchy of Hitchin's system variationally. In doing so, we will find that we are naturally led to a multiform version of the $3$d mixed BF theory introduced in \cite{VW}, see also \cite{Z, GW, GRW}. We thus generalise \cite{VW} in three key respects: 1) we consider Hitchin systems with marked points on Riemann surfaces $C$ of \emph{arbitrary} genus rather than just $\CP$, 2) we encode the entire hierarchy of such Hitchin systems rather than one individual flow in this hierarchy, and 3) we \textit{derive} the $3$d mixed BF action as a variational description of Hitchin systems rather than taking that action as a given. The present paper thus naturally sits at the intersection of the following three topics: Hitchin's integrable systems, mixed HT gauge theories and Lagrangian multiforms. 
\begin{center}
\begin{tikzpicture}
  \def\xrad{2.6}
  \def\yrad{1.2}

  \coordinate (A) at (0.1, -0.3);
  \coordinate (B) at (2.6, -0.2);
  \coordinate (C) at (1.4, 1);

  \definecolor{colorA}{RGB}{70,130,180}
  \definecolor{colorB}{RGB}{220,20,60}
  \definecolor{colorC}{RGB}{60,179,113}

  \filldraw[fill=colorA, fill opacity=0.1, draw=colorA, thick, rotate=12] 
    (A) ellipse [x radius=\xrad, y radius=\yrad + 0.1];
  \filldraw[fill=colorB, fill opacity=0.1, draw=colorB, thick, rotate=-12] 
    (B) ellipse [x radius=\xrad - 0.1, y radius=\yrad + 0.1];
  \filldraw[fill=colorC, fill opacity=0.1, draw=colorC, thick] 
    (C) ellipse [x radius=\xrad-0.7, y radius=\yrad+0.4];

  \node[align=center, text=colorA] at (-1, -0.6) {mixed HT\\ gauge theories};
  \node[align=center, text=colorB] at (3.8, -0.6) {Lagrangian\\multiforms};
  \node[align=center, text=colorC] at (1.5, 1.8) {Hitchin's\\ integrable system};

  \node at (1.4, 0.2) {This paper};
\end{tikzpicture}
\end{center}
As such, it provides a first instance of the merging of the framework of Lagrangian multiforms with mixed HT gauge theories, using Hitchin's integrable system as a driving example.

\bigskip

In the remainder of the introduction we give a detailed outline of the structure of the paper and summarise the main results. For ease of presentation in this introduction, we will focus here on the case of Hitchin's system on a Riemann surface $C$ of arbitrary genus $g \geq 2$ \emph{without} marked points. The case with marked points is dealt with in detail in the main body of the paper.

In the absence of marked points, the phase space of Hitchin's system is given by the cotangent bundle $T^\ast \text{Bun}_G(C)$. A point in the base $\text{Bun}_G(C)$ is a stable holomorphic principal $G$-bundle $\mathcal P_{\rm hol}$, of a fixed topological type, which can always be described using a single holomorphic transition function $\gamma : U_0 \cap U_1 \to G$ relative to an open cover $\{ U_0, U_1 \}$ of $C$ with $U_0$ an open neighbourhood of a fixed point $\mathsf p \in C$ and $U_1 \coloneqq C \setminus \{ \mathsf p \}$. A point in the fibre $T^\ast_{\mathcal P_{\rm hol}} \text{Bun}_G(C)$ above $\mathcal P_{\rm hol}$ is called a holomorphic Higgs field $L$. Identifying $\g$ with its dual $\g^\ast$ using a fixed nondegenerate invariant bilinear pairing on $\g$, it is given explicitly by a pair of $\g$-valued $(1,0)$-forms $L^0$ and $L^1$ on $U_0$ and $U_1$, respectively, which are related via the adjoint action $L^0 = \gamma L^1 \gamma^{-1}$ on the overlap $U_0 \cap U_1$. The pair $(\mathcal P_{\rm hol}, L)$ is an example of a stable Higgs bundle. The Hitchin map, which is also known as the Hitchin fibration, provides a complete set of Poisson-commuting Hamiltonians $H_i$ for $i =1,\ldots, n$, where $n$ denotes half the dimension of the phase space. These Hamiltonians induce commuting time flows $\partial_{t^i}$, $i=1,\ldots, n$ on $T^\ast \text{Bun}_G(C)$ whose actions on the pair $(L^0, \gamma)$ are given by
\begin{equation} \label{Hitchin flow intro}
\partial_{t^i} L^0 = [M^0_i, L^0], \qquad
M^0_i = \gamma M^1_i \gamma^{-1} + \partial_{t^i} \gamma \gamma^{-1}
\end{equation}
for $i =1, \ldots, n$, where $M^1_i$ are $\g$-valued meromorphic functions on $U_1$ each with a simple pole at a fixed marked point $\mathsf q_i \in U_1$ whose residue there is determined by the Hitchin Hamiltonian $H_i$, and $M^0_i$ are $\g$-valued holomorphic functions on $U_0$.

\medskip

One of the main goals of the present paper is to give a variational description of the hierarchy of commuting flows \eqref{Hitchin flow intro}. To do so, we will exploit the fact that the Hitchin phase space $T^\ast \text{Bun}_G(C)$ can be obtained as a symplectic reduction of the infinite-dimensional cotangent bundle $T^\ast \mathcal M$ of the space $\mathcal M$ of stable holomorphic structures on a fixed smooth principal $G$-bundle $\mathcal P$ by the action of the group $\mathcal G = \Aut\mathcal P$ of smooth bundle automorphisms of $\mathcal P$, as we now briefly recall.

A holomorphic structure on $\mathcal P$ can be specified in terms of a partial\footnote{The notion of partial connection is explained below after \eqref{A definition}.} $(0,1)$-connection $A''$ on $\mathcal P$ given in local coordinate patches by $\g$-valued $(0,1)$-forms on $C$. A smooth bundle automorphism $g \in \mathcal G$ acts on such a holomorphic structure $A''$ by gauge transformations $A'' \mapsto g A'' g^{-1} - \bar\partial g g^{-1}$ and two holomorphic structures related in this way define the same holomorphic principal $G$-bundle $\mathcal P_{\rm hol}$. In particular, the action of $\mathcal G$ on the space $\mathcal M$ of stable holomorphic structures is free and we have an isomorphism $\mathcal M / \mathcal G \cong \text{Bun}_G(C)$. Moreover, the action of $\mathcal G$ on $\mathcal M$ lifts to a Hamiltonian action on $T^\ast \mathcal M$ with moment map given by $\mu(B, A'') = \bar\partial^{A''} B$, where $B$ denotes the smooth Higgs field parametrising the fibre of $T^\ast_{A''} \mathcal M$ over a fixed holomorphic structure $A'' \in \mathcal M$. We prefer to use the symbol $B$ over the more standard $\Phi$ as the Higgs field will shortly become identified with the corresponding field with the same name in $3$d mixed BF theory. The starting point for our construction is then the fact that the Hitchin phase space is given by the symplectic quotient
\begin{equation} \label{Hitchin phase space intro}
T^\ast \text{Bun}_G(C) \;\cong\; \mu^{-1}(0) / \mathcal G \,.
\end{equation}
Specifically, in \S\ref{sec: Lag for Hitchin} we introduce a natural lift of the Hitchin map to the cotangent bundle $T^\ast \mathcal M$ which induces $n$ commuting flows on this infinite-dimensional symplectic manifold and that can be described variationally using a natural geometric multiform action $S_\Gamma[B, A'', t]$ on $T^\ast \mathcal M$. Upon performing the symplectic reduction to $\mu^{-1}(0) / \mathcal G$ at the level of the action, which is the content of \S\ref{sec: 3d BF}, we obtain a multiform version of the action for the gauge theory known as $3$d mixed holomorphic-topological BF theory. After subsequently passing to holomorphic local trivialisations and going partially on-shell to obtain an equivalent description of the dynamics on $T^\ast \text{Bun}_G(C)$ via the isomorphism \eqref{Hitchin phase space intro}, which is the content of \S\ref{sec: Hitchin holomorphic}, this will lead to the desired action for the dynamics \eqref{Hitchin flow intro} of the Hitchin system on $T^\ast \text{Bun}_G(C)$.

\medskip

In order to describe this in more detail, it is useful to first recall how the analog of symplectic reduction can be implemented in the Lagrangian formalism. Consider a symplectic manifold $T^\ast M$ equipped with a Hamiltonian action of a Lie group $G$ with moment map $\mu : T^\ast M \to \g^\ast$. If the Hamiltonian is $G$-invariant then so is the first-order action $S[p,q] = \int_0^1 \big( p_\mu \dot{q}^\mu - H(p,q) \big) dt$ and symplectic reduction to $\mu^{-1}(0) / G$ is implemented by adding the term $\int_0^1 \langle \mu(p,q) , \mathcal A \rangle$ to the action where $\mathcal A = \mathcal A_t dt$ is a $\g$-valued Lagrange multiplier. Indeed, the restriction to the level set $\mu^{-1}(0)$ is implemented dynamically via the equation of motion for $\mathcal A$ while the quotient by $G$ is implemented as a gauge symmetry since the global $G$-symmetry of the original action is promoted to a gauge symmetry of the new action if we let $\mathcal A$ transform as gauge field $\mathcal A \mapsto g \mathcal A g^{-1} - \partial_t g g^{-1} \d t$ for any $g \in C^\infty\big( (0,1), G \big)$. In \S\ref{sec: univar finite-dim} we generalise this to the setting of Lagrangian multiforms.

The symplectic reduction in \eqref{Hitchin phase space intro} is implemented at the level of the action in \S\ref{sec: Lag for Hitchin mod G}. Specifically, starting from the geometric multiform action $S_\Gamma[B, A'', t]$ on $T^\ast \mathcal M$, where in the present multiform setting the fields $B$, $A''$ and also the time variables $t^i$ depend on an auxiliary space of multi-times $\RR^n$, we introduce a $\g$-valued gauge field $\mathcal A = \mathcal A_j \d u^j$ (with an implicit sum over $j=1,\ldots, n$) which transforms as a connection along $\RR^n$, namely $\mathcal A \mapsto g \mathcal A g^{-1} - \d_{\RR^n} g g^{-1}$. To implement the symplectic reduction by the group $\mathcal G$ we then add the new term $\int_C \langle \mu, \mathcal A \rangle = - \int_C \langle B, \bar\partial^{A''} \mathcal A \rangle$ to the Lagrangian $1$-form $\int_C \langle B, \d_{\RR^n} A'' \rangle - H_i(B, A'') \d_{\RR^n} t^i$ underlying the original action $S_\Gamma[B, A'', t]$. Remarkably, the resulting Lagrangian $1$-form is exactly a multiform version of the Lagrangian for $3$d mixed BF theory; see Theorem \ref{thm: BF Lagrangian}. In particular, the partial $(0,1)$-connection $A''$ and $\RR^n$-connection $\mathcal A$ combine into a partial connection $A = A'' + \mathcal A$ which together with the Higgs field $B$ form the field content of $3$d mixed BF theory. Moreover, the lift of the Hitchin Hamiltonians $H_i$ to the cotangent bundle $T^\ast \mathcal M$ now play the role of type B defects in the language of \cite{VW}.

By construction, the Lagrangian $1$-form for $3$d mixed BF theory encodes the dynamics of the Hitchin system but in terms of the degrees of freedom of the cotangent bundle $T^\ast \mathcal M$, i.e. in terms of a smooth partial $(0,1)$-connection $A''$ specifying the holomorphic structure of $P$ parametrising the base, a smooth Higgs field $B$ parametrising the fibre, and a smooth $\RR^n$-connection $\mathcal A$ serving as Lagrange multiplier for the constraint $\bar\partial^{A''} B = 0$. To obtain the sought-after Lagrangian $1$-form describing the dynamics of the Hitchin system on the actual Hitchin phase space $T^\ast \text{Bun}_G(C)$, which is parametrised by pairs $(\mathcal P_{\rm hol}, L)$ with $\mathcal P_{\rm hol}$ a stable holomorphic principal $G$-bundle on $C$ and $L$ a holomorphic Higgs fields, we need to perform one final step detailed in \S\ref{sec: Hitchin holomorphic}. Specifically, the holomorphic principal $G$-bundle $\mathcal P_{\rm hol}$ is obtained by moving to a local trivialisation of the smooth principal $G$-bundle $\mathcal P$ in which $A'' = 0$. Moreover, to turn the smooth Higgs field $B$ into the holomorphic Higgs field $L$ we simply enforce the constraint $\bar\partial^{A''} B = 0$. The resulting Lagrangian $1$-form, derived in Theorem \ref{thm: Lag 1-form for Hitchin}, reads
\begin{equation} \label{Lag 1-form intro}
\Lag_{\rm H} = - \frac{1}{2 \pi i} \int_{c_{\mathsf p}} \big\langle L^0, \d_{\RR^n} \gamma \gamma^{-1} \big\rangle - H_i \d_{\RR^n} t^i \,,
\end{equation}
where $c_{\mathsf p}$ is a small counter-clockwise oriented circle around the marked point $\mathsf p$. More precisely, in Theorem \ref{thm: Lag 1-form for Hitchin} we derive the more general Lagrangian $1$-form for a Hitchin system with marked points $\mathsf p_\alpha \in C$, for $\alpha = 1, \ldots, N$ with $N \in \ZZ_{\geq 1}$, at which we attach degrees of freedom $\varphi_\alpha$ living in coadjoint orbits of fixed elements $\Lambda_\alpha \in \g^\ast$. See \S\ref{sec: adding punctures} for details. In Theorem \ref{thm: eom for 1d action} we show that the variation of the corresponding action reproduces the equations of motion \eqref{Hitchin flow intro} of the Hitchin system, as expected.
A visual summary of the whole procedure, detailing the passage from the geometrical action on $T^\ast \mathcal M$ to that of the Hitchin system on $T^\ast \text{Bun}_G(C)$, is sketched in Figure \ref{fig: paper schematic}.
\begin{figure}[ht]
\begin{tikzcd}[
  row sep=1em,
  column sep=1.5em,
  cells={nodes={minimum width=2.5cm, minimum height=0cm, align=center}}
]

  \textbf{\small \quad Hitchin's system on} & \textbf{\small \quad Phase-space parametrised by} & \textbf{\small Variational description} \\[-1mm]

  {\tikz[baseline=(X.base)]
    \node[draw, thick, rounded corners, align=center, minimum width=2.5cm, minimum height=1cm] (X)
    {$T^* \mathcal{M}$};}
  \arrow[d, shorten <=-1pt, shorten >=-1pt, "\;\text{Symplectic}\;"' {pos=0.27}, "\;\text{reduction}" {pos=0.25}, "\; \mu \,:\, T^\ast \mathcal M \,\to\, \mathfrak{G}^\ast" {pos=0.65}] \arrow[in=15, out=-15, shorten <=-5pt, shorten >=-5pt, loop, looseness=2, "\mathcal{G}"'] & \makebox[6cm][l]{\small $\begin{array}{ll}
\mathcal P\phantom{''} \quad  \text{smooth principal } G\text{-bundle} \\
A'' \quad  \text{holomorphic structure}\\
B\phantom{''} \quad  \text{smooth Higgs field} \end{array}$}
& {\small $\begin{array}{c} \text{canonical action} \\
S_\Gamma[B, A'', t] \;\; (\S\ref{sec: Lag for Hitchin})\end{array}$}\\

  {\tikz[baseline=(X.base)]
    \node[draw, thick, rounded corners, align=center, minimum width=2.5cm, minimum height=1cm] (X)
    {$\mu^{-1}(0)/\mathcal{G}$};}
  \arrow[d, shorten <=-1pt, shorten >=-1pt, "\text{Solve }"' {pos=0.25}, "\; \mu\,=\,0 \; \text{and}" {pos=0.25}, "\text{move to }"' {pos=0.67}, "\;A''\,=\,0" {pos=0.65}] & \makebox[6cm][l]{\small $\begin{array}{ll} \mathcal P\phantom{''} \quad  \text{smooth principal } G\text{-bundle} \\
A'' \quad  \text{holomorphic structure}\\
\mathcal A\phantom{''} \quad \text{Lagrange multiplier for $\mu = 0$}\\
B\phantom{''} \quad  \text{smooth Higgs field} \end{array} \bigg\} \; A$}
& {\small $\begin{array}{c} \text{3$d$ mixed BF action} \\
S_{3d, \Gamma}[B, A, t] \;\; (\S\ref{sec: Lag for Hitchin mod G}) \end{array}$}\\

  {\tikz[baseline=(X.base)]
    \node[draw, thick, rounded corners, align=center, minimum width=2.5cm, minimum height=1cm] (X)
    {$T^* \mathrm{Bun}_G(C)$};} & \makebox[6cm][l]{\small $\begin{array}{ll} \mathcal P_{\rm hol} \;\;  \text{holomorphic principal } G\text{-bundle} \\
L \quad\;\;  \text{holomorphic Higgs field}/\text{Lax matrix} \end{array}$}
& {\small $\begin{array}{c} \text{1$d$ unifying action} \\
S_{{\rm H},\Gamma}[L^0, \gamma, t] \;\; (\S\ref{sec:unifyingmultiform}) \end{array}$}
\end{tikzcd}
\begin{tikzpicture}[remember picture, overlay]
  \node at (0.2,5.1) {1)};
  \node at (0.2,2.8) {2)};
  \node at (0.2,0.7) {3)};
  \draw[black!50] (4.65,0) -- (4.65,6.8);
  \draw[black!50] (4.75,0) -- (4.75,6.8);
  \draw[black!50] (11.8,0) -- (11.8,6.8);
  \draw[black!50] (11.9,0) -- (11.9,6.8);
  \draw[black!50] (0,6.05) -- (16.35,6.05);
  \draw[black!50] (4.65,4.1) -- (16.35,4.1);
  \draw[black!50] (4.65,1.55) -- (16.35,1.55);
\end{tikzpicture}
\caption{The three levels of Hitchin's integrable system and their corresponding actions.}
\label{fig: paper schematic}
\end{figure}

Finally, in \S\ref{sec: examples} we explicitly compute our unifying Lagrangian $1$-form $\Lag_{\rm H}$, or its multiform action $S_{{\rm H},\Gamma}[L^0, \gamma, (\varphi_\alpha), t]$ from \S\ref{sec:unifyingmultiform}, in special cases. Specifically, we exploit the invariance of the action under changes of holomorphic local trivialisations to fix particularly nice representatives for the transition function $\gamma$. In the $g=0$ case, our Lagrangian $1$-form recovers that of the rational Gaudin hierarchy first obtained in \cite{CDS} by a completely different method. In the $g=1$ case, we obtain a novel Lagrangian $1$-form for the elliptic Gaudin hierarchy -- and the elliptic spin Calogero-Moser hierarchy as a special case -- thus filling a gap in the landscape of Lagrangian multiforms. 

We also include a technical Appendix \ref{sec: Stability} in which we compare, in the $G = SL_m(\CC)$ case, the differential notion of the stability condition we use in the main text to the more familiar algebraic notion of the stability condition used in the literature.

\paragraph{Acknowledgements.}
AAS is funded by the School of Mathematics EPSRC Doctoral Training Partnership Studentship (Project Reference Number 2704447). BV would like to thank Sylvain Lacroix and Graeme Wilkin for useful discussions. BV gratefully acknowledges the support of the Leverhulme Trust through a Leverhulme Research Project Grant (RPG-2021-154).

\section{Gauged univariational principle and symmetry reduction} \label{sec: univar finite-dim}

We recall in \S\ref{sec: univar principle CH} the phase-space Lagrangian $1$-form introduced in \cite{CH}. In particular, in \S\ref{sec: univar principle} we formulate the univariational principle of \cite{CH}.
In \S\ref{sec: group actions} we explain how to incorporate a symmetry, represented by the free action of a connected Lie group $G$ on $M$, and recast the known relation between Noether charges and the moment map $\mu : T^\ast M \to \g^\ast$ in the context of Lagrangian $1$-forms. Then, in \S\ref{gauged_univariational_principle} we introduce the gauged univariational principle and explain how it describes, purely in the variational language of Lagrangian $1$-forms, the symplectic reduction procedure to $\mu^{-1}(0)/G$ traditionally presented in Hamiltonian terms.

\subsection{Lagrangian 1-forms, univariational principle and symmetry} \label{sec: univar principle CH}

\subsubsection{Univariational principle for Lagrangian $1$-forms} \label{sec: univar principle}

Let $M$ be an $m$-dimensional manifold with coordinates $q^\mu$ for $\mu = 1, \ldots, m$. The cotangent bundle $T^\ast M$ is parametrised by coordinates $(q^\mu, p_\mu)$ where $p_\mu$ for $\mu = 1, \ldots, m$ are dual coordinates along the fibres. The tautological $1$-form $\alpha$ and symplectic form $\omega$ of $T^\ast M$ are given by $\alpha=p_\mu \d q^\mu$ and $\omega = \d p_\mu\wedge \d q^\mu$. Here and in what follows, we always use the summation convention according to which repeated upstairs and downstairs indices of any kind are summed over.
Given any function $f$ on $T^\ast M$, the associated Hamiltonian vector field $\mathcal{X}_f$ is defined by the property $\mathcal{X}_f \lrcorner \omega = \d f$ and is given explicitly by
\begin{equation} \label{Ham vec def}
\mathcal{X}_f = \frac{\partial f}{\partial q^\mu} \frac{\partial}{\partial p_\mu} - \frac{\partial f}{\partial p_\mu} \frac{\partial}{\partial q^\mu} \,.
\end{equation}
For any pair of functions $f$ and $g$, their Poisson bracket is then given by $\{ f, g \} = - \mathcal{X}_f g = \mathcal{X}_g f$.

Let us introduce the phase-space Lagrangian $1$-form on $T^\ast M\times\RR^n$ given by
\begin{equation}\label{L}
\Lag \coloneqq \alpha - H_i \d t^i = p_\mu \d q^\mu - H_i \d t^i,
\end{equation}
where $H_1,\ldots, H_n$ are $n \in \ZZ_{\geq 1}$ real functions on $T^\ast M$ and $t^i$ are Cartesian coordinates on $\RR^n$. (We do not necessarily assume here any relation between the integers $n$ and $m = \dim M$; the case $n=1$ corresponds to an ordinary Hamiltonian system while the case $n = m$ will correspond to that of a Liouville integrable system.) Note that in \eqref{L}, the phase space variables $p_\mu$, $q^\mu$ and the time variables $t^i$ are treated on an equal footing. Accordingly, the action associated to a parametrised curve $\gamma:(0,1)\to T^\ast M\times \RR^n$ is now
\begin{equation} \label{pre ungauged action}
S_0[\gamma] = \int_0^1\gamma^\ast \Lag=\int_0^1\left(p_\mu \frac{\d q^\mu}{\d s} - H_i(p,q) \frac{\d t^i}{\d s}\right)\d s\,.
\end{equation}

If we were to apply the usual principle of least action, or variational principle, to \eqref{pre ungauged action} then we would seek a curve $\gamma:(0,1)\to T^\ast M\times \RR^n$ which is a critical point of this action.
That is to say, under an arbitrary variation $\delta \gamma(s)=(\delta p_\mu(s), \delta q^\mu(s), \delta t^i(s))$ satisfying the boundary conditions $0=\lim_{s\to0,1}\delta q(s)=\lim_{s\to0,1}\delta t(s)$ we want
\begin{align}
0 = \delta S_0[\gamma]&=\int_0^1\left(\delta p_\mu \frac{\d q^\mu}{\d s} + p_\mu \frac{\d \delta q^\mu}{\d s}- \left( \frac{\partial H_i}{\partial p_\mu} \delta p_\mu+ \frac{\partial H_i}{\partial q^\mu} \delta q^\mu \right)\frac{\d t^i}{\d s}-H_i \frac{\d \delta t^i}{\d s}\right) \d s\nonumber\\
& =\int_0^1\left(\delta p_\mu \left(\frac{\d q^\mu}{\d s}-\frac{\partial H_i}{\partial p_\mu}\frac{\d t^i}{\d s} \right) - \delta q^\mu\left(\frac{\d p_\mu}{\d s}+ \frac{\partial H_i}{\partial q^\mu}\frac{\d t^i}{\d s}\right) + \frac{\d H_i}{\d s}\delta t^i\right) \d s\,,
\end{align}
where in the second line we integrated by parts and used the boundary conditions. This leads to the Euler-Lagrange equations
\begin{equation} \label{EL equations ungauged}
\gamma'\lrcorner \d \Lag = 0 \quad \Longleftrightarrow \quad \frac{\d q^\mu}{\d s}=\frac{\partial H_i}{\partial p_\mu} \frac{\d t^i}{\d s}\,,~~
\frac{\d p_\mu}{\d s}=-\frac{\partial H_i}{\partial q^\mu} \frac{\d t^i}{\d s}\,,~~\frac{\d H_i}{\d s}=0
\end{equation}
where $\gamma'$ denotes the tangent vector to the curve $\gamma$.
It is clear from the first two Euler-Lagrange equations that locally along the curve we must always have $\frac{\d t^i}{\d s} \neq 0$ for some $i \in \{ 1, \ldots, n \}$ since otherwise we would have $\gamma'=0$. We could then look for solutions of \eqref{EL equations ungauged} such that $\frac{\d t^i}{\d s} \neq 0$ for a fixed $i$ which up to a change of variable amounts to working with the parameter $s = t^i$ along the curve. However, this would inevitably single out 
$H_i$ as our `preferred' Hamiltonian.

Rather than looking for individual curves $\gamma : (0,1) \to T^\ast M \times \RR^n$ which are critical points of the action $S_0[\gamma]$ in \eqref{pre ungauged action}, the univariational principle takes on board the fundamental idea of the multiform framework and replaces the curve $\gamma$ by an {\it immersion}
\begin{equation} \label{map Sigma}
\Sigma : \RR^n \longrightarrow T^\ast M \times \RR^n \,, \quad (u^j) \longmapsto \big( p_\mu(u), q^\mu(u), t^i(u) \big)
\end{equation}
\ie a map $\Sigma$ such that its tangent map 
\begin{equation}
    \d_u\Sigma:T_u\RR^n\to T_{\Sigma(u)}(T^\ast M \times \RR^n)
\end{equation}
is injective for every $u\in\RR^n$. In particular, $\Sigma$ is then a (local) diffeomorphism onto its image in $T^\ast M \times \RR^n$.
In view of applying the univariational principle, $\Sigma$ is required to satisfy boundary conditions that fix the functions $\lim_{\lVert u \rVert\to\infty}q^\mu(u)$ and $\lim_{\lVert u \rVert\to\infty}t^i(u)$ of $S^{n-1}$ at spatial infinity.

Introducing also a curve 
\begin{equation} \label{Gamma def}
\Gamma : (0,1) \longrightarrow \RR^n \,, \quad s \longmapsto \big( t^i(s) \big)
\end{equation}
such that $\lim_{s\to0,1} \lVert \Gamma(s) \rVert =\infty$, we can evaluate the action \eqref{pre ungauged action} on the composite map
\begin{equation*}
\begin{tikzcd}
\gamma \coloneqq \Sigma \circ \Gamma : (0,1) \arrow[r,"\Gamma"] & \RR^n \arrow[r,"\Sigma"] & T^\ast M \times \RR^n\,,
\end{tikzcd}
\end{equation*}
leading now to a family of multiform actions labelled by $\Gamma$, namely
\begin{equation} \label{ungauged action}
S_\Gamma[\Sigma] \coloneqq S_0[\Sigma \circ \Gamma] = \int_0^1\left(p_\mu \frac{\partial q^\mu}{\partial u^j} - H_i(p,q) \frac{\partial t^i}{\partial u^j}\right) \frac{\d u^j}{\d s} \d s\,.
\end{equation}
Crucially, we view this family of actions as depending only on the map $\Sigma$ in \eqref{map Sigma} and we treat the curve $\Gamma$ in \eqref{Gamma def} as a parameter labelling the family. The distinction between the roles of $\Sigma$ and $\Gamma$, reflected in the notation $S_\Gamma[\Sigma]$, underpins the \emph{univariational principle} which can now be formulated as seeking an immersion $\Sigma$ which is simultaneously a critical point of the family of actions $S_\Gamma[\Sigma]$ for all curves $\Gamma$ in $\RR^n$. In other words, it consists in the single-step procedure: 
\begin{center}
Find $\Sigma$ such that $\delta_\Sigma S_\Gamma[\Sigma] = 0$ for all curves $\Gamma$ in $\RR^n$.    
\end{center}
As shown in \cite{CH}, the univariational principle implies that the map $\Sigma$ in \eqref{map Sigma} can be written as $\Sigma = (f_\Sigma, \id_{\RR^n})$ for some map $f_\Sigma : \RR^n \to T^\ast M$, $(t^i)\mapsto \big( p_\mu(t), q^\mu(t)\big)$ and the univariational principle is equivalent to the system of equations
\begin{subequations} \label{univar EL eq}
\begin{align}
\label{univar EL eq a} &\frac{\partial q^\mu}{\partial t^i}=\frac{\partial H_i}{\partial p_\mu} \,,~~
\frac{\partial p_\mu}{\partial t^i}=-\frac{\partial H_i}{\partial q^\mu}\,,~~i=1,\dots,n \,,~~\mu = 1, \ldots, m \,,\\
\label{univar EL eq b} &\frac{\partial H_i}{\partial p_\mu}\frac{\partial H_j}{\partial q^\mu}-\frac{\partial H_i}{\partial q^\mu}\frac{\partial H_j}{\partial p_\mu}=\{H_i,H_j\}=0\,,~~i,j=1,\dots,n\,.    
\end{align}
\end{subequations}
We refer to the proof of Theorem \ref{thm: gauged univar} below for details in the more general gauge invariant setting.
We see from \eqref{univar EL eq b} that the univariational principle already encodes the closure relation through the variation of $\Sigma$, as a result of having included the times $t^i$ among the dynamical variables in the map $\Sigma$.
In particular, the univariational principle admits solutions if and only if the Hamiltonian functions $H_1, \ldots, H_n$ mutually Poisson commute. We also see from \eqref{univar EL eq a} that the time coordinate $t^i$ parametrises the Hamiltonian flow of $H_i$ for each $i=1,\ldots, n$.
The univariational principle applied to the Lagrangian $1$-form \eqref{L} therefore encodes multitime Hamiltonian mechanics, which includes the case of Liouville integrable systems when $n = m$.
\begin{remark}
Readers familiar with Lagrangian multiforms may wonder how the univariational principle adopted here connects to the bivariational principle that has been employed in the literature until the recent work \cite{CH}. In the bivariational formulation, one views $\Sigma$ as $(f_\Sigma, \id_{\RR^n})$ for some map $f_\Sigma : \RR^n \to T^\ast M$, $(t^i)\mapsto \big( p_\mu(t), q^\mu(t)\big)$ {\it from the outset} so that $S_0[\Sigma \circ \Gamma]$ in \eqref{ungauged action} becomes the (perhaps more familiar) action 
 \begin{equation}
S[f_\Sigma , \Gamma]\coloneqq S_0[f_\Sigma \circ \Gamma] = \int_0^1\left(p_\mu \frac{\partial q^\mu}{\partial t^j} - H_j(p,q) \right) \frac{\d t^j}{\d s} \d s\,.
\end{equation}
The price to pay in this formulation is that one has to vary with respect to both $f_\Sigma$ {\it and} $\Gamma$ to obtain the full set of equations \eqref{univar EL eq}.

We also note that the univariational principle is a natural generalisation to the multitime case of the well-known {\it extended phase space} approach to classical mechanics. That approach corresponds to having a single time $t$ in our case, \ie $n=1$, so that $\Gamma:(0,1)\to \RR$, $s\mapsto t$ is a (re)parametrisation of time. This idea applied in the Lagrangian formalism can be found for instance in \cite{Sou} but is likely much older. 
\end{remark}

\subsubsection{Group actions and symmetry} \label{sec: group actions}

Suppose now that $M$ admits a free right action of a connected Lie group $G$, $\rho : G \times M \to M$.
Let $X_a$ be a basis of the Lie algebra $\g$ of $G$ and let $X_a^\mu(q)\parder{}{q^\mu}$ denote the corresponding fundamental vector fields generating the action of $\g$ on $M$. The right action of $G$ lifts to $T^\ast M$, which we also denote by $\rho : G \times T^\ast M \to T^\ast M$, and the corresponding action of $\g$ is generated by vector fields
\begin{equation}\label{action coords}
    X_a^\sharp = X_a^\mu(q)\parder{}{q^\mu} - p_\nu \parder{X_a^\nu}{q^\mu} \parder{}{p_\mu}
\end{equation}
satisfying $[X_a^\sharp,X_b^\sharp]=f_{ab}{}^c X_c^\sharp$, where $f_{ab}{}^c$ denote the structure constants of $\g$.
We further lift this to a right action $\rho : G \times T^\ast M \times \RR^n \to T^\ast M \times \RR^n$ of $G$ on $T^\ast M \times \RR^n$ by letting $G$ act trivially on $\RR^n$ so that the infinitesimal action of $\g$ is still generated by the same vector fields \eqref{action coords}.

In what follows we will consider only the corresponding left action of the group $G$ on $T^\ast M \times \RR^n$, given by $G \times T^\ast M \times \RR^n \to T^\ast M \times \RR^n$, $(g, x) \mapsto g \cdot x \coloneqq \rho_{g^{-1}}(x)$ and which is infinitesimally generated by the vector fields $\delta_{X_a}\coloneqq - X_a^\sharp$ for $a = 1, \ldots, \dim \g$.

\begin{proposition}\label{prop_global_symmetry}
The action \eqref{ungauged action} is invariant under the infinitesimal action of $G$ on $T^\ast M \times \RR^n$ generated by \eqref{action coords} if and only if each $H_i$ is invariant under this infinitesimal group action, i.e.
\begin{equation} \label{Hi invariant}
\mathcal{L}_{X^\sharp_a}H_i=0\,\,,~~i=1,\dots,n\,,~~a=1,\dots, {\rm dim}~\g\,.
\end{equation}
The Noether charges associated with this global $G$ symmetry are then given by $\mu_a(p,q) = - p_\nu X_a^\nu(q)$.
\end{proposition}
\begin{proof}
The variation of any given map $\Sigma : \RR^n \to T^\ast M\times \RR^n$, $(u^j) \mapsto \big( p_\mu(u), q^\mu(u), t^i(u) \big)$ under the infinitesimal left action of $G$ on $T^\ast M \times \RR^n$ generated by the vector fields \eqref{action coords} is given by $\delta_X \Sigma(u) = - X^\sharp \big( \Sigma(u) \big)$, or more explicitly in components
\begin{subequations} \label{inf var p q t}
\begin{align}
    \delta_X q^\mu(u) &= - \lambda^a X_a^\mu(q(u))\\
    \delta_X p_\mu(u) &= \lambda^a p_\nu(u)\parder{X_a^\nu}{q^\mu}(q(u))\\
    \delta_X t^i(u) &= 0
\end{align}
\end{subequations}
for arbitrary $X=\lambda^a X_a \in\g$. Noting that
\begin{equation} \label{delta p dq}
\delta_X \bigg( p_\mu \frac{\partial q^\mu}{\partial u^j} \bigg) = \lambda^a p_\nu \parder{X_a^\nu}{q^\mu} \frac{\partial q^\mu}{\partial u^j} - \lambda^a p_\nu \frac{\partial X_a^\nu}{\partial u^j} = 0\,,
\end{equation}
where the last step is by the chain rule for the function $X_a^\nu(q(u))$, the corresponding variation of the action $S_\Gamma[\Sigma]$ in \eqref{ungauged action} is then
\begin{align} \label{inf_invariance_S}
\delta_X S_\Gamma[\Sigma] = \lambda^a\int_0^1 \bigg(X_a^\mu\frac{\partial H_i}{\partial q^\mu}-p_\nu\parder{X_a^\nu}{q^\mu}\frac{\partial H_i}{\partial p_\mu}\bigg) \frac{\d t^i}{\d s} \d s
= \lambda^a\int_0^1 \big( \mathcal L_{X^\sharp_a} H_i \big) \frac{\d t^i}{\d s} \d s\,,
\end{align}
where the second step is by definition \eqref{action coords} of $X^\sharp_a$. Now the resulting expression in \eqref{inf_invariance_S} should be zero for any $\lambda^a$, any map $\Sigma$ and any curve $\Gamma$ from which \eqref{Hi invariant} follows.

The expression $\mu_a = - p_\nu X^\nu_a$ for the Noether charge associated with this global symmetry can be obtained by using the standard trick of promoting the constant parameters $\lambda^a \in \CC$ to functions $\lambda^a : \RR^n \to \CC$. This leads to an additional term $\frac{\partial \lambda^a}{\partial u^j} \mu_a$ on the right-hand side of \eqref{delta p dq} so that the variation of the action now reads $\delta_X S_\Gamma[\Sigma] = \int_0^1 \frac{\d \lambda^a}{\d s} \mu_a \d s$, as required.
\end{proof}

In the Hamiltonian formalism, the Noether charges $\mu_a$ from Proposition \ref{prop_global_symmetry} are encoded in a moment map $\mu : T^\ast M \to \g^\ast$ such that
\begin{equation}\label{moment_map}
\langle\mu(p,q),X_a\rangle=\mu_a(p,q)=-p_\nu X^\nu_a(q),
\end{equation}
where $\langle~,~\rangle : \g^\ast \times \g \to \CC$ denotes the canonical pairing.  In general, a moment map is required to satisfy the equations
\begin{equation} \label{moment map eqs}
\langle\d \mu, X_a\rangle = X_a^\sharp \lrcorner \omega \quad\text{and}\quad \mathcal{L}_{X_a^\sharp} \mu + \mathrm{ad}^\ast_{X_a} \mu = 0.
\end{equation}
In our case, the first of these follows by a direct calculation from \eqref{action coords}, \eqref{moment_map} and $\omega=\d p_\mu\wedge \d q^\mu$.  The second is proved as follows.  First we note that $\langle\mu,X_b\rangle=-X_b^\sharp\lrcorner\alpha$ and that $\mathcal{L}_{X_a^\sharp}\alpha=0$.  So
\begin{equation}
\langle\mathcal{L}_{X_a^\sharp}\mu,X_b\rangle=-\mathcal{L}_{X_a^\sharp}(X_b^\sharp\lrcorner\alpha)=-[X_a^\sharp,X_b^\sharp]\lrcorner\alpha = \langle\mu,[X_a,X_b]\rangle=-\langle\mathrm{ad}^\ast_{X_a} \mu,X_b\rangle\,,
\end{equation}
where in the second last step we used the fact that $[X_a^\sharp,X_b^\sharp] = [X_a,X_b]^\sharp$. Hence the second equation of \eqref{moment map eqs} holds.

\subsection{Gauged Lagrangian 1-form and gauged univariational principle}\label{gauged_univariational_principle}

When a physical system is invariant under the action of a Lie group $G$, such as in the context of \S\ref{sec: group actions}, this symmetry can be used to reduce the number of degrees of freedom. Indeed, since the components of the moment map $\mu : T^\ast M \to \g^\ast$ are preserved under all the Hamiltonian flows $\mathcal X_{H_i}$, i.e. $\mathcal X_{H_i} \mu_a = 0$, one may consistently restrict to the zero-level set $\mu^{-1}(0) \subset T^\ast M$ of this moment map by imposing the constraint $\mu(p,q) = 0$. Furthermore, by the $G$-equivariance property of the moment map, the zero-level set $\mu^{-1}(0)$ is invariant under the action of $G$ so that we may further reduce the number of degrees of freedom by working on the quotient space $\mu^{-1}(0)/G$.

In order to implement the above symplectic reduction procedure in the variational setting, and thereby construct a Lagrangian $1$-form and action for the reduced system, we can impose the constraint $\mu(p,q) = 0$ in $\g^\ast$ by introducing a $\g$-valued Lagrange multiplier.
Specifically, since we want to gauge the action of $G$ on the map $\Sigma : \RR^n \to T^\ast M \times \RR^n$ described infinitesimally in \eqref{inf var p q t},
we introduce a $\g$-valued Lagrange multiplier $1$-form on $\RR^n$ which we denote by
\begin{equation} \label{cal A def}
\mathcal A =\mathcal A_j \d u^j = \mathcal A^a_j X_a \d u^j \in \Omega^1(\RR^n, \g) \,.
\end{equation}

Using the canonical pairing $\langle~,~\rangle : \g^\ast \times \g \to \CC$ we can combine \eqref{cal A def} with the moment map $\mu : T^\ast M \to \g^\ast$, which we view as a map $\mu : T^\ast M \times \RR^n \to \g^\ast$ that is constant along $\RR^n$, to obtain a $1$-form $\langle\Sigma^\ast \mu, \mathcal A\rangle \in \Omega^1(\RR^n)$. The gauging procedure then simply consists in adding this term to the pullback $\Sigma^\ast \Lag \in \Omega^1(\RR^n)$ of the Lagrangian $1$-form $\Lag$. We obtain a family of actions for a map $\Sigma : \RR^n \to T^\ast M \times \RR^n$ and a $1$-form $\mathcal A \in \Omega^1(\RR^n, \g)$, parametrised by curves $\Gamma : (0,1) \to \RR^n$,
\begin{equation}\label{uni_gauged action}
S_\Gamma[\Sigma, {\cal A}] = \int_0^1 \Gamma^\ast \Big( \Sigma^\ast \Lag + \langle\Sigma^\ast \mu, {\cal A}\rangle \Big) \,.
\end{equation}
The constraint $\mu(p,q) = 0$ will now be enforced dynamically through the equations of motion for $\mathcal A$, see \S\ref{sec: gauged univar princ} below. Moreover, the effect of adding this new term to the Lagrangian is to promote the $G$-symmetry we started with, from \S\ref{sec: group actions}, to a gauge symmetry as we now show.

\subsubsection{Gauging a symmetry in a Lagrangian $1$-form} \label{sec: gauging symmetry}

We consider local gauge transformations parametrised by smooth maps $g: \RR^n \to G$, which act pointwise on the map $\Sigma : \RR^n \to T^\ast M \times \RR^n$ and on the Lagrange multiplier $\mathcal A \in \Omega^1(\RR^n, \g)$ as a gauge transformation, i.e.
\begin{subequations} \label{gauge_transfo_A}
\begin{align}
\label{gauge_transfo_A a} \Sigma(u) &\longmapsto g(u) \cdot \Sigma(u) \,,\\
\label{gauge_transfo_A b} \mathcal A &\longmapsto g \,\mathcal A \,g^{-1}- \d_{\RR^n} g \,g^{-1} \,,
\end{align}
\end{subequations}
where $\d_{\RR^n}$ denotes the de Rham differential on $\RR^n$.
Due to its transformation property \eqref{gauge_transfo_A b}, we will henceforth often refer to the Lagrange multiplier $1$-form $\mathcal A \in \Omega^1(\RR^n, \g)$ as a gauge field.

\begin{proposition}\label{prop_local_symmetry}
The action $S_\Gamma[\Sigma,\mathcal A]$ is invariant under an infinitesimal version of the $G$-valued local gauge transformation in \eqref{gauge_transfo_A} if and only if the functions $H_i$ are $G$-invariant, i.e.
\begin{equation}\label{moment map eqs 2}
\mathcal{L}_{X^\sharp_a}H_i=0\,\,,~~i=1,\dots,n\,,~~a=1,\dots, {\rm dim}~\g\,.
\end{equation}
\end{proposition}
\begin{proof}
Let $\Sigma:\RR^n \to T^\ast M\times \RR^n$, $(u^j) \mapsto \big( p_\mu(u), q^\mu(u), t^i(u) \big)$ be a given map as in the proof of Proposition \ref{prop_global_symmetry}. The action \eqref{uni_gauged action} for this $\Sigma$ and any curve $\Gamma : (0,1) \to \RR^n$ then explicitly reads
\begin{equation} \label{uni_gauged action explicit}
S_\Gamma[\Sigma,\mathcal A] = \int_0^1 \bigg(p_\mu \frac{\partial q^\mu}{\partial u^j} -H_i(p,q)\frac{\partial t^i}{\partial u^j} + \mu_a(p,q) \mathcal A^a_j \bigg) \frac{\d u^j}{\d s} \d s\,,
\end{equation}
where $p_\mu$, $q^\mu$, $t^i$ and $A^a_j$ all depend on $s$ through their dependence on the functions $u^j(s)$.
Consider an infinitesimal gauge transformation parametrised by a $\g$-valued function $X:\RR^n \to\g$, which we write in components as $X = \lambda^a(u) X_a$. The variations of the components of $\Sigma$ and $\mathcal A$ under the infinitesimal left action of $G$ on $T^\ast M \times \RR^n$ then read (cf. the variations \eqref{inf var p q t} where $\lambda^a$ was constant)
\begin{subequations}
\begin{align}
    \delta_X q^\mu(u) &= - \lambda^a(u)X_a^\mu\big(q(u)\big)\\
    \delta_X p_\mu(u) &= \lambda^a(u) p_\nu(u)\parder{X_a^\nu}{q^\mu}\big(q(u)\big)\\
    \delta_X t^i(u) &= 0 \\
    \delta_X \mathcal A^a_j(u) &= - \frac{\partial \lambda^a}{\partial u^j} - f_{bc}{}^a \mathcal A^b_j (u) \lambda^c(u)\,.
\end{align}
\end{subequations}
By the exact same computation as in the proof of Proposition \ref{prop_global_symmetry}, where in the end of that proof $\lambda^a$ was already treated as a function $\lambda^a : \RR^n \to \CC$, the corresponding variation of $S_\Gamma[\Sigma, \mathcal A]$ reads
\begin{equation} \label{variaton action gauged}
\delta_X S_\Gamma[\Sigma, \mathcal A] = \int_0^1\left[ \lambda^a \big( \mathcal L_{X^\sharp_a} H_i \big) \frac{\partial t^i}{\partial u^j} - \lambda^a \mathcal A^b_j \Big(\mathcal L_{X^\sharp_a}\mu_b -f_{ab}{}^c\mu_c\Big) \right] \frac{\d u^j}{\d s} \d s\,.
\end{equation}
We note, in particular, that the term $- \mu_a \frac{\partial \lambda^a}{\partial u^j}$ appearing in the variation $\mu_a \delta_X A^a_j$ cancels with the term $- p_\nu X^\nu_a \frac{\partial \lambda^a}{\partial u^j}$ from the variation \eqref{delta p dq} when $\lambda^a$ is not constant.
The right-hand side of \eqref{variaton action gauged} needs to vanish for all functions $\lambda^a$, $\mathcal A^a_i$ and all maps $\Sigma$ and curves $\Gamma$, but the last bracketed term vanishes by the second relation in \eqref{moment map eqs} since
\begin{equation}
    \langle\mathcal{L}_{X^\sharp_a}\mu, X_b\rangle+ \langle\mathrm{ad}^\ast_{X_a}\mu, X_b\rangle = \mathcal{L}_{X^\sharp_a}\mu_b - \langle\mu, [X_a,X_b]\rangle = \mathcal{L}_{X^\sharp_a}\mu_b -f_{ab}{}^c\mu_c \,.
\end{equation}
Therefore gauge invariance is equivalent to the condition \eqref{moment map eqs 2}.
\end{proof}

By combining Propositions \ref{prop_global_symmetry} and \ref{prop_local_symmetry} we see that the gauged action $S_\Gamma[\Sigma, \mathcal A]$ in \eqref{uni_gauged action} is invariant under the gauge group $C^\infty(\RR^n, G)$ if and only if the original action $S_\Gamma[\Sigma]$ in \eqref{ungauged action} is invariant under the Lie group $G$.

\subsubsection{Gauged univariational principle} \label{sec: gauged univar princ}

Recall from \S\ref{sec: univar principle} that the univariational principle seeks an immersion $\Sigma : \RR^n \to T^\ast M\times \RR^n$ such that for all curves $\Gamma : (0,1) \to \RR^n$ we have $\delta_\Sigma S_\Gamma[\Sigma] = 0$.
We introduce the gauged version in the following definition.  To do so we need to introduce the gauge-covariant derivative of $\Sigma$.  This is given by the formula
\begin{equation}\label{eq:covariant derivative Sigma}
    D^{\A}_u\Sigma=\d_u\Sigma-\A^\sharp_u\,:\,T_u\RR^n\to T_{\Sigma(u)}(T^\ast M\times \RR^n),
\end{equation}
in which $\d_u\Sigma$ is the usual differential of $\Sigma$ and $\A^\sharp_u$ is defined by
\begin{equation}
\A^\sharp_u :T_u\RR^n\to T_{\Sigma(u)}(T^\ast M\times \RR^n), \quad \parder{}{u^i}\mapsto \A_i^a(u)X_a^\sharp(\Sigma(u))\,.
\end{equation}
where we recall that $X_a^\sharp$ is given in \eqref{action coords}.  Put differently, $\A^\sharp_u$ is the composition of $\A:T_u\RR^n\to \g$, the map $\g\to \Gamma(T(T^\ast M\times \RR^n))$ induced by the action of $G$ on $T^\ast M\times \RR^n$, and the map $\Gamma(T(T^\ast M\times \RR^n))\to T_{\Sigma(u)}(T^\ast M\times \RR^n)$ given by evaluation at $\Sigma(u)$.  The derivative \eqref{eq:covariant derivative Sigma} is gauge-covariant in the sense that
\begin{equation}
    D_u^{g\A g^{-1}- \d gg^{-1}}(g\cdot \Sigma)= \d_{\Sigma(u)}g(u) \circ D_u^\A\Sigma,
\end{equation}
where $\d_{\Sigma(u)}g(u):T_{\Sigma(u)}(T^\ast M\times\RR^n)\to T_{g(u)\cdot \Sigma(u)}(T^\ast M\times\RR^n)$ is the differential of $g(u):T^\ast M\times\RR^n\to T^\ast M\times\RR^n$.
\begin{definition}
The \emph{gauged univariational principle} seeks a map
\begin{equation*}
\Sigma : \RR^n \longrightarrow T^\ast M\times \RR^n
\end{equation*}
and a gauge field ${\cal A} \in \Omega^1(\RR^n, \g)$ such that the linear map $ D^{\A}_u\Sigma=\d_u\Sigma-\A^\sharp_u$ is injective for every $u \in \RR^n$,
and such that
the pair $(\Sigma, {\cal A})$ is simultaneously a critical point of the family of actions $S_\Gamma[\Sigma, {\cal A}]$ for all curves $\Gamma : (0,1) \to \RR^n$, namely such that
\begin{equation*}
\delta_\Sigma S_\Gamma[\Sigma, \mathcal A] = 0 \quad \text{and} \quad \delta_{\mathcal A} S_\Gamma[\Sigma, \mathcal A] = 0 \,.
\end{equation*}
\end{definition}
\begin{remark}
The condition that $\d_u\Sigma-\A^\sharp_u$ is injective is the natural gauge-covariant analogue of the requirement that the map $\Sigma$ in \eqref{map Sigma} is an immersion.
\end{remark}

\begin{theorem} \label{thm: gauged univar}
The gauged univariational principle applied to the gauge invariant action $S_\Gamma[\Sigma, \mathcal A]$ in \eqref{uni_gauged action} gives rise to the following set of equations
\begin{subequations} \label{EL gauged univar thm}
\begin{align} \label{A variation}
\mu(p,q) &= 0\,,\\
\label{univariational eqs1}
\parder{q^\mu}{t^i} - \frac{\partial H_i}{\partial p_\mu} &= \widetilde{\mathcal A}^a_i X_a^\mu\,, \\
\label{univariational eqs2}
\parder{p_\mu}{t^i}+ \frac{\partial H_i}{\partial q^\mu} &=- \widetilde{\mathcal A}^a_i p_\nu \parder{X_a^\nu}{q^\mu}\,, \\
\label{univariational eqs3} \{H_i,H_j\} &= 0 \,,
\end{align}
\end{subequations}
where the composition $\pi_{\RR^n} \circ \Sigma : \RR^n \to \RR^n$, $(u^j) \mapsto \big( t^i(u) \big)$ with the projection $\pi_{\RR^n} : T^\ast M \times \RR^n \to \RR^n$, $( p_\mu, q^\mu,t^i )\mapsto (t^i)$ to the second factor is a (local) diffeomorphism and is used to define $\widetilde{\mathcal A}^a_i = \frac{\partial u^j}{\partial t^i} \mathcal A^a_j$.

Moreover, the $\g$-valued connection $\d_{\RR^n} + \mathcal A$ is flat, i.e. it satisfies the zero-curvature equation $F = \d_{\RR^n}\mathcal A + \frac 12 [\mathcal A, \mathcal A] = 0$ or in components
\begin{equation}
\label{flatness}
    F_{ij}^a \coloneqq \parder{\widetilde{\mathcal A}_j^a}{t^i} - \parder{\widetilde{\mathcal A}_i^a}{t^j}+f_{bc}{}^a \widetilde{\mathcal A}_i^b \widetilde{\mathcal A}_j^c=0\,,~~a=1,\dots,{\rm dim}~\g\,,~~i,j=1,\dots,n\,.
\end{equation}
\end{theorem}
\begin{proof}
Let $\Sigma : \RR^n \to T^\ast M\times \RR^n$, $(u^j) \mapsto \big( p_\mu(u), q^\mu(u), t^i(u) \big)$ and $\Gamma : (0,1) \to \RR^n$ be arbitrary maps. The action \eqref{uni_gauged action} for these can be written as in \eqref{uni_gauged action explicit}.

We first perform a variation of the action \eqref{uni_gauged action explicit} with respect to the functions $\mathcal A^a_j : \RR^n \to \RR$
\begin{equation}
\delta_{\mathcal A} S_\Gamma[\Sigma, \mathcal A] = \int_0^1 \mu_a(p, q) \delta \mathcal A^a_j\big( u(s) \big) \frac{\d u^j}{\d s} \d s\,.
\end{equation}
In order for this to vanish for arbitrary variations $\delta \mathcal A^a_j$ we must have
\begin{equation} \label{EL constraint gauged univar}
\mu_a(p, q) = 0\,,
\end{equation}
which is the expected constraint equation $\mu = 0$.

Next, consider the variation of the action \eqref{uni_gauged action explicit} with respect to the map $\Sigma : \RR^n \to T^\ast M \times \RR^n$. Under an arbitrary variation $\delta \Sigma(u) = \big( \delta p_\mu(u), \delta q^\mu(u), \delta t^i(u) \big)$ satisfying the boundary conditions $0=\lim_{\lVert u \rVert\to\infty} \delta q^\mu(u)=\lim_{\lVert u \rVert\to\infty} \delta t^i(u)$, the action varies as
\begin{multline} \label{delta gamma gauged action}
\delta_\Sigma S_\Gamma[\Sigma, \mathcal A] =\int_0^1\frac{\d}{\d s}\bigg[p_\mu\delta q^\mu-H_i\delta t^i\bigg]\d s + \int_0^1\bigg[ \bigg( \frac{\partial q^\mu}{\partial u^j} - \frac{\partial t^i}{\partial u^j} \frac{\partial H_i}{\partial p_\mu} + \frac{\partial \mu_a}{\partial p_\mu} \mathcal A^a_j \bigg) \, \delta p_\mu \\
- \bigg( \frac{\partial p_\mu}{\partial u^j} + \frac{\partial t^i}{\partial u^j} \frac{\partial H_i}{\partial q^\mu} - \frac{\partial \mu_a}{\partial q^\mu} \mathcal A^a_j \bigg) \, \delta q^\mu 
+ \bigg( \frac{\partial H_i}{\partial p_\mu} \frac{\partial p_\mu}{\partial u^j} + \frac{\partial H_i}{\partial q^\mu} \frac{\partial q^\mu}{\partial u^j} \bigg) \, \delta t^i \bigg] \frac{\d u^j}{\d s} \d s.
\end{multline}
The total derivative term on the right-hand side of this equation vanishes due to the boundary conditions on $\delta\Sigma$ and the boundary condition $\lim_{s\to 0,1} \lVert \Gamma(s) \rVert = \infty$ on $\Gamma$. We now want the above variation to vanish for all $\delta \Sigma$ and all curves $\Gamma$, and hence for arbitrary $\frac{\d u^j}{\d s}$, which is equivalent to the system of equations
\begin{subequations} \label{eom gauged action}
\begin{align}
\label{eom gauged action a} \frac{\partial q^\mu}{\partial u^j} &= \frac{\partial t^i}{\partial u^j} \frac{\partial H_i}{\partial p_\mu} - \frac{\partial \mu_a}{\partial p_\mu} \mathcal A^a_j \,,\\
\label{eom gauged action b} \frac{\partial p_\mu}{\partial u^j} &= - \frac{\partial t^i}{\partial u^j} \frac{\partial H_i}{\partial q^\mu} + \frac{\partial \mu_a}{\partial q^\mu} \mathcal A^a_j \,,\\
\label{eom gauged action c} \frac{\partial H_i}{\partial p_\mu} \frac{\partial p_\mu}{\partial u^j} + \frac{\partial H_i}{\partial q^\mu} \frac{\partial q^\mu}{\partial u^j} &= 0 \,.
\end{align}
\end{subequations}

Next, we claim that these equations of motion imply that the Jacobian matrix $\big( \frac{\partial t^j}{\partial u^j} \big)$ must be invertible, thus ensuring that the composition $\pi_{\RR^n} \circ \Sigma : \RR^n \to \RR^n$ is a (local) diffeomorphism. To show this we adapt the similar argument \cite[\S2.2]{CH} to the present case with gauge field.  
The differential of $\Sigma$ at $u\in\RR^n$ reads
\begin{equation}
\label{push_forward}
\d_u\Sigma:T_u\RR^n\to T_{\Sigma(u)}(T^\ast M\times \RR^n), \quad \parder{}{u^i}\mapsto \parder{p_\mu}{u^i}\parder{}{p_\mu} + \parder{q^\mu}{u^i}\parder{}{q^\mu} + \parder{t^j}{u^i}\parder{}{t^j}\,,
\end{equation}
and therefore proving the claim is equivalent to showing that the tangent map of the composition $\pi_{\RR^n} \circ\Sigma:\RR^n\to\RR^n$, $u^j\mapsto t^i(u)$ at every $u \in \RR^n$, namely
\begin{equation}
\d_u (\pi_{\RR^n} \circ \Sigma)= \d_{\Sigma(u)} \pi_{\RR^n} \circ \d_u \Sigma:T_u \RR^n\to T_{\pi_{\RR^n}(\Sigma(u))}\RR^n,\quad \parder{}{u^i}\mapsto \parder{t^j}{u^i}\parder{}{t^j} \,,
\end{equation}
is injective when the univariational equations \eqref{eom gauged action a}-\eqref{eom gauged action c} hold. And to do this we will use the fact that the latter can equivalently be written as 
\begin{equation}
\label{compact_univ_eqs}
\d_u\Sigma\left(\parder{}{u^i}\right)\lrcorner \d\Lag -\A^\sharp_u\left(\parder{}{u^i}\right)\lrcorner\omega=0\quad\forall i,
\end{equation}
in which
\begin{equation}
\label{form_dL}
\d\Lag = \omega-\d H_i\wedge \d t^i = \d p_\mu\wedge \d q^\mu - \parder{H_i}{p_\mu}\d p_\mu\wedge \d t^i - \parder{H_i}{q^\mu}\d  q^\mu\wedge \d t^i.
\end{equation}

So let $V\in\ker \big( \d_{\Sigma(u)} \pi_{\RR^n} \circ \d_u \Sigma \big)$. On the one hand, \eqref{compact_univ_eqs} and \eqref{form_dL} imply
\begin{equation}
    0= (\d_u \Sigma-\A^\sharp_u)(V)\lrcorner\omega - [\d_u \Sigma(V)\lrcorner \d H_i]\d t^i+[\d_u \Sigma(V)\lrcorner \d t^i]\d H_i.
\end{equation}
On the other hand, $\d_u \Sigma(V)\lrcorner \d t^i=(\d_{\Sigma(u)} \pi_{\RR^n} \circ \d_u \Sigma)(V)\lrcorner \d u^i=0$ since $V\in\ker \big( \d_{\Sigma(u)} \pi_{\RR^n} \circ \d_u \Sigma \big)$.
Thus we are left with
\begin{equation}
(\d_u \Sigma-\A^\sharp_u)(V)\lrcorner\omega - [\d_u \Sigma(V)\lrcorner \d H_i]\d t^i=0
\end{equation}
and each term must individually vanish since the first belongs to $T^*(T^\ast M)$ and the second belongs to $T^*\RR^n$. In particular, since $\omega$ is nondegenerate, we must have $(\d_u \Sigma-\A^\sharp_u)(V)=0$, from which it follows that $V=0$ since $\d_u \Sigma-\A^\sharp_u$ is injective.
Note that since $\A^\sharp_u\left(\parder{}{u^i}\right)\lrcorner \d H_k=0$ (invariant Hamiltonians), \eqref{compact_univ_eqs} can be compactly written using the gauge-covariant derivative \eqref{eq:covariant derivative Sigma}
\begin{equation}
  D^{\A}_u\Sigma\left(\parder{}{u^i}\right)\lrcorner \d\Lag =0\quad\forall i\,.
\end{equation}
As a consequence, performing a suitable coordinate transformation on $\RR^n$, we may bring $\Sigma$ to the canonical form $(f_\Sigma, \id_{\RR^n})$ for some function $f_\Sigma : \RR^n \to T^\ast M$, $t^i \mapsto \big( p_\mu(t), q^\mu(t) \big)$.
We may then rewrite \eqref{eom gauged action} as
\begin{subequations} \label{EL eqs gauged univar}
\begin{align}
\label{EL eqs gauged univar a} \frac{\partial q^\mu}{\partial t^i} &= \frac{\partial H_i}{\partial p_\mu} - \frac{\partial \mu_a}{\partial p_\mu} \widetilde{\mathcal A}^a_i\,, \\
\label{EL eqs gauged univar b} \frac{\partial p_\mu}{\partial t^i} &= - \frac{\partial H_i}{\partial q^\mu} + \frac{\partial \mu_a}{\partial q^\mu} \widetilde{\mathcal A}^a_i\,, \\
\label{EL eqs gauged univar c} \frac{\partial H_i}{\partial p_\mu} \frac{\partial p_\mu}{\partial t^j} + \frac{\partial H_i}{\partial q^\mu} \frac{\partial q^\mu}{\partial t^j} &= 0\,,
\end{align}
\end{subequations}
where $\widetilde{\mathcal A}^a_i \coloneqq \frac{\partial u^j}{\partial t^i} \mathcal A^a_j$ is defined using the inverse $\big( \frac{\partial u^j}{\partial t^i} \big)$ of the matrix $\big( \frac{\partial t^i}{\partial u^j} \big)$.  Inserting the expression \eqref{moment_map} for the moment map into \eqref{EL eqs gauged univar a} and \eqref{EL eqs gauged univar b} gives \eqref{univariational eqs1} and \eqref{univariational eqs2}.  To obtain \eqref{univariational eqs3} we substitute \eqref{EL eqs gauged univar a} and \eqref{EL eqs gauged univar b} into the left-hand side of \eqref{EL eqs gauged univar c}, yielding
\begin{equation*}
\{ H_i, H_j \} = \{ H_i, \mu_a \} \widetilde{\mathcal A}^a_j\,.
\end{equation*}
The right-hand side vanishes using Proposition \ref{prop_global_symmetry} since $\{ H_i, \mu_a \} = \mathcal L_{X^\sharp_a} H_i = 0$, where the first equality follows using the definitions \eqref{action coords}, \eqref{Ham vec def} and \eqref{moment_map}. It follows that $\{ H_i, H_j \} = 0$, which completes the derivation of the Euler-Lagrange equations \eqref{EL gauged univar thm}.

Finally, relation \eqref{flatness} is a consequence of the vanishing of the Lie bracket of the coordinate vector fields on $\RR^n$ pushed forward by the map $\Sigma$ and restricted to the solution space. More explicitly, under the map $\Sigma$ a coordinate vector field $\parder{}{u^j}$ is sent to the vector field $Z_j=\parder{q^\mu}{u^j}\parder{}{q^\mu}+\parder{p_\mu}{u^j}\parder{}{p_\mu}+\parder{t^k}{u^j}\parder{}{t^k}$ on $T^\ast M\times \RR^n$ (see \eqref{push_forward}) whose restriction to the solution space reads, recalling \eqref{Ham vec def} and \eqref{action coords}, 
\begin{align}
    Z_j&=\left(\frac{\partial t^i}{\partial u^j} \frac{\partial H_i}{\partial p_\mu} - \frac{\partial \mu_a}{\partial p_\mu} \mathcal A^a_j\right)\parder{}{q^\mu}+\left(  - \frac{\partial t^i}{\partial u^j} \frac{\partial H_i}{\partial q^\mu} + \frac{\partial \mu_a}{\partial q^\mu} \mathcal A^a_j\right)\parder{}{p_\mu}+\parder{t^i}{u^j}\parder{}{t^i}\nonumber\\
    &=\frac{\partial t^i}{\partial u^j}\left(-\mathcal X_{H_i} +\widetilde{\mathcal A}_i^aX^\sharp_a+\parder{}{t^i}   \right)\,.
\end{align}
Now $\Big[\parder{}{u^j},\parder{}{u^k}\Big]=0$ implies $[Z_j,Z_k]=0$ which in turn yields
\begin{align}
    0&=\Big[ \mathcal X_{H_i} -\widetilde{\mathcal A}_i^aX^\sharp_a-\parder{}{t^i}  ,  \mathcal X_{H_j} -\widetilde{\mathcal A}_j^aX^\sharp_a-\parder{}{t^j} \Big]\nonumber\\
    &=[\mathcal X_{H_i},\mathcal X_{H_j}]- \widetilde{\mathcal A}_i^a[X^\sharp_a,\mathcal X_{H_j}] + \widetilde{\mathcal A}_j^a[X^\sharp_a,\mathcal X_{H_i}] + \left( \parder{}{t^i}\widetilde{\mathcal A}_j^a-\parder{}{t^j}\widetilde{\mathcal A}_i^a+ f_{bc}{}^a\widetilde{\mathcal A}_i^b\widetilde{\mathcal A}_j^c\right) X^\sharp_a  \nonumber\\
  \label{FX}  &=[\mathcal X_{H_i},\mathcal X_{H_j}]- \widetilde{\mathcal A}_i^a[X^\sharp_a,\mathcal X_{H_j}] + \widetilde{\mathcal A}_j^a[X^\sharp_a,\mathcal X_{H_i}] + F_{ij}^a X^\sharp_a 
\end{align}
where $F^a_{ij}$ are the components of the curvature of $\mathcal A$ defined in the statement of the theorem.
The second and third terms in \eqref{FX} vanish because $[X^\sharp_a,\mathcal X_{H_i}] = \mathcal X_{X^\sharp_a H_i}$ and $X^\sharp_a H_i = \mathcal{L}_{X^\sharp_a} H_i=0$, where the last step uses the condition \eqref{moment map eqs 2} from Proposition \ref{prop_local_symmetry} since we are assuming that $S_\Gamma[\Sigma, \mathcal A]$ is gauge invariant. The first term in \eqref{FX} vanishes because $[\mathcal X_{H_i}, \mathcal X_{H_j}] = \mathcal X_{\{H_i, H_j\}} = 0$ using the fact that $\{H_i,H_j\}=0$, as we have already established. Therefore $F_{ij}^a X^\sharp_a=0$. Since the action of $G$ is free, the tangent vectors $X^\sharp_a$ are linearly independent at each point $(p_\mu,q^\mu)$ of $T^\ast M$, and so we deduce the relation \eqref{flatness}, as required.

\end{proof}

The equation \eqref{A variation} together with \eqref{univariational eqs1}-\eqref{univariational eqs2} for any $i = 1, \ldots, n$ represent Hamiltonian flow equations for a constrained Hamiltonian system. Indeed, the effect of the first equation \eqref{A variation} is to impose the set of constraints $\mu^a(p,q) = 0$ on the phase space $T^\ast M$, restricting the dynamics to the submanifold $\mu^{-1}(0) \subset T^\ast M$. On the other hand, the equations \eqref{univariational eqs1}-\eqref{univariational eqs2} written in the form \eqref{EL eqs gauged univar a}-\eqref{EL eqs gauged univar b} describe, for each $i=1,\ldots, n$, the flow equation for the time-dependent Hamiltonian
\begin{equation}
\widetilde{H}_i(p,q, t) \coloneqq H_i(p,q) - \mu_a(p,q) \widetilde{\mathcal A}^a_i(t) \,.
\end{equation}
This flow is a linear combination of the Hamiltonian flow $\mathcal X_{H_i}$ of the unconstrained system and an arbitrary time-dependent linear combination of the flows $\mathcal X_{\mu_a}$ for $a = 1, \ldots, \dim \g$ along the orbits of the $G$-action, which implements the quotienting by $G$ to $\mu^{-1}(0)/G$. The last equation \eqref{univariational eqs3} encodes the Poisson commutativity of the flows for different $i=1,\ldots, n$ and was obtained in our variational setting (it relates to the so-called closure relation of Lagrangian multiform theory, see e.g. \cite{CDS} for a detailed exposition).

\section{3d mixed BF theory and Hitchin's system} \label{sec: 3d BF}

We will now generalise the construction of \S\ref{sec: univar finite-dim} to an infinite-dimensional context, where the role of the underlying manifold $M$ in \S\ref{sec: univar finite-dim} will be played here by the space $\M$ of holomorphic structures on a principle $G$-bundle $\mathcal P \to C$ over a compact Riemann surface $C$, for some connected Lie group $G$, and the role of the symmetry group $G$ in \S\ref{sec: univar finite-dim} will be played here by the group $\mathcal G = \Aut\mathcal P$ of (fibre-preserving) automorphisms of $\mathcal P$.

The cotangent bundle $T^\ast \mathcal M$ equipped with its free action of $\mathcal G$ is the well-known setup for the famous Hitchin system \cite{H} and by selecting the Hamiltonians of this system as our invariant functions $H_i$ on $T^\ast \mathcal M$, we will arrive in \S\ref{sec: Lag for Hitchin} at a variational description of the Hitchin system. The first main result of this section is that the Lagrangian $1$-form for the Hitchin system on the symplectic quotient $\mu^{-1}(0)/\mathcal G$ is given by the 3d mixed BF Lagrangian $1$-form with type B line defects, in the terminology of \cite{VW}, associated with each of the Hitchin Hamiltonians.
We then go on in \S\ref{sec: adding punctures} to extend the construction to the case of Hitchin systems on a compact Riemann surface $C$ with marked points. Our second main result is to show that this leads to a 3d mixed BF Lagrangian $1$-form with type B \textit{and} type A defects, in the terminology of \cite{VW}.

Since the construction of examples discussed later in \S\ref{sec: Hitchin holomorphic} and \S\ref{sec: examples} will crucially rely on some intricacies of the definition of $T^\ast \mathcal M$ and $\mathcal G$, we start by recalling the relevant details in \S\ref{sec: M and G details} and \S\ref{sec: T*M and G Hitchin}, thereby also setting our conventions.

\subsection{Geometric setup for the Hitchin system}\label{sec: geometric setup}

Let $G$ be a complex connected Lie group with Lie algebra $\g$. Let $C$ be a compact Riemann surface and fix a holomorphic atlas $\{ (U_I, z_I) \}_{I \in \mathcal I}$ of $C$ with $z_I : U_I \to \CC$ local holomorphic coordinates on each open subset $U_I \subset C$ and $\mathcal I$ some indexing set. 

\subsubsection{Holomorphic structures on a principal $G$-bundle} \label{sec: M and G details}

We fix a smooth principal $G$-bundle $\pi :\mathcal P\to C$ which is specified relative to the open cover $\{ U_I \}$ of $C$ by local trivialisations $\psi_I : \pi^{-1}(U_I) \SimTo U_I \times G$, $p \mapsto (\pi(p), f_I(p))$. The principal bundle is equipped with a free right action $G \times \mathcal P \to\mathcal P$, $p \mapsto p \cdot g$ and the local trivialisations should be $G$-equivariant, i.e. $f_I(p \cdot g) = f_I(p) g$ for any $g \in G$. The transition between local trivialisations $\psi_J$ and $\psi_I$ on overlapping charts $U_I \cap U_J \neq \emptyset$ is given by
\begin{equation*}
\psi_I \circ \psi_J^{-1} : (U_I \cap U_J) \times G \, \longrightarrow \, (U_I \cap U_J) \times G \,,\quad (x, g) \,\longmapsto\, \big( x, g_{IJ}(x) g \big)
\end{equation*}
with smooth transition functions $g_{IJ} : U_I \cap U_J \to G$, given by $g_{IJ}(x) = f_I(p) f_J(p)^{-1}$ for any $p \in\mathcal P$ with $\pi(p) = x \in U_I \cap U_J$, satisfying the \v{C}ech cocycle condition $g_{IJ} g_{JK} = g_{IK}$ on triple overlaps $U_I \cap U_J \cap U_K \neq \emptyset$.

A \emph{change of local trivialisation} of $\mathcal P$ is specified by a family of smooth maps $h_I : U_I \to G$, i.e. a \v{C}ech $0$-cochain $h = (h_I)_{I \in \mathcal I} \in \check{C}^0(C, G)$.
Indeed, given local trivialisations $\psi_I : \pi^{-1}(U_I) \SimTo U_I \times G$, $p \mapsto (\pi(p), f_I(p))$ we can define new local trivialisations by
\begin{equation} \label{new trivialisations}
\tilde \psi_I : \pi^{-1}(U_I) \overset{\cong}\longrightarrow U_I \times G \,, \quad p \longmapsto \big( \pi(p), h_I(\pi(p)) f_I(p) \big)\,.
\end{equation}
The transition functions of $\mathcal P$ relative to these new local trivialisations are the smooth maps
\begin{equation} \label{gauge transf compat}
\tilde g_{IJ} = h_I g_{IJ} h_J^{-1} : U_I \cap U_J \to G \,.
\end{equation}
Said differently, the change of local trivialisations from $\{ \psi_I \}_{I \in \mathcal I}$ to $\{ \tilde \psi_I \}_{I \in \mathcal I}$ on the fixed bundle $\mathcal P$ can be seen as producing a new principal $G$-bundle $\tilde {\mathcal P} \to C$ that is smoothly isomorphic to $\mathcal P$.

An automorphism of $\mathcal P$, or more precisely a fibre-preserving automorphism of $\mathcal P$ which we will sometimes refer to as a \emph{gauge transformation}, is a \v{C}ech $0$-cochain $g = (g_I)_{I \in \mathcal I} \in \check{C}^0(C, G)$ which preserves the transition functions of $P$ in the sense that
\begin{equation} \label{bundle morphism compat}
g_I = g_{IJ} g_J g_{IJ}^{-1}
\end{equation}
on any overlap $U_I \cap U_J \neq \emptyset$.
We can describe the action of $g$ on $\mathcal P$ relative to a fixed choice of local trivialisations $\{\psi_I \}_{I \in \mathcal I}$ as sending $\psi_I(p) = (\pi(p), f_I(p))$ to $\psi_I(g \cdot p) \coloneqq \big( \pi(p), g_I(\pi(p)) f_I(p) \big)$. The compatibility condition \eqref{bundle morphism compat} ensures that this is well-defined on $\mathcal P$, in the sense that we can either perform the gauge transformation directly in the local trivialisation $\psi_I$ or we can first move to the local trivialisation $\psi_J$, perform the gauge transformation there and then move back to the local trivialisation $\psi_I$. Both give the same result. We let
\begin{equation} \label{cal G def}
\mathcal G \coloneqq \Aut \mathcal P \subset \check{C}^0(C, G)
\end{equation}
denote the infinite-dimensional group of automorphisms of the principal $G$-bundle $\mathcal P$. Note that, by the condition \eqref{bundle morphism compat}, we can equally describe automorphisms of $\mathcal P$ as sections of the fibre bundle $\mathcal P \times_{\Ad} G$ associated with the adjoint representation of $G$ on itself.

A \emph{holomorphic structure} on $\mathcal P$ is a choice of local trivialisations $\{ \psi_I \}_{I \in \mathcal I}$ with respect to which the transition functions $g_{IJ} : U_I \cap U_J \to G$ are holomorphic.
It can equally be described \cite{AB} as a family of $\g$-valued $(0,1)$-forms $A''_{\bar z_I}(z_I,\bar z_I)\d \bar z_I \in \Omega^{0,1}(U_I, \g)$ relative to a choice of local trivialisations $\{ \psi_I \}_{I \in \mathcal I}$, denoted collectively as $A''$, such that 
\begin{equation} \label{01-form compat}
A''_{\bar z_I}\d\bar z_I = g_{IJ} A''_{\bar z_J} g_{IJ}^{-1}\d\bar z_J - \bar\partial g_{IJ} g_{IJ}^{-1} \,.
\end{equation}
on $U_I \cap U_J \neq \emptyset$. We let $\mathcal M$ denote the infinite-dimensional space of holomorphic structures on $\mathcal P$.

Under a change of local trivialisations $h \in \check{C}^0(C, G)$, the holomorphic structure $A''$ is described in the new local trivialisations \eqref{new trivialisations} by the family of $\g$-valued $(0,1)$-forms 
\begin{equation}
\label{change_triv_on_A}
\tilde A''_{\bar z_I}\d\bar z_I  = h_I A''_{\bar z_I} h_I^{-1}\d\bar z_I  - \bar\partial h_I h_I^{-1} \in \Omega^{0,1}(U_I, \g)\,.
\end{equation}
In particular, by solving the equations $A''_{\bar{z}_I} = h_I^{-1} \partial h_I/\partial\bar{z}_I$, which is always possible locally \cite[Section 5]{AB}, we obtain smooth maps $h_I : U_I \to G$ which define a new local trivialisation where
\begin{equation} \label{A zbar 0 gauge}
\tilde A''_{\bar z_I}(z_I,\bar z_I) = 0\,.
\end{equation}
This represents the same holomorphic structure $A''$ of $\mathcal P$ but now in an adapted local trivialisation of $\mathcal P$ where its components vanish. In particular, it now follows from \eqref{01-form compat} that in this new local trivialisation the transition functions $\tilde g_{IJ} : U_I \cap U_J \to G$ of the bundle are holomorphic.

Under a gauge transformation by $g \in \mathcal G$, a holomorphic structure $A'' \in \mathcal M$ is transformed to a new holomorphic structure on $\mathcal P$ given by the family of $\g$-valued $(0,1)$-forms 
\begin{equation} \label{gauge transformed Aalpha}
{}^{g_I} A''_{\bar z_I}\d\bar z_I  \coloneqq g_I  A''_{\bar z_I} g_I^{-1} \d\bar z_I- \bar\partial g_I g_I^{-1} \in \Omega^{0,1}(U_I, \g)\,.
\end{equation}
Let ${}^g A'' \coloneqq g A'' g^{-1} - \bar\partial g g^{-1} \in \mathcal M$ denote this transformed holomorphic structure. We have a left action of $\mathcal G$ on $\mathcal M$ given by
\begin{equation} \label{G action on M}
\mathcal G \times \mathcal M \,\longrightarrow\, \mathcal M\,, \qquad (g, A'') \,\longmapsto\, g\cdot A''\coloneqq {}^g A''=g A'' g^{-1} - \bar\partial g g^{-1}\,.
\end{equation}
Any $A'' \in \mathcal M$ determines a Dolbeault operator $\bar\partial^{A''}$ on any vector bundle $V_{\mathcal P} \coloneqq \mathcal P \times_\rho V$ associated with $\mathcal P$ in some representation $\rho : G \to \Aut V$, which acts on local sections over $U_I$ as $\bar\partial + \rho( A''_{\bar z_I})\d\bar z_I$.
In terms of Dolbeault operators, the left action \eqref{G action on M} reads $(g, \bar\partial^{A''}) \mapsto g \bar\partial^{A''} g^{-1}$.

The Lie algebra $\mathfrak G \coloneqq \text{Lie}(\mathcal G)$ of $\mathcal{G}$ consists of sections $X$ of the vector bundle $\g_{\mathcal P} = {\mathcal P}\times_{{\rm ad}} \g$ associated with ${\mathcal P}$ in the adjoint representation. Explicitly, this is given by a family of $\g$-valued functions $X^I \in C^\infty(U_I, \g)$ in each local trivialisation such that on each overlap $U_I \cap U_J \neq \emptyset$ we have the relation $X^I = g_{IJ} X^J g_{IJ}^{-1}$. The left action \eqref{G action on M} of the group $\mathcal G$ induces an infinitesimal left action of a Lie algebra element $X \in \mathfrak G$ on $A'' \in \mathcal M$ given by
\begin{equation} \label{g left action on M}
\delta_X A'' = -\bar\partial^{A''} X = - \bar\partial X - [A'', X]\,,
\end{equation}
or in local trivialisations by $\delta_{X^I} A''^{I} = -\bar\partial X^I - [A''^{I}, X^I]$.

\subsubsection{Cotangent bundle $T^\ast \mathcal M$ and action of $\mathcal G$} \label{sec: T*M and G Hitchin}

The tangent space $T_{A''} \mathcal M$ at any point $A'' \in \mathcal M$ is given by the space of sections of the bundle $\bigwedge^{0,1}C\otimes\g_{\mathcal P}$.  Likewise, the cotangent space $T^\ast_{A''} \mathcal M$ is the space of sections of $\bigwedge^{1,0}C\otimes\g_{\mathcal P}^\ast$, where $\g_{\mathcal P}^\ast \coloneqq {\mathcal P} \times_{{\rm ad}^\ast} \g^\ast$ is the vector bundle associated with ${\mathcal P}$ in the coadjoint representation. Concretely, any $X \in T_{A''} \mathcal M$ is described by a family of $\g$-valued $(0,1)$-forms $X^I = X_{\bar z_I}(z_I, \bar z_I) \d \bar z_I \in \Omega^{0,1}(U_I, \g)$ and any $Y \in T^\ast_{A''} \mathcal M$ by a family of $\g^\ast$-valued $(1,0)$-form $Y^I = Y_{z_I}(z_I,\bar z_I)\d z_I \in \Omega^{1,0}(U_I, \g^\ast)$ such that on any overlap $U_I \cap U_J \neq \emptyset$ we have the relations
\begin{equation} \label{10 and 01-form compat}
X^I = g_{IJ} X^J g_{IJ}^{-1}\,,\qquad Y^I = \Ad^\ast_{g_{IJ}} Y^J \,,
\end{equation}
respectively. Using the canonical pairing $\langle~,~\rangle : \g^\ast \times \g \to \CC$ we obtain a family of local $(1,1)$-forms $\langle Y^I, X^I\rangle \in \Omega^{1,1}(U_I)$, where we suppress a wedge product between the $1$-forms $Y^I$ and $X^I$. It follows from \eqref{10 and 01-form compat} that these local $(1,1)$-forms agree on overlaps, i.e. $\langle Y^I, X^I\rangle = \langle Y^J, X^J\rangle$, and hence define a global $(1,1)$-form on $C$ which we denote by $\langle Y, X\rangle \in \Omega^{1,1}(C)$. In particular, we can integrate the latter over the compact Riemann surface $C$ to obtain a pairing
\begin{equation} \label{TM T*M pairing}
T_{A''}^\ast \mathcal M \times T_{A''} \mathcal M \longrightarrow \CC \,, \qquad  (Y,X)\,\longmapsto\, \frac{1}{2 \pi i}\int_C \langle Y, X\rangle \,.
\end{equation}

A point in the cotangent bundle $T^\ast \mathcal M$ is given by a pair $(B, A'')$, with $A'' \in \mathcal M$ a holomorphic structure on ${\mathcal P}$ parametrising the base and $B$ a section of $\bigwedge^{1,0}C\otimes\g_{\mathcal P}^\ast$ parametrising the fibre. To describe vector fields on $T^\ast \mathcal M$ we note that we have the canonical isomorphism
\begin{equation} \label{T of T*M}
T_{(B, A'')} ( T^\ast \mathcal M ) \cong T_{A''}^\ast \mathcal M \oplus T_{A''} \mathcal M \,.
\end{equation}
The differential at $(B, A'') \in T^\ast \mathcal M$ of the projection $\pi_{\mathcal M} : T^\ast \mathcal M \to \mathcal M$, $(B, A'') \mapsto A''$ is a linear map $\delta_{(B, A'')} \pi_{\mathcal M}$ from the tangent space of $T^\ast \mathcal M$ at $(B, A'')$ to the tangent space of $\mathcal M$ at $A''$. Under the isomorphism \eqref{T of T*M}, it is given simply by the projection onto the second summand. By a standard abuse of notation, we will identify the map $\pi_{\mathcal M}$ with its value $A''$ at a generic point $(B, A'')$ and denote this differential by
\begin{equation} \label{delta A def}
\delta A'' : T_{(B, A'')}(T^\ast \mathcal M ) \,\longrightarrow\, T_{A''} \mathcal M \,, \qquad (Y,X) \,\longmapsto\, X\,.
\end{equation}
The tautological $1$-form on $T^\ast \mathcal M$ is then defined using the pairing \eqref{TM T*M pairing} as
\begin{equation} \label{tauto}
\alpha_{(B, A'')} \coloneqq \frac{1}{2 \pi i}\int_C \langle B, \delta A''\rangle \,.
\end{equation}
More explicitly, we can describe this as a map $\alpha_{(B, A'')} : T_{(B, A'')}(T^\ast \mathcal M ) \to \CC$ given by $\alpha_{(B, A'')}(Y, X) \coloneqq \frac{1}{2 \pi i}\int_C\langle B,X\rangle$.
The corresponding symplectic form $\omega \coloneqq \delta \alpha$ is given by 
\begin{equation} \label{omega Hitchin}
\omega_{(B, A'')}\big( (Y_1,X_1), (Y_2,X_2) \big) = \frac{1}{2 \pi i}\int_C\langle Y_1,X_2\rangle - \frac{1}{2 \pi i}\int_C\langle Y_2,X_1\rangle\,.
\end{equation}

Recall the left action \eqref{G action on M} of the group of gauge transformations $\mathcal G$ on the space of holomorphic structures $\mathcal M$. We can lift this to an action of $\mathcal G$ on $T^\ast \mathcal M$ as follows. In a local trivialisation, an element $g \in \mathcal G$ is represented by smooth maps $g_I : U_I \to G$ and a section $B \in T^\ast_{A''} \mathcal M$ of the bundle $\bigwedge^{1,0} C \otimes \g^\ast_{\mathcal P}$ is described by a family of $\g^\ast$-valued $(1,0)$-forms $B^I \in \Omega^{1,0}(U_I, \g^\ast)$. Since
\begin{equation} \label{10-form compat}
\Ad^\ast_{g_I} B^I = \Ad^\ast_{g_I g_{IJ}} B^J = \Ad^\ast_{g_{IJ} g_J} B^J = \Ad^\ast_{g_{IJ}} \big( \Ad^\ast_{g_J} B^J \big) \,,
    \end{equation}
we obtain a well-defined left action of $g = (g_I)_{I \in \mathcal I} \in \mathcal G$ on the fibres 
\begin{equation}\label{G act on B}
T^\ast_{A''} \mathcal M \,\longrightarrow\, T^\ast_{g \cdot A''} \mathcal M \,, \qquad B \,\longmapsto\, g \cdot B\coloneq \Ad^\ast_g B
\end{equation}
given explicitly in the local trivialisation over $U_I$ by $B^I \mapsto \Ad^\ast_{g_I} B^I$. Combining this with the left action of $\mathcal G$ on the base $\mathcal M$, we obtain the desired left action of $\mathcal G$ on $T^\ast \mathcal M$ given by
\begin{equation} \label{cal G action on T*M}
\mathcal G \times T^\ast \mathcal M \,\longrightarrow\, T^\ast \mathcal M \,, \qquad  \big(g, (B, A'') \big) \,\longmapsto\, g \cdot \big( B, A'' \big) \coloneqq \big( \Ad^\ast_g B, {}^g A'' \big) \,.
\end{equation}
This induces an infinitesimal left action of a Lie algebra element $X \in \mathfrak G$ on $(B, A'') \in T^\ast \mathcal M$ given by
\begin{equation} \label{global_inf_action}
\big( \delta_X B, \delta_X A'' \big) \coloneqq \big( {\rm ad}^*_X B, -\bar{\partial}^{A''} X \big) \,,
\end{equation}
where $\delta_X B = {\rm ad}^*_X B$ is given in local trivialisations by $\delta_{X^I} B^I = {\rm ad}^*_{X^I} B^I \in \Omega^{1,0}(U_I, \g^\ast)$.

As in section \ref{sec: univar finite-dim}, we will need a notion of freeness for the action of the group $\mathcal G$ on $T^\ast \mathcal M$.  We will say that $(B, A'')\in T^\ast\mathcal{M}$ is \emph{stable} if
\begin{equation}\label{inf_freeness}
\bar\partial X+[A'',X]=0,\quad {\rm ad}^*_X B=0\implies X=0.
\end{equation}
for all $X \in \mathfrak{G}$.  This means that the stabiliser of $(B, A'')$ is not a continuous subgroup of $\mathcal{G}$, so it is an infinitesimal version of freeness.

In the literature on Higgs bundles, the notion of a stable principal Higgs bundle has been introduced in \cite{BO}.  We expect that any Higgs bundle that is stable in the sense of \cite{BO} satisfies \eqref{inf_freeness}.  This was proved for semisimple irreducible reductive algebraic groups $G$ in the case $B=0$ in \cite{Ram}, Proposition 3.2.  In the appendix, we prove this statement for $B\neq0$ and $G=SL_m(\CC)$ using the Kobayashi--Hitchin correspondence of Simpson.  It would be interesting to prove it for more general $G$.

\subsection{Lagrangian 1-form on \texorpdfstring{$T^\ast \mathcal M$}{TM}} \label{sec: Lag for Hitchin}

In the present infinite-dimensional setting, the analogue of the Lagrangian $1$-form \eqref{L} is
\begin{equation} \label{L_Phi_A}
\Lag \coloneqq \alpha_{(B, A'')} - H_i( B,A'')\d t^i = \frac{1}{2 \pi i}\int_C \langle B,\delta A''\rangle - H_i( B,A'')\d t^i\,.
\end{equation}

To write down the corresponding action, i.e. the analogue of \eqref{ungauged action}, let $\Sigma : \RR^n \to T^\ast \mathcal M \times \RR^n$ be an immersion. For any $u \in \RR^n$ we write $\Sigma(u) = \big( B(u), A''(u), t(u) \big)$. It is helpful to take a moment to describe each component of $\Sigma(u)$ in more detail. Firstly, $B(u)$ describes an $\RR^n$-dependent section of $\bigwedge^{1,0} C \otimes \g^\ast_{\mathcal P}$ given in each local trivialisation of $\mathcal P$ over $U_I$ by $\g^\ast$-valued $(1,0)$-forms $B_{z_I}(z_I, \bar z_I, u) \d z_I \in \Omega^{1,0}(U_I \times \RR^n, \g^\ast)$. These are related in chart overlaps $U_I \cap U_J \neq \emptyset$ by the second relation in \eqref{10 and 01-form compat}, explicitly
\begin{equation} \label{B(s) transition}
B_{ z_I}(z_I, \bar z_I,u)\d z_I = \Ad^\ast_{g_{IJ}} B_{ z_J}(z_J, \bar z_J,u)\d z_J\,.
\end{equation}
Secondly, $A''(u)$ describes an $\RR^n$-dependent element of $\mathcal M$ given in each local trivialisation of ${\mathcal P}$ over $U_I$ by $\g$-valued $(0,1)$-forms $A''_{\bar z_I}(z_I, \bar z_I, u) \d \bar z_I \in \Omega^{0,1}(U_I, \g)$. These are related in chart overlaps $U_I \cap U_J \neq \emptyset$ by \eqref{01-form compat}, explicitly
\begin{equation} \label{A(s) transition}
A''_{\bar z_I}(z_I, \bar z_I,u)\d\bar z_I = g_{IJ} A''_{\bar z_J}(z_J, \bar z_J,u) g_{IJ}^{-1}\d\bar z_J - \bar\partial g_{IJ} g_{IJ}^{-1} \,.
\end{equation}
Finally, $t(u)$ describes an $\RR^n$-dependent point in $\RR^n$ with components $t^i(u)$ for $i=1,\ldots, n$.

Note that since the transition functions $g_{IJ} : U_I \cap U_J \to G$ of $\mathcal P$ obviously do not depend on the parameter $u = (u^j) \in \RR^n$, differentiating the relation \eqref{A(s) transition} with respect to $u^j$ we obtain
\begin{equation} \label{dA(s) transition}
\partial_{u^j} A''_{\bar z_I}(z_I, \bar z_I,u)\d\bar z_I = g_{IJ} \partial_{u^j} A''_{\bar z_J}(z_J, \bar z_J,u) g_{IJ}^{-1}\d\bar z_J \,.
\end{equation}
Thus, for every $j = 1, \ldots, n$, the family of $\g$-valued $(0,1)$-forms $\partial_{u^j} A''_{\bar z_I}\d\bar z_I \in \Omega^{0,1}(U_I, \g)$ defines an $\RR^n$-dependent section of $\bigwedge^{0,1} C \otimes \g_{\mathcal P}$ which we denote by $\partial_{u^j} A''(u)$. Similarly, we denote by $\partial_{u^j} B(u)$ the family of $\g^\ast$-valued $(1,0)$-forms $\partial_{u^j} B_{ z_I}\d z_I \in \Omega^{1,0}(U_I, \g^\ast)$.
Now let $\Sigma_1 : \RR^n \to T^\ast \mathcal M$, $u \mapsto \big( B(u), A''(u) \big)$ be the component of the map $\Sigma$ in $T^\ast \mathcal M$. Its differential at $u \in \RR^n$ is
\begin{equation*}
\d_u \Sigma_1 : T_u \RR^n \,\longrightarrow\, T_{\Sigma_1(u)}(T^\ast \mathcal M ) \cong T^\ast_{A''(u)} \mathcal M \oplus T_{A''(u)} \mathcal M \,, \qquad
\frac{\partial}{\partial u^j} \,\longmapsto\, \big( \partial_{u^j} B(u), \partial_{u^j} A''(u) \big)\,.
\end{equation*}
The pullback of the differential $\delta A''$ defined in \eqref{delta A def} by the map $\Sigma_1$ is then given by the composition $\big( \Sigma^\ast (\delta A'') \big)(u) = \delta_{\Sigma_1(u)} A'' \circ \d_u \Sigma_1$ and hence $\Sigma^\ast (\delta A'') = \partial_{u^j} A''(u) \wedge\d u^j = -\d_{\RR^n} A''(u)$.
We now find that the pullback $\Sigma^\ast \Lag$ of the Lagrangian \eqref{L_Phi_A} by the map $\Sigma$ is given by
\begin{equation} \label{Sigma pull L}
\Sigma^\ast \Lag = \frac{1}{2 \pi i}\int_C \big\langle B(u), \d_{\RR^n} A''(u) \big\rangle - H_i\big( B(u), A''(u) \big) \d_{\RR^n} t^i \,.
\end{equation}
Given an arbitrary curve $\Gamma : (0,1) \to \RR^n$, $s \mapsto \big( u^j(s) \big)$ can now finally write down the analogue of the action \eqref{ungauged action} in the present case, which reads
\begin{equation} \label{ungauged action Phi A}
S_\Gamma[\Sigma] = \int_0^1 (\Sigma \circ \Gamma)^\ast \Lag = \int_0^1 \bigg(\! - \frac{1}{2 \pi i} \int_C \big\langle B(u), \partial_{u^j} A''(u) \big\rangle - H_i\big( B(u), A''(u) \big)\frac{\partial t^i}{\partial u^j} \bigg) \frac{\d u^j}{\d s} \d s \,.
\end{equation}
\begin{remark} \label{rem: gauge in Hitchin}
At this stage, a few important remarks on notations and terminology are in order to avoid confusion. Note that the kinetic part $\frac{1}{2 \pi i}\int_C\langle B, \delta A''\rangle$ of \eqref{L_Phi_A} is the direct analogue of $p_\mu \d q^\mu$ in \eqref{L}. In particular, the integration over $C$ in the present setting is the analogue of the summation over $\mu \in \{ 1, \ldots, m \}$ in the finite-dimensional setting of \S\ref{sec: univar finite-dim}. This means that, although the first term in the action \eqref{ungauged action Phi A} involves the integral of a $3$-form over $C \times (0,1)$, from the point of view of Lagrangian multiform theory we should really regard the whole action $S_\Gamma[\Sigma]$ as the integral of a $1$-form on $(0,1)$, namely the pullback of \eqref{L_Phi_A} along $\Sigma \circ \Gamma : (0,1) \to T^\ast \mathcal M \times \RR^n$.

The terminology ``gauge group'' for the group ${\cal G}$ of fibre-preserving automorphisms is standard in the geometric formulation of the Hitchin system and that is why we used it here. However, it is crucial to note that, at this stage of the construction, the group $\mathcal G$ is the analogue of what we called the {\it global} symmetry group $G$ in \S\ref{sec: univar finite-dim}, which we will later gauge by considering transformations parametrised by maps $g:\RR^n \to {\cal G}$ by analogy with the finite-dimensional setting of \S\ref{sec: gauging symmetry}; see \S\ref{sec: Lag for Hitchin mod G} for details. A crude way to say this is that we will ``gauge the gauge group $\mathcal{G}$''.
\end{remark}

With this in mind, we shall now prove the analogue of Proposition \ref{prop_global_symmetry} in the present infinite-dimensional context. In fact, the setting of \S\ref{sec: group actions} was very generic and the infinitesimal action of a Lie algebra element $X \in \g$ on $M$ was specified only implicitly through the vector fields $X^\sharp$. By contrast, in the present context, we have an explicit description of the action of the Lie algebra $\mathfrak G$ on $T^\ast \mathcal M$ in \eqref{global_inf_action} and even of the action of the group $\mathcal G$ on $T^\ast \mathcal M$ in \eqref{cal G action on T*M}. We can therefore prove a stronger statement than Proposition \ref{prop_global_symmetry} in the present case. We first need to lift the action of $\mathcal G$ to $T^\ast \mathcal M \times \RR^n$ by letting it act trivially on $\RR^n$, i.e. we set $g \cdot t = t$ for any $g \in \mathcal G$ and $t \in \RR^n$.

\begin{proposition} \label{prop: Hi invariance Hitchin}
The action \eqref{ungauged action Phi A} is invariant under the action of $\mathcal G$ on $T^\ast \mathcal M \times \RR^n$ given in \eqref{cal G action on T*M} if and only if each $H_i$ for $i=1, \ldots, n$ is invariant under the group action, i.e.
\begin{equation} \label{Hi invariance prop}
H_i\big( \Ad^\ast_g B, {}^g A'' \big) = H_i(B, A'')
\end{equation}
for any $(B, A'') \in T^\ast \mathcal M$ and $g \in \mathcal G$. Moreover, the Noether charge associated with an infinitesimal bundle morphism $X \in \mathfrak G$ is given by
\begin{equation} \label{moment_map_Hitchin}
\mu_{(B, A'')}(X) = \frac{1}{2 \pi i}\int_C \big\langle B,\bar{\partial}^{A''} X \big\rangle \,.
\end{equation}
\begin{proof}
We closely follow the proof of Proposition \ref{prop_global_symmetry}. Let $\Sigma : \RR^n \to T^\ast \mathcal M \times \RR^n$ be given by $\Sigma(u) = \big( B(u), A''(u), t(u) \big)$. By \eqref{cal G action on T*M}, its pointwise image under the left action of any $g \in \mathcal G$ is $g \cdot \Sigma : \RR^n \to T^\ast \mathcal M \times \RR^n$ given by
\begin{equation} \label{g act on gamma Hitchin}
(g \cdot \Sigma)(u) = \big( \Ad^\ast_g B(u), {}^g A''(u), t(u) \big)\,.
\end{equation}
Note, in particular, that this implies $\partial_{u^j} (g \cdot \Sigma)(u) = \big( \Ad^\ast_g \partial_{u^j} B(u), g \partial_{u^j} A''(u) g^{-1}, \partial_{u^j} t(u) \big)$. The action for the transformed map $g \cdot \Sigma$ therefore reads
\begin{align*}
S_\Gamma[g \cdot \Sigma] &= \int_0^1\left(-\frac{1}{2 \pi i}\int_C \big\langle \Ad^\ast_g B(u), g \partial_{u^j} A''(u) g^{-1} \big\rangle - H_i\big( \Ad^\ast_g B(u), {}^g A''(u) \big)\frac{\partial t^i}{\partial u^j}\right) \frac{\d u^j}{\d s} \d s\\
&= \int_0^1\left(-\frac{1}{2 \pi i}\int_C \big\langle B(u), \partial_{u^j} A''(u) \big\rangle - H_i\big( \Ad^\ast_g B(u), {}^g A''(u) \big)\frac{\partial t^i}{\partial u^j}\right) \frac{\d u^j}{\d s} \d s\\
&= S_\Gamma[\Sigma] + \int_0^1 \Big( H_i\big( B(u), A''(u) \big) - H_i\big( \Ad^\ast_g B(u), {}^g A''(u) \big) \Big) \frac{\d t^i}{\d s} \d s
\end{align*}
The result now follows since for $g \in \mathcal G$ to be a symmetry means that $S_\Gamma[g \cdot \Sigma] = S_\Gamma[\Sigma]$ for any map $\Sigma$ and any curve $\Gamma$ and hence the integral on the right-hand side must vanish for any curve $\Gamma$ but this, in turn, is equivalent to the condition \eqref{Hi invariance prop}.

To work out the Noether charge associated with the infinitesimal symmetry generated by a Lie algebra element $X \in \mathfrak G$, we introduce an arbitrary smooth function $\lambda : \RR^n \to \CC$ and consider now the pointwise variations of the map $\Sigma(u) = \big( B(u), A''(u), t(u) \big)$ given by 
\begin{align*}
\delta_{\lambda(u) X} B(u) &= \lambda(u) {\rm ad}^*_X B(u) \,,\\
\delta_{\lambda(u) X} A''(u) &= - \lambda(u) \bar\partial^{A''(u)} X = - \lambda(u) \big( \bar\partial X + [A''(u), X] \big)\,.
\end{align*}
The variation of the action \eqref{ungauged action Phi A} then only has a contribution from the kinetic term, which reads
\begin{align*}
\delta_{\lambda(u) X} S_\Gamma[\Sigma] &= \delta_{\lambda(u) X} \int_0^1 \frac{-1}{2 \pi i}\int_C \big\langle B(u), \partial_{u^j} A''(u) \big\rangle \frac{\d u^j}{\d s} \d s \\
&=  \int_0^1 \partial_{u^j} \lambda(u) \left(\frac{1}{2 \pi i}\int_C \big\langle B(u), \bar\partial^{A''(u)} X \big\rangle \right)\frac{\d u^j}{\d s} \d s \,,
\end{align*}
cf. the end of the proof of Propostion \ref{prop_global_symmetry}. From this we read off the desired expression \eqref{moment_map_Hitchin}. 
\end{proof}
\end{proposition}

\begin{remark} \label{rem: mu well defined}
We make a trivial but important comment concerning the expression \eqref{moment_map_Hitchin}, to help avoid potential confusion later. Note that this expression is well-defined since $\bar\partial^{A''} X$ is a section of $\bigwedge^{0,1} C \otimes \g_{\mathcal P}$. Indeed, using the relations \eqref{01-form compat} and $X^I = g_{IJ} X^J g_{IJ}^{-1}$, respectively, between the local expressions of $A'' \in \mathcal M$ and $X \in \mathfrak G$ in overlapping charts $U^I \cap U^J \neq \emptyset$, we find that
\begin{equation*}
\bar\partial X^I + [A''_{\bar z_I}, X^I]\d\bar z_I = g_{IJ} \big( \bar\partial X_J + [ A''_{\bar z_J}, X_J]\d\bar{z}^J \big) g_{IJ}^{-1} \,.
\end{equation*}
We thus have $\langle B^I, \bar\partial^{A''} X^I \rangle = \langle B^J, \bar\partial^{A''} X^J \rangle$ on overlaps $U_I \cap U_J \neq \emptyset$ so that these define a global $(1,1)$-form $\langle B, \bar\partial^{A''} X\rangle \in \Omega^{1,1}(C)$ which can be integrated over the compact Riemann surface $C$. 
Likewise, we have a well-defined global $(1,1)$-form $\langle\bar\partial^{A''} B, X\rangle \in \Omega^{1,1}(C)$ given by the expression $\langle \bar\partial^{A''} 
B^I, X^I \rangle$ in each local chart $U_I$. Moreover, these are related by
\begin{equation} \label{three well defined sections}
\langle B, \bar\partial^{A''} X\rangle - \langle\bar\partial^{A''} B, X\rangle = - \d_C \langle B, X\rangle
\end{equation}
where $\langle B,X\rangle \in \Omega^{1,0}(C)$ is a well-defined $(1,0)$-form on $C$. The relative sign on the left-hand side comes from the fact that $B$ is a $1$-form and the operator $\bar\partial^{A''}$ has cohomological degree $1$. Integrating both sides over $C$ and using Stokes's theorem on the right-hand side, noting that $C$ has no boundary, we deduce that the Noether charge \eqref{moment_map_Hitchin} can equivalently be rewritten as
\begin{equation} \label{moment_map_Hitchin 2}
\mu_{(B, A'')}(X) = \frac{1}{2 \pi i}\int_C \big\langle \bar{\partial}^{A''} B, X \big\rangle\,.
\end{equation}
In other words, the value of the corresponding moment map $\mu : T^\ast \mathcal M \to \mathfrak G^\ast$ at $(B, A'') \in T^\ast \mathcal M$ is the element of $\mathfrak G^\ast$ given by the linear map $X \mapsto \frac{1}{2 \pi i}\int_C ( \bar{\partial}^{A''} B, X )$ which takes in any vector $X \in \mathfrak G$, i.e. a section of $\g_{\mathcal P}$, and integrates it against the section $\bar{\partial}^{A''} B$ of $\bigwedge^{1,1} C \otimes \g^\ast_{\mathcal P}$.
\end{remark}

The construction of the Hamiltonians $H_i : T^\ast \mathcal M \to \CC$ satisfying \eqref{Hi invariance prop} will be inspired by that of the Hitchin map \cite{H}. Recall that the latter is constructed from a choice of
\begin{itemize}
  \item[$(i)$] $G$-invariant homogeneous polynomials $P_r:\g^\ast \to\CC$ for $r = 1,\ldots, \text{rk}\,\g$ of degree $d_r+1$, where $E = \{ d_r \}_{r=1}^{\text{rk} \, \g}$ is to the set of exponents of $\g$.
\end{itemize}
Given such data, the Hitchin map 
\begin{equation} \label{Hitchin map}
P \coloneqq (P_1, \ldots, P_{\text{rk}\, \g}) : H^0 \Big( C, {\textstyle \bigwedge^{1,0}} C \otimes \g^\ast_{\mathcal P} \Big) \,\longrightarrow\, \bigoplus_{r=1}^{\text{rk}\,\g} H^0 \Big( C, \big( {\textstyle \bigwedge^{1,0}} C \big)^{\otimes (d_r+1)} \Big)\,,
\end{equation}
takes as input a holomorphic section $B$ of the bundle $\bigwedge^{1,0} C \otimes \g^\ast_{\mathcal P}$.
Recall that this is given by a family of $\g^\ast$-valued $(1,0)$-forms $B^I = B_{z_I}(z_I) \d z_I \in \Omega^{1,0}(U_I, \g^\ast)$ in the local trivialisation over the chart $(U_I, z_I)$, where here the function $B_{z_I}$ depends holomorphically on $z_I$, satisfying the second relation in \eqref{10 and 01-form compat}, explicitly $B^I = \Ad^\ast_{g_{IJ}} B^J$ on $U_I \cap U_J \neq \emptyset$.
Since $P_r$ is $G$-invariant, it follows that we have
\begin{equation*}
P_r(B^I) = P_r\big( \Ad^\ast_{g_{IJ}} B^J \big) = P_r(B^J)
\end{equation*}
on non-trivial overlaps $U_I \cap U_J \neq \emptyset$. In this way, we obtain a holomorphic section of the bundle $(\bigwedge^{1,0} C)^{\otimes (d_r+1)}$ for each $r =1, \ldots, \text{rk}\, \g$, which we denote by $P_r(B)$, and the Hitchin map returns a holomorphic section of $\bigoplus_{r =1}^{\text{rk}\,\g} (\bigwedge^{1,0} C)^{\otimes (d_r+1)}$. To obtain individual complex-valued Hamiltonians one can then expand each component $P_r(B)$ of the Hitchin map in a basis of holomorphic $(d_r+1,0)$-differentials on $C$.

However, in our present setting, $B \in T^\ast_{A''} \mathcal M$ is only a \emph{smooth} section of $\bigwedge^{1,0} C \otimes \g^\ast_{\mathcal P}$. We can still form a smooth section $P_r(B)$ of $(\bigwedge^{1,0} C)^{\otimes (d_r+1)}$ for every $r = 1, \ldots, \text{rk}\,\g$, however we can no longer expand it in a basis of holomorphic sections of $(\bigwedge^{1,0} C)^{\otimes (d_r+1)}$. Instead, in order to produce complex-valued Hamiltonians we will proceed along the lines of \cite{VW} by introducing a marked point on $C$ for each Hamiltonian $H_i$ (note that in \cite{VW} only one marked point was needed since a single Hamiltonian was considered). We therefore introduce the following additional data:
\begin{itemize}
  \item[$(ii)$] Points $\mathsf q_{rl} \in C$ labelled by pairs $(r,l)$ with $r = 1,\ldots, \text{rk}\,\g$ and $l = 1, \ldots, m_r$, where
\begin{equation*}
m_r \coloneqq \dim \Big( H^0 \Big( C, \big( {\textstyle \bigwedge^{1,0}} C \big)^{\otimes (d_r+1)} \Big) \Big) = \left\{
\begin{array}{ll}
(2 d_r+1) (g-1) \,, & \quad \text{for}\; g \geq 2 \,,\\
g \,, & \quad \text{for} \; g=0,1
\end{array}
\right.
\end{equation*}
and a set of holomorphic tangent vectors $V_{\mathsf q_{rl}} \in T^{1,0}_{\mathsf q_{rl}} C$.
\end{itemize}

We can now define $H_i : T^\ast \mathcal M \to \CC$ by evaluating the smooth section $P_r(B)$ of $(\bigwedge^{1,0} C)^{\otimes (d_r+1)}$ at $\mathsf q_{rl} \in C$ and pairing the resulting element $P_r\big( B(\mathsf q_{rl}) \big) \in (\bigwedge^{1,0}_{\mathsf q_{rl}} C)^{\otimes (d_r+1)}$ with $V_{\mathsf q_{rl}}^{d_r+1} \in (T^{1,0}_{\mathsf q_{rl}} C)^{\otimes (d_r+1)}$, that is
\begin{equation} \label{pre Hitchin Hamiltonians}
H_i(B) \coloneqq \big\langle P_r\big( B(\mathsf q_{rl}) \big), V_{\mathsf q_{rl}}^{ d_r+1} \big\rangle \,.
\end{equation}
Concretely, if $(U_I, z_I)$ is a local chart around one of the points $\mathsf q_{rl} \in C$ then we can pick $V_{\mathsf q_{rl}} = \partial_{z_I}$ and the above geometric construction amounts to writing $P_r(B) = P_r\big( B_{z_I}(z_I, \bar z_I) \big) \d z_I^{\otimes (d_r+1)}$ locally in the coordinate $z_I$ and then evaluating its component at $\mathsf q_{rl} \in \CC$, i.e.
\begin{equation} \label{Hitchin Hamiltonians}
H_i(B) = P_r\big( B_{z_I}(\mathsf q_{rl}) \big)\,.
\end{equation}
The understanding in \eqref{pre Hitchin Hamiltonians} and \eqref{Hitchin Hamiltonians} is that  the label $i$ on the Hamiltonians runs over pairs $(r, l)$ with $r = 1, \ldots, \text{rk}\, \g$ and $l = 1, \ldots, m_r$.
Note that $m_r$ being the dimension of the space of holomorphic sections of $(\bigwedge^{1,0} C)^{\otimes (d_r+1)}$ ensures that, when the points $\mathsf q_{rl}$ are generic, the number of Hamiltonians we produce coincides with the number of Hamiltonians obtained via the construction of the Hitchin map \eqref{Hitchin map} when the smooth section $B$ of the bundle $\bigwedge^{1,0} C \otimes \g^\ast_{\mathcal P}$ becomes holomorphic. Indeed, our evaluation prescription \eqref{pre Hitchin Hamiltonians} defines a bijection
\begin{equation*}
H^0 \Big( C, \big( {\textstyle \bigwedge^{1,0}} C \big)^{\otimes (d_r+1)} \Big) \,\overset{\cong}\longrightarrow\, \CC^{m_r} \,,
\end{equation*}
for generic points $\mathsf q_{rl} \in C$, $l = 1, \ldots, m_r$. It is enough to show this is injective, which follows from the fact that divisors $D_r \coloneqq \sum_{l=1}^{m_r} \mathsf q_{rl}$ on $C$ for which $\deg(D_r) = m_r = \dim \big( H^0 \big( C, ( {\textstyle \bigwedge^{1,0}} C \big)^{\otimes (d_r+1)} \big) \big)$ and $\dim \big( H^0 \big( C, ( {\textstyle \bigwedge^{1,0}} C \big)^{\otimes (d_r+1)} \otimes \mathcal O(-D_r) \big) \big) = 0$ are generic.

Since the label $i$ on the Hamiltonians runs over pairs $(r, l)$ with $r = 1, \ldots, \text{rk}\, \g$ and $l = 1, \ldots, m_r$, it runs from $1$ to
\begin{equation}
n \coloneqq \sum_{r=1}^{\text{rk}\, \g} m_r = (g-1) \sum_{r=1}^{\text{rk}\, \g} (2 d_r + 1) = (g-1) \dim \g 
\end{equation}
when $g \geq 2$ and from $1$ to $n \coloneqq \text{rk}\, \g$ when $g=1$. This number coincides with half the dimension of the phase space of Hitchin's integrable system when $g \geq 2$ and $g=1$, respectively. When $g=0$, however, there are no Hamiltonians since the set of points $\mathsf q_{rl} \in C$ introduced in condition $(ii)$ above is empty.
Producing non-trivial integrable systems in the case of genus $g=0$ will require introducing additional marked points on $C$ which we will turn to in \S\ref{sec: adding punctures} below.

{\bf Label notations:} From now on, for notational convenience, we will simply use the common label $i$ for the Hamiltonians $H_i$, the polynomials $P_r$, the points $\mathsf q_{rl}$ and the times $t^i$ associated to the Hamiltonians. This amounts to relabelling $P_r$ as $P_{rl}$, with the understanding that $P_{rl} = P_{rl'}$ for any $l, l' =1, \ldots, m_r$, so that we can write simply $H_i(B) = P_i\big( B_{z_I}(\mathsf q_i) \big)$. Accordingly, we can keep denoting by $t^i$ the time associated to $H_i$, rather than the cumbersome $t^{rl}$ or $t^{(r,l)}$.

\subsection{Lagrangian 1-form for the Hitchin system on \texorpdfstring{$\mu^{-1}(0)/\mathcal G$}{mu G}} \label{sec: Lag for Hitchin mod G}

So far we have introduced the action $S_\Gamma[\Sigma]$ in \eqref{ungauged action Phi A} for an immersion $\Sigma : \RR^n \to T^\ast \mathcal M \times \RR^n$ and an arbitrary curve $\Gamma : (0,1) \to \RR^n$, and showed in Proposition \ref{prop: Hi invariance Hitchin} that it is invariant under the left action of the group $\mathcal G$ provided that the Hamiltonians $H_i$ themselves are $\mathcal G$-invariant in the sense that \eqref{Hi invariance prop} holds. Note that the latter condition clearly holds for the Hamiltonians introduced in \eqref{Hitchin Hamiltonians} by virtue of the $G$-invariance of the polynomials $P_i : \g^\ast \to \CC$. Moreover, we identified the moment map $\mu : T^\ast \mathcal M \to \mathfrak G^\ast$ associated with this symmetry as given by \eqref{moment_map_Hitchin}.

We now generalise the gauging procedure of \S\ref{gauged_univariational_principle} to the present infinite-dimensional setting.  This requires introducing two elements: gauge transformations $g$ and a gauge field $\mathcal{A}$.  We begin with the gauge transformations.

For the Hitchin system, the group by which we wish to quotient is the group $\mathcal{G}$ of automorphisms of the bundle ${\mathcal P}\to C$.  In the language of \S \ref{sec: univar finite-dim}, the group $\mathcal{G}$ corresponds to the \textit{global} symmetry group $G$. Thus here, in order to gauge $\mathcal G$, we consider the group of local transformations $g:\RR^n\to \mathcal{G}$, \ie automorphisms of the bundle ${\mathcal P}\to C$ that depend smoothly on $u\in\RR^n$.  In a local trivialisation of ${\mathcal P}$, this is represented by $G$-valued functions $g_I(z_I,\bar{z}_I,u)$. On overlaps $U_I \cap U_J \neq \emptyset$ they satisfy \eqref{bundle morphism compat}, in which $g_{IJ}$ are the transition functions relative to the chosen trivialisation of ${\mathcal P}\to C$.

The action of a local gauge transformation $g$ on the map $\Sigma:u\mapsto(A''(u),B(u),t(u))$ is exactly as in \S \ref{sec: geometric setup}, except that $g$, $A''$, $B$ and $t$ depend on the parameters $u$ (in addition to depending on local coordinates $z_I$).  Explicitly, we write
\begin{equation}
g\cdot\Sigma :u\longmapsto \big( {}^{g(u)}A''(u),\Ad^\ast_{g(u)}B(u),t(u) \big),
\end{equation}
where ${}^{g(u)}A''(u)$ and $\Ad^\ast_{g(u)}B(u)$ are defined as in \eqref{G action on M} and \eqref{G act on B} for each $u$.

Next, we introduce the gauge field $\mathcal{A}=\mathcal{A}_i \d u^i$.  This consists of functions $\mathcal{A}_i$ from $\RR^n$ to the Lie algebra $\mathfrak{G}$ of $\mathcal{G}$.  The Lie algebra $\mathfrak{G}$ is the space of sections of the bundle $\mathfrak{g}_{\mathcal P}$, so each $\mathcal{A}_i$ is a section of $\mathfrak{g}_{\mathcal P}$ that depends smoothly on $u\in\mathbb{R}^n$.  In any local trivialisation, $\mathcal{A}_i$ is represented by functions $\mathcal{A}^I_{i}(z_I,\bar{z}_I,u)$ that take values in $\mathfrak{g}$, and on the overlap $U_I \cap U_J \neq \emptyset$ these are related by
\begin{equation}\label{Ai compat}
\mathcal{A}^I_{i}=g_{IJ}\mathcal{A}^J_{i}g_{IJ}^{-1} \,.
\end{equation}
Under a local transformation $g:\RR^n\to\mathcal{G}$, these transform as
\begin{equation}\label{G act on Ai}
\mathcal{A}^I_{i}\longmapsto g_I\mathcal{A}^I_{i}g_{I}^{-1} - \frac{\partial g_{I}}{\partial u^i}g_I^{-1}\,.
\end{equation}

Having introduced the local transformations and the gauge field, we are now able to write down the analogue of the gauged action \eqref{uni_gauged action} for the Hitchin system. Recall that this consists in adding to the ungauged action \eqref{ungauged action Phi A} a term that couples the moment map \eqref{moment_map_Hitchin} to the gauge field as
\begin{equation}\label{gauged Hitchin action}
S_\Gamma[\Sigma,\mathcal{A},t] = 
\int_0^1 \left( -\frac{1}{2 \pi i}\int_C \left\langle B(u), \frac{\partial A''}{\partial u^j}- \bar{\partial}^{A''}\mathcal{A}_j \right\rangle - H_i\big( B(u) \big)\frac{\partial t^i}{\partial u^j} \right) \frac{\d u^j}{\d s} \d s \,.
\end{equation}

We now show that the first term in this action is exactly a multiform version of the $3$d mixed BF Lagrangian. To do so, let
\begin{equation}
A \coloneqq A'' + \mathcal{A}\,.
\end{equation}
This is a partial connection on the pullback bundle $\pi_C^\ast {\mathcal P}={\mathcal P}\times\RR^n$ over $C\times\mathbb{R}^n$ along the projection $\pi_C:C\times\RR^n\to C$.  By definition, the bundle $\pi_C^\ast {\mathcal P}$ is trivialised over open sets $U_I\times \RR^n$, and the transition functions between these sets are the transition functions $g_{IJ}$ of ${\mathcal P}$, which obviously satisfy
\begin{equation}\label{dgdu}
\frac{\partial g_{IJ}}{\partial u^j}=0 \,.
\end{equation}
In these local trivialisations, $A$ takes the form
\begin{equation}\label{A definition}
A^I = A''_{\bar{z}_I}(z_I,\bar{z}_I,u)\d\bar{z}_I + \mathcal{A}^I_{i}(z_I,\bar{z}_I,u)\d u^i\,.
\end{equation}
This looks like the local expression for a connection, except that it is missing a $\d z_I$-component, which is why we refer to it as a ``partial'' connection.  A partial connection is a connection that can only take derivatives in certain directions (the holomorphic structure $A''$ of $\mathcal P$ is also an example of a partial connection).  To show that $A$ is a well-defined partial connection we must check that it satisfies
\begin{equation}\label{partial connection compat}
A^I = g_{IJ}A^Jg_{IJ}^{-1} - \bar{\partial}g_{IJ}g_{IJ}^{-1} - \d_{\RR^n} g_{IJ} g_{IJ}^{-1}
\end{equation}
on overlaps $U_I \cap U_J \neq \emptyset$, and that it is independent of the choice of local trivialisations of ${\mathcal P}\to C$. Equations \eqref{01-form compat}, \eqref{Ai compat} and \eqref{dgdu} show that \eqref{partial connection compat} is satisfied. Independence of the choice of local trivialisation follows from \eqref{partial connection compat}: if we use different local trivialisations of ${\mathcal P}\to C$ over open sets $\{U_I\}_{I\in\mathcal{I}'}$ to construct a connection $A'$, then $A$ and $A'$ will be related as in \eqref{partial connection compat} on the overlaps $U_I\cap U_J$ for $I\in\mathcal{I}$ and $J\in\mathcal{I}'$.  Therefore $A'$ and $A$ represent the same connection.

The curvature $F_A$ of the partial connection $A$ is defined in local trivalisations by
\begin{equation}\label{FA def}
F_A^{I}=\left(\frac{\partial \mathcal A^I_{i}}{\partial \bar{z}_I}-\frac{\partial A''_{\bar{z}_I}}{\partial u^i}+[A''_{\bar{z}_I},\mathcal A^I_{i}]\right)\d\bar{z}_I\wedge \d u^i + \frac12\left(\frac{\partial \mathcal A^I_{j}}{\partial u^i}-\frac{\partial \mathcal A^I_{i}}{\partial u^j}+[\mathcal A^I_{i},\mathcal A^I_{j}]\right)\d u^i \wedge \d u^j\,.
\end{equation}
The family $B(u)$ of sections of $\bigwedge^{1,0}C\otimes\g_{\mathcal P}^\ast$ naturally determines a section of $\pi_C^\ast\bigwedge^{1,0}C\otimes\g_{\pi_C^\ast \mathcal P}^\ast$, which we will denote by the same symbol $B$.
The following theorem is a direct consequence of our definitions of $B$ and $A$, with corresponding curvature $F_A$.
\begin{theorem} \label{thm: BF Lagrangian}
The gauged multiform action \eqref{gauged Hitchin action} for Hitchin's system is equivalent to the multiform action for $3$d mixed BF theory on $C \times \RR^n$ for the pair $(B, A)$ with a type B line defect along each coordinate $t^i$ determined by the Hitchin Hamiltonian $H_i$ given in \eqref{Hitchin Hamiltonians}, namely
\begin{equation} \label{3d BF action}
S_\Gamma[B, A, t] = \frac{1}{2 \pi i}\int_{C\times\Gamma} \big\langle B, F_A \big\rangle - \int_0^1 H_i\big( B(u(s)) \big) \frac{\d t^i}{\d s} \d s\,.
\end{equation}
for an arbitrary curve $\Gamma : (0,1) \to \RR^n$, $s \mapsto u(s)$.
\end{theorem}
The proof relies on the simple observation that, when pulled back along the one-dimensional curve $\Gamma$, only the first term in the expression \eqref{FA def} of the curvature is nonzero. That being said, the significance of Theorem \ref{thm: BF Lagrangian} is that it shows how the multiform version of $3$d mixed BF theory with type B defect is derived from our procedure of gauging a natural Lagrangian multiform on a cotangent bundle, applied to the Hitchin setup. As shown below, the similar derivation with the inclusion of marked points in the Hitchin picture corresponds to the inclusion of so-called type A defects in the $3$d mixed BF (multiform) picture.

In this new interpretation of the action \eqref{gauged Hitchin action}, local (gauge) transformations $g:\RR^n\to\mathcal{G}$ play the role of bundle automorphisms of $\pi_C^\ast {\mathcal P}$, with local expressions $g_I(z_I,\bar{z}_I,u)$ in the trivialisations over $U_I\times\RR^n$.  From \eqref{G action on M}, \eqref{G act on B}, \eqref{G act on Ai} and \eqref{A definition} it follows that these bundle automorphisms act on $A$ and $B$ in the expected way:
\begin{align}
A^I &\mapsto {}^{g_I}A^I=g_IA^Ig_I^{-1}- \bar\partial g_{I} g_{I}^{-1} - \d_{\RR^n} g_{I} g_{I}^{-1} \,,\\
B^I &\mapsto \Ad_{g_I}B^I\,.
\end{align}

\begin{remark}
The partial connection $A$ on $\pi_C^\ast {\mathcal P}$ was defined using the local trivialisations of $\pi_C^\ast {\mathcal P}$ obtained canonically from the local trivialisations of ${\mathcal P}$. In particular, in such a local trivialisation, the transition functions $g_{IJ}$ are independent of $u^j$ and $A$ transforms as in \eqref{partial connection compat} but without the last term. However, the resulting action \eqref{3d BF action} is well-defined in any local trivialisation of $\pi_C^\ast {\mathcal P}$, including those for which the transitions functions do depend on $u^j$, in which case $A$ transforms as in \eqref{partial connection compat} with the last term present. Indeed, with $A$ transforming as in \eqref{partial connection compat}, the curvature $F_A$ defines a section of $\bigwedge^2(C \times \RR^n) \otimes \g_{\pi_C^\ast {\mathcal P}}$ with $\g_{\pi_C^\ast {\mathcal P}}$ the vector bundle associated to $\pi_C^\ast {\mathcal P}$ in the adjoint representation and hence $(B, F_A)$ is a well-defined global $3$-form on $C \times \RR^n$. We will exploit such a change of local trivialisation leading to $u^j$-dependent transition functions later on in \S\ref{sec:unifyingmultiform} to prove Theorem \ref{thm: Lag 1-form for Hitchin}.
\end{remark}

To derive the equations of motion for \eqref{3d BF action}, we first introduce some notation.
Let $(U_I, z_I)$ be a coordinate chart containing the point $\mathsf q \in U_I$ with coordinate $w = z_I(\mathsf q)$. We denote by $\delta(z_I-w)$ the $2$-dimensional Dirac $\delta$-distribution on $U_I$ at the point $\mathsf q \in U_I$ in the local coordinate $z_I$. It has the defining property
\begin{equation}\label{eq:Dirac-rel}
\int_{U_I} f(z_I) \delta(z_I - w) \d z_I \wedge \d \bar z_I = f(w)
\end{equation}
for any function $f : U_I \to \CC$. Note that the distribution valued $2$-form $\delta(z_I - w) \d z_I \wedge \d \bar z_I$ is invariant under coordinate transformations. Specifically, if $(U_J, z_J)$ is another chart with $\mathsf q \in U_I \cap U_J \neq \emptyset$ and letting $w' \coloneqq z_J(\mathsf q)$, then we have $\delta(z_I - w) \d z_I \wedge \d \bar z_I = \delta(z_J - w') \d z_J \wedge \d \bar z_J$. It will therefore be convenient for later to introduce the notation
\begin{equation} \label{delta coord indep}
\delta_{\mathsf q} \coloneqq \delta\big( z_I - z_I(\mathsf q) \big) \d z_I \wedge \d \bar z_I
\end{equation}
for any $\mathsf q \in C$, which is independent of the coordinate chart $(U_I, z_I)$ used, so long as it contains the point $\mathsf q$.
We will also make use of the following useful fact
\begin{equation}\label{eq:del-Dirac-rel}
\partial_{\bar z_I} \frac{1}{z_I - w} = - 2 \pi i \, \delta (z_I - w).
\end{equation}

Given any polynomial $P : \g^\ast \to \CC$, its gradient $\nabla P : \g^\ast \to \g$ is defined through the first order term in the variation of $P(\phi)$ at any $\phi \in \g^\ast$, namely
\begin{equation} \label{gradient def}
P(\phi + \delta \phi) = P(\phi) + \big\langle \delta\phi, \nabla P(\phi) \big\rangle + O(\delta \phi^2) \,.
\end{equation}
Note, in particular, that for a coadjoint-invariant polynomial $P$, i.e. such that $P(\Ad^\ast_g \phi) = P(\phi)$ for any $g \in G$, it follows from considering $g = e^{h X}$ with $X \in \g$ and $h$ small that $\big\langle {\rm ad}^\ast_X \phi, \nabla P(\phi) \big\rangle = 0$.

\begin{theorem}\label{th_one_form_Hitchin}
The gauged univariational principle applied to the $3$d mixed BF multiform action $S_\Gamma[B,A,t]$ in \eqref{3d BF action} gives rise to a set of equations for the fields $B$, $A''$ and $\mathcal A$. By working in any chart $(U_I, z_I)$ of $C$ and working in terms of the components $B_{z_I}(z_I, \bar z_I, u) $, $A''_{\bar z_I}(z_I, \bar z_I, u) $ and $\mathcal A^I_j(z_I, \bar z_I, u)$ of these fields, these equations take the following form:
\begin{subequations} \label{UV set}
\begin{align}
\partial_{\bar z_I} \widetilde{\mathcal A}^I_i - \partial_{t^i} A''_{\bar z_I} + [A''_{\bar z_I}, \widetilde{\mathcal A}^I_{i}] &= 2 \pi i\,\nabla P_i\big( B_{z_I}(\mathsf q_i) \big) \delta\big( z_I-z_I(\mathsf q_i) \big)\,, \label{UV F}\\
\partial_{\bar z_I} B_{z_I} + {\rm ad}^*_{A''_{\bar z_I}} B_{z_I} &= 0 \label{UV Phi z}\,,\\
\partial_{t^j} B_{z_I}+{\rm ad}^*_{\widetilde{\mathcal A}^I_{j}} B_{z_I} &= 0 \label{UV Phi i}\,,\\
\partial_{t^j} H_i(B)&=0 \label{UV P}
\end{align}
\end{subequations}
where we used the invertibility of the map $(u^j) \mapsto \big( t^i(u) \big)$ to define $\widetilde{\mathcal A}^I_{i} = \frac{\partial u^j}{\partial t^i} \mathcal A^I_{j}$.
The equations associated to any pair of overlapping charts $(U_I, z_I)$ and $(U_J, z_J)$ are compatible on $U_I \cap U_J \neq \emptyset$.

Finally, if $B,A''$ satisfy the stability condition \eqref{inf_freeness} then the following zero-curvature equations also hold
\begin{equation}
\label{ZC_eqs}
\partial_{t^i} \widetilde{\mathcal A}^I_{j} - \partial_{t^j} \widetilde{\mathcal A}^I_{i} + \big[ \widetilde{\mathcal A}^I_{i}, \widetilde{\mathcal A}^I_{j} \big] =0\,.
\end{equation}
\end{theorem}
\begin{proof}
Since the action \eqref{3d BF action} is local in all the fields, and in particular in the fields $B$ and $A$, to work out the equations of motion for the latter at any point $\mathsf p \in C$ it is sufficient to restrict the integration over $C$ in the action to any chart $(U_I, z_I)$ for which $\mathsf p \in U_I$. In other words, to work out the equations of motion for the $\g$-valued fields $B^I = B_{z_I}(z_I, \bar z_I, u) \d z_I$, $A''^I = A''_{\bar z_I}(z_I, \bar z_I, u) \d\bar z_I$ and $\mathcal A^I = \mathcal A^I_{j}(z_I, \bar z_I, u) \d u^j$ it is sufficient to consider the variation of the following action
\begin{align} \label{3d BF action explicit}
S_{I, \Gamma}[B_{z_I}, A''_{\bar z_I}, \mathcal A^I_{j}, t] &= \int_0^1 \Bigg( \frac{1}{2 \pi i} \int_{U_I} \big\langle B_{z_I}(z_I, \bar z_I, u(s)), F^I_{\bar z_Ij}(z_I, \bar z_I, u(s)) \big\rangle\, \d z_I \wedge \d\bar z_I \notag\\
&\qquad\qquad\qquad\qquad - {\sum_{\substack{i = 1\\\mathsf q_i \in U_I}}^n } P_i\big( B_{z_I}(\mathsf q_i, u(s)) \big) \frac{\partial t^i}{\partial u^j} \Bigg) \frac{\d u^j}{\d s} \d s
\end{align}
where we used the fact that the pullback of the curvature $F^I \in \Omega^2(U_I \times \RR^n, \g)$ along the curve $\Gamma : (0,1) \to \RR^n$ is given by
\begin{equation*}
F^I(u(s)) = F^I_{\bar z_Ij}(z_I, \bar z_I, u) \, \d\bar z_I \wedge\frac{\d u^j}{\d s} \d s \in \Omega^2\big( U_I \times (0,1), \g \big) \,,
\end{equation*}
noting that the pullback of $F^I_{ij}(z_I, \bar z_I, u) \d u^i \wedge \d u^j$ vanishes since $\d s \wedge \d s = 0$.

We now work out the variation of \eqref{3d BF action explicit} with respect to all four fields $B_{z_I}(z_I, \bar z_I, u)$, $A''_{\bar z_I}(z_I, \bar z_I, u)$, $\mathcal A^I_{j}(z_I, \bar z_I, u)$ and $t^i(u)$. Firstly, we have
\begin{align*}
&\delta_{B_{z_I}} S_{I, \Gamma}[B_{z_I}, A''_{\bar z_I}, \mathcal A^I_{j}, t]\\
&\quad =  \int_0^1 \int_{U_I} \Bigg( \frac{1}{2 \pi i}\Big\langle \delta B_{z_I}(z_I, \bar z_I, u(s)), F^I_{\bar z_Ij}(z_I, \bar z_I, u(s)) \Big\rangle \\
&\qquad - \sum_{\substack{i = 1\\\mathsf q_i \in U_I}}^n \Big\langle \delta B_{z_I}(z_I, \bar z_I, u(s)), \nabla P_i\big( B_{z_I}(z_I, \bar z_I, u(s)) \big) \Big\rangle \delta\big( z_I - z_I(\mathsf q_i) \big) \frac{\partial t^i}{\partial u^j} \Bigg) \frac{\d u^j}{\d s} \d z_I \wedge \d\bar z_I \wedge \d s
\end{align*}
where in the second term on the right-hand side we used both the defining properties \eqref{eq:Dirac-rel} of the Dirac $\delta$-distribution and \eqref{gradient def} of the gradient of $P_i$. Since the above variation should vanish for all curves $\Gamma : (0,1) \to \RR^n$, $s \mapsto \big( u^j(s) \big)$ and for all variations $\delta^I B_{z_I}(z_I, \bar z_I, u)$ we deduce that
\begin{subequations}
\begin{align} \label{3dBF eq a}
&\partial_{\bar z_I} \mathcal A^I_{j}(z_I, \bar z_I, u(s)) - \partial_{u^j} A''_{\bar z_I}(z_I, \bar z_I, u(s)) + \big[ A''_{\bar z_I}(z_I, \bar z_I, u(s)), \mathcal A^I_{j}(z_I, \bar z_I, u(s)) \big] \notag\\
&\qquad\qquad\qquad\qquad\qquad\qquad = 2 \pi i\sum_{\substack{i = 1\\\mathsf q_i \in U_I}}^n \nabla P_i\big( B_{z_I}(z_I, \bar z_I, u(s)) \big) \delta\big( z_I - z_I(\mathsf q_i) \big) \frac{\partial t^i}{\partial u^j}\,.
\end{align}
The first equation \eqref{UV F} now follows by multiplying both sides by $\frac{\partial u^j}{\partial t^k}$, summing over $j = 1, \ldots, n$ and then relabelling $k$ as $j$, and using the definition of $\widetilde{\mathcal A}^I_{j}$ given in the statement of the theorem. Next, for the variation of \eqref{3d BF action explicit} with respect to $A''_{\bar z_I}(z_I, \bar z_I, u)$ we find
\begin{align*}
&\delta_{A''_{\bar z_I}} S_{I, \Gamma}[B_{z_I}, A''_{\bar z_I}, \mathcal A^I_{j}, t]\\
&\quad = - \int_0^1 \frac{1}{2 \pi i}\int_{U_I} \Big\langle B_{z_I}(z_I, \bar z_I, u(s)),\\
&\qquad\qquad\qquad\qquad \partial_{u^j} \delta A''_{\bar z_I}(z_I, \bar z_I, u(s)) + \big[ \mathcal A^I_{j}(z_I, \bar z_I, u(s)), \delta A''_{\bar z_I}(z_I, \bar z_I, u(s)) \big] \Big\rangle \frac{\d u^j}{\d s} \d z_I \wedge \d\bar z_I \wedge \d s\,.
\end{align*}
Since this should vanish for all curves $\Gamma : (0,1) \to \RR^n$, $s \mapsto \big( u^j(s) \big)$ and all variations $\delta A''_{\bar z_I}(z_I, \bar z_I, u)$ we deduce that \eqref{UV Phi i} must hold, again after multiplying by $\frac{\partial u^j}{\partial t^k}$, summing over $j = 1, \ldots, n$ and relabelling $k$ as $j$, and also using the definition of $\widetilde{\mathcal A}^I_{j}$. And by a very similar computation, varying \eqref{3d BF action explicit} with respect to $\mathcal A^I_{j}(z_I, \bar z_I, u)$ we deduce that \eqref{UV Phi z} must hold.
Finally, varying \eqref{3d BF action explicit} with respect to $t^i(u)$ leads to
\begin{equation*}
\delta_{t^i(u)} S_{I, \Gamma}[B_{z_I}, A''_{\bar z_I}, \mathcal A^I_{j}, t] = - \int_0^1 \bigg( H_i\big( B(u(s)) \big) \frac{\partial \delta t^i}{\partial u^j} \bigg) \frac{\d u^j}{\d s} \d s
\end{equation*}
\end{subequations}
and since this should vanish for all curves $\Gamma : (0,1) \to \RR^n$, $s \mapsto \big( u^j(s) \big)$ and all variations $\delta t^i(u)$ we deduce that \eqref{UV P} must hold.

The derivation of \eqref{ZC_eqs} follows the same reasoning as in the proof of Theorem \ref{thm: gauged univar} for \eqref{flatness}. Using \eqref{UV Phi i}, and with $\widetilde{F}_{ij}^I=\partial_{t^i} \widetilde{\mathcal A}^I_{j} - \partial_{t^j} \widetilde{\mathcal A}^I_{i} + \big[ \widetilde{\mathcal A}^I_{i}, \widetilde{\mathcal A}^I_{j} \big]$, we deduce 
\begin{equation}
\left[\partial_{t^i},\partial_{t^j}\right] B_{z_I}+{\rm ad}^*_{\widetilde{F}_{ij}^I} B_{z_I} = 0\,.
\end{equation}
The first term on the left-hand side vanishes since partial derivatives commute and hence
\begin{equation} \label{Stab 1 for F=0}
{\rm ad}^*_{\widetilde{F}_{ij}^I} B_{z_I} = 0\,.     
\end{equation}
Similarly, using \eqref{UV F}, we deduce
\begin{align} \label{Deriving dF AF zero}
&\left[\partial_{t^i},\partial_{t^j}\right] A''_{\bar z_I}+\partial_{\bar z_I}\widetilde{F}_{ji}^I+\big[ A''_{\bar z_I}, \widetilde{F}_{ji}^I\big] \notag\\
&\qquad\qquad=
 2 \pi i\left(\partial_{t^j}\,\nabla P_i\big( B_{z_I}(\mathsf q_i) \big)+ \big[ \widetilde{\mathcal A}^I_{j},\nabla P_i\big( B_{z_I}(\mathsf q_i) \big) \big] \right) \delta\big( z_I-z_I(\mathsf q_i) \big) \notag\\
&\qquad\qquad\qquad - 
 2 \pi i\left(\partial_{t^i}\,\nabla P_j\big( B_{z_I}(\mathsf q_j) \big)+ \big[ \widetilde{\mathcal A}^I_{i},\nabla P_j\big( B_{z_I}(\mathsf q_j) \big) \big] \right) \delta\big( z_I-z_I(\mathsf q_j) \big)\,.
\end{align}
Now note that \eqref{UV Phi i} implies 
\begin{equation*}
\partial_{t^j} \nabla P_i\big( B_{z_I}(\mathsf q_i) \big)+\Big[{\widetilde{\mathcal A}^I_{j}},\nabla P_i\big( B_{z_I}(\mathsf q_i) \big)\Big] = 0
\end{equation*}
so that both terms on the right-hand side of \eqref{Deriving dF AF zero} vanish. The first term on the left-hand side of \eqref{Deriving dF AF zero} also vanishes since partial derivatives commute and \eqref{Deriving dF AF zero} reduces to hence
\begin{align} \label{Stab 2 for F=0}
\partial_{\bar z_I}\widetilde{F}_{ji}^I+[A''_{\bar z_I}, \widetilde{F}_{ji}^I]=0\,.
\end{align}
Equation \eqref{ZC_eqs} now follows from \eqref{Stab 1 for F=0} and \eqref{Stab 2 for F=0} using the stability condition \eqref{inf_freeness}.
\end{proof}

\begin{remark}
In standard Lagrangian multiform terminology, the first three equations \eqref{UV F}-\eqref{UV Phi i} are the multitime Euler-Lagrange equations and the last equation \eqref{UV P} is (related to) the closure relation. Note that here the latter holds automatically as a consequence of \eqref{UV Phi i} and $G$-invariance of $P_i$, thus confirming that we indeed have a Lagrangian multiform ($1$-form here) for the Hitchin system. To see this, note that
\begin{equation}
\partial_{t^j} P_i\big( B_{z_I} \big) = \big\langle \partial_{t^j} B_{z_I},\nabla P_i( B_{z_I}) \big\rangle = \Big\langle \partial_{t^j} B_{z_I} + {\rm ad}^*_{{\mathcal A}^I_{j}} B_{z_I}\,,\nabla P_i( B_{z_I}) \Big\rangle = 0\,,
\end{equation}    
where we used the infinitesimal coadjoint-invariance of $P_i$ in the second equality, see the argument just after \eqref{gradient def}, and in the last equality we have used \eqref{UV Phi i}.
\end{remark}

\subsection{Adding type A defects at marked points on \texorpdfstring{$C$}{C}} \label{sec: adding punctures}

We will now extend the construction from \S\ref{sec: Lag for Hitchin} and \S\ref{sec: Lag for Hitchin mod G} in order to obtain a variational principle for Hitchin's system on a Riemann surface $C$ with $N \in \ZZ_{\geq 1}$ distinct marked points $\{ \mathsf p_\alpha \}_{\alpha=1}^N$, in particular to allow the Higgs field $B \in T^\ast_{A''} \mathcal M$ to have simple poles at these marked points. This requires introducing extra structure at the marked points $\mathsf p_\alpha$ which will encode the residues of $B$.

To each marked point $\mathsf p_\alpha$ we will associate a coadjoint orbit $\mathcal{O}_\alpha$ in the fibre of $\g^\ast_{\mathcal P}$ over $\mathsf p_\alpha$.  This coadjoint orbit $\mathcal{O}_\alpha$ is determined by an element $\Lambda_\alpha\in\g^\ast$.  In a local trivialisation of ${\mathcal P}$ over $U_I\ni \mathsf p_\alpha$, an element $X\in \mathcal{O}_\alpha$ is represented by
\begin{equation}\label{coadjoint orbit}
    X^I=\mathrm{Ad}^\ast_{\varphi_\alpha^I}\Lambda_\alpha\quad\text{ for some }\varphi_\alpha^I\in G.
\end{equation}
If $\mathsf p_\alpha\in U_I\cap U_J$ then we must have that $X^I=\mathrm{Ad}^\ast_{g_{IJ}(\mathsf p_\alpha)}X^J$, so we require that $\varphi_\alpha^I=g_{IJ}(\mathsf p_\alpha)\varphi_\alpha^J$.  This means that $\varphi_\alpha^I$ represents an element $\varphi_\alpha$ of the fibre of ${\mathcal P}$ over ${\mathsf p}_\alpha$.

The Kostant--Kirillov symplectic form on $\mathcal{O}_\alpha$ is given by $\omega_\alpha(X,Y) \coloneqq \langle\Lambda_\alpha,[X,Y]\rangle$ for every $X,Y\in\mathfrak{g}$ representing tangent vectors in $T_{\Lambda_\alpha} \mathcal{O}_\alpha$.  The Kostant-Kirillov form $\omega_\alpha$ can equally be described as follows. Consider the map
\begin{equation}
G\longrightarrow \mathcal O_\alpha \,, \qquad \varphi_\alpha^I \longmapsto \Ad^*_{\varphi_\alpha^I} \Lambda_\alpha \,.
\end{equation}
The pullback of $\omega_\alpha$ under this map, which by abuse of notation we also denote by $\omega_\alpha$, is the $2$-form given by $\omega_\alpha = \d\big\langle \Lambda_\alpha,(\varphi^I_\alpha)^{-1} \d\varphi^I_\alpha \big\rangle$. Note that this is exact, and that both $\omega_\alpha$ and $\big\langle \Lambda_\alpha,(\varphi^I_\alpha)^{-1} \d\varphi^I_\alpha \big\rangle$ are invariant under changes of local trivialisation.  We will denote the latter by $\big\langle \Lambda_\alpha,(\varphi_\alpha)^{-1} \d\varphi_\alpha \big\rangle$ to emphasise this.

Instead of working on the phase space $(T^\ast \mathcal M, \omega)$ as in \S\ref{sec: Lag for Hitchin} and \S\ref{sec: Lag for Hitchin mod G}, we will now work with the extended phase space
\begin{equation} \label{extended phase space}
\bigg( T^\ast \mathcal M \times \prod_{\alpha = 1}^N \mathcal{O}_\alpha\,, \; \omega + \sum_{\alpha = 1}^N \omega_\alpha \bigg).
\end{equation}
We extend the tautological $1$-form $\alpha$ on $T^\ast \mathcal M$, given in \eqref{tauto}, to a $1$-form on the extended phase space \eqref{extended phase space}, which we still denote by $\alpha$ and is given explicitly by
\begin{equation}
\label{ext_tauto}
\alpha_{(B, A'', (\varphi_\alpha))} \coloneqq \frac{1}{2 \pi i}\int_C\langle B,\delta A''\rangle + \sum_{\alpha = 1}^N \big\langle \Lambda_\alpha, (\varphi_\alpha)^{-1} \d \varphi_\alpha \big\rangle\,.
\end{equation}
To evaluate the second term, we choose (for each $\alpha$) an open set $U_I$ containing ${\mathsf p}_\alpha$, and evaluate $\big\langle \Lambda_\alpha, (\varphi_\alpha^I)^{-1} \d \varphi_\alpha^I \big\rangle$, where $\varphi_\alpha^I\in G$ is the representative of $\varphi_\alpha$ in our chosen trivialisation.

The action of the group $\mathcal{G}$ of bundle automorphisms on $T^\ast \mathcal M \times \prod_{\alpha = 1}^N \mathcal{O}_\alpha$ is as follows, cf. \eqref{cal G action on T*M},
\begin{align} \label{group_action_marked}
\mathcal G \times T^\ast \mathcal M \times \prod_{\alpha = 1}^N \mathcal{O}_\alpha &\longrightarrow T^\ast \mathcal M \times \prod_{\alpha = 1}^N \mathcal{O}_\alpha \notag\\
\Big( g, \big(B, A'', (\varphi_\alpha^I) \big) \Big) &\longmapsto g\cdot \big( B, A'', (\varphi^I_\alpha) \big) = \Big( \Ad^*_g B, {}^g A'', \big( g^I(\mathsf p_\alpha) \varphi^I_\alpha \big) \Big)\,.
\end{align}
This induces the following infinitesimal left action
\begin{equation}
\label{global_inf_action_X}
\delta_X B = {\rm ad}^*_X B\,, \quad \delta_X A'' = - \bar{\partial}^{A''} X \,, \quad \delta_X \big( \Ad_{\varphi^I_\alpha}^\ast \Lambda_\alpha \big) = {\rm ad}^*_{X^I(\mathsf p_\alpha)} \big( \Ad_{\varphi_\alpha^I}^\ast \Lambda_\alpha \big)
\end{equation}
of $X \in \mathfrak G$ on $\big( B, A'', (\varphi_\alpha^I) \big) \in T^\ast \mathcal M \times \prod_{\alpha = 1}^N \mathcal{O}_\alpha$.
\begin{remark}
The phase space \eqref{extended phase space} can be described more invariantly as follows.  Let $G_{\Lambda_\alpha}$ be the stabiliser of $\Lambda_\alpha\in\g^\ast$.  Then $\Ad^\ast_{\varphi_\alpha^I}\Lambda_\alpha$ corresponds to a left coset $\varphi_\alpha^IG_{\Lambda_\alpha}\in G/G_{\Lambda_\alpha}$.  This represents an element of the quotient $\mathcal P_{{\mathsf p}_\alpha}/G_{\Lambda_\alpha}$ in our chosen local trivialisation.  An element of $\mathcal P_{{\mathsf p}_\alpha}/G_{\Lambda_\alpha}$ is usually called a reduction of the structure group of $\mathcal P$ at ${\mathsf p}_\alpha$ to $G_{\Lambda_\alpha}$.  So our phase space is the product of the cotangent bundle of the space of holomorphic structures and the space of reductions of structure group to $G_{\Lambda_\alpha}$ at the subset of marked points $\{ {\mathsf p}_\alpha \}_{\alpha = 1}^N$.
\end{remark}

We proceed with generalising the construction of \S\ref{sec: Lag for Hitchin}. In the present setup, the Lagrangian $1$-form \eqref{L} and action \eqref{ungauged action} take the form
\begin{equation*}
\Lag =\frac{1}{2 \pi i}\int_C\langle B,\delta A''\rangle + \sum_{\alpha=1}^N\langle\Lambda_\alpha,\varphi_\alpha^{-1}\d \varphi_\alpha\rangle - H_i\big( B, A'', (\varphi_\alpha) \big) \d t^i\,,    
\end{equation*}
for some Hamiltonians $H_i : T^\ast \mathcal M \times \prod_{\alpha=1}^N \mathcal O_\alpha \to \CC$.
Its pullback along an arbitrary immersion $\Sigma : \RR^n \to T^\ast \mathcal M \times \prod_{\alpha=1}^N \mathcal O_\alpha \times \RR^n$, $\Sigma(u) = \big( B(u), A''(u), (\varphi_\alpha(u)), t(u) \big)$ is then
\begin{equation*}
\Sigma^\ast \Lag = \frac{1}{2 \pi i}\int_C \langle B,\d_{\RR^n} A'' \rangle + \sum_{\alpha=1}^N \langle \Lambda_\alpha,\varphi_\alpha^{-1}\d_{\RR^n} \varphi_\alpha \rangle - H_i\big( B, A'', (\varphi_\alpha) \big) \d_{\RR^n} t^i\,,
\end{equation*}
so that for any parametrised curve $\Gamma : (0,1) \to \RR^n$, $s \mapsto \big( u^j(s) \big)$ the corresponding action is given by $S_\Gamma[\Sigma] \coloneqq \int_0^1 (\Sigma \circ \Gamma)^\ast \Lag$, which explicitly reads
\begin{align} \label{ungauged action Phi A Lambda}
S_\Gamma[\Sigma] &= \int_0^1\Bigg(\! -\frac{1}{2 \pi i}\int_C \big\langle B(u),\partial_{u^j} A''(u) \big\rangle + \sum_{\alpha=1}^N \big\langle \Lambda_\alpha,\varphi_\alpha(u)^{-1} \partial_{u^j} \varphi_\alpha(u) \big\rangle \notag\\
&\qquad\qquad\qquad\qquad\qquad\qquad - H_i\Big( B(u(s)), A''(u(s)), \big( \varphi_\alpha(u(s)) \big) \Big) \frac{\partial t^i}{\partial u^j}\Bigg) \frac{\d u^j}{\d s} \d s\,.
\end{align}
Note that this action, with the Hamiltonians $H_i$ as in \eqref{Hitchin Hamiltonians}, bears some resemblance to the one in \cite[\S4.1.3]{O}. However, the latter differs from \eqref{ungauged action Phi A Lambda} in two key respects. First, it is given by a sum of individual actions on $\RR \times C$ for each time flow of the Hitchin system, rather than by a family of actions on $\RR^n \times C$. As recalled in \S\ref{sec: univar principle}, the latter formulation is crucial to encode integrability through the Euler–Lagrange equation \eqref{univar EL eq b}. Second, the action in \cite[\S4.1.3]{O} is treated as a gauge theory with ``gauge group'' $\mathcal G$. The invariance of the action \eqref{ungauged action Phi A Lambda} under the action of the group $\mathcal G$ on the extended phase space \eqref{extended phase space} is indeed the content of the next proposition, which is an extension of Proposition \ref{prop: Hi invariance Hitchin} to the case with type A defects. However, recall from \S\ref{sec: Lag for Hitchin mod G}, see also the second part of Remark \ref{rem: gauge in Hitchin}, that unlike \cite{O} we are treating $\mathcal G$ as a \emph{global} symmetry group. The \emph{gauging} of this global $\mathcal G$-symmetry is what led to the $3$d mixed BF action with type B defects in Theorem \ref{thm: BF Lagrangian}; see also Theorem \ref{thm: 3d BF with type A defects 1} below for the analogue in the presence of type A defects.

\begin{proposition} \label{prop: Hi invariance Hitchin type A}
The action \eqref{ungauged action Phi A Lambda} is invariant under the action of $\mathcal G$ on $T^\ast \mathcal M \times \prod_{\alpha=1}^N \mathcal O_\alpha \times \RR^n$ given by \eqref{group_action_marked} and extended trivially to the $\RR^n$ factor if and only if each $H_i$ for $i=1, \ldots, n$ is invariant under the group action, i.e.
\begin{equation} \label{Hi invariance prop 2}
H_i\Big( \Ad^*_g B, {}^g A'', \big( g(\mathsf p_\alpha) \varphi_\alpha \big) \Big) = H_i\big( B, A'', (\varphi_\alpha) \big)
\end{equation}
for any $\big( B, A'', (\varphi_\alpha) \big) \in T^\ast \mathcal M \times \prod_{\alpha=1}^N \mathcal O_\alpha$ and $g \in \mathcal G$. Moreover,
the value of the associated moment map $\mu : T^\ast \mathcal M \times \prod_{\alpha=1}^N \mathcal O_\alpha \to \mathfrak{G}^\ast$ on an infinitesimal bundle morphism $X \in \mathfrak G$ is given by
\begin{equation} \label{moment_map_Hitchin type A}
\mu_{(B, A'',(\varphi_\alpha))}(X) = \frac{1}{2 \pi i}\int_C \big\langle B, \bar\partial^{A''} X \big\rangle  +\sum_{\alpha=1}^N \int_C\big\langle \Ad^\ast_{\varphi_\alpha} \Lambda_\alpha, X \big\rangle\delta_{\mathsf p_\alpha} \,.
\end{equation}
\begin{proof}
For the invariance of the action under the group action \eqref{group_action_marked}, by Proposition \ref{prop: Hi invariance Hitchin} we only need to check the invariance of the new kinetic terms for the $\varphi_\alpha$ but this is immediate.

The extra term in the moment map \eqref{moment_map_Hitchin type A} compared to the expression in \eqref{moment_map_Hitchin} comes from generalising the argument in the last part of the proof of Proposition \ref{prop: Hi invariance Hitchin} to the present case and, again, the new contribution comes from the coadjoint orbit terms in the kinetic part. 
\end{proof}
\end{proposition}

Following the discussion in Remark \ref{rem: mu well defined}, we can use Stokes's theorem to rewrite the expression \eqref{moment_map_Hitchin type A} for the moment map as
\begin{equation*}
\mu_{(B, A'',(\varphi_\alpha))}(X) = \int_C \bigg\langle \frac{1}{2 \pi i}\bar\partial^{A''} B + \sum_{\alpha=1}^N \Ad^\ast_{\varphi_\alpha} \Lambda_\alpha \delta_{\mathsf p_\alpha} ,X \bigg\rangle\,.
\end{equation*}
In particular, we see that the vanishing of $\mu$ at any $\big( B, A'', (\varphi_\alpha) \big) \in T^\ast \mathcal M \times \prod_{\alpha=1}^N \mathcal O_\alpha$ corresponds to the condition
\begin{equation} \label{mu_with_punctures}
\bar{\partial}^{A''} B = -2 \pi i\sum_{\alpha=1}^N \Ad^\ast_{\varphi_\alpha} \Lambda_\alpha \delta_{\mathsf p_\alpha}\,,
\end{equation}
which reads locally as 
\begin{equation}
\label{local_mu_constraint}
    \partial_{\bar z_I} B_{z_I} + {\rm ad}^*_{A''_{\bar z_I}} B_{z_I} = 2 \pi i\sum_{\substack{\alpha=1\\ \mathsf p_\alpha \in U_I}}^N \Ad^\ast_{\varphi^I_\alpha} \Lambda_\alpha \delta\big( z_I - z_I(\mathsf p_\alpha) \big)\,.
\end{equation}

To produce a (non trivial) dynamical theory we must choose Hamiltonian functions $H_i$ on the symplectic manifold $T^\ast\M\times \prod_\alpha\mathcal{O}_\alpha$.  We choose these in a similar way as in the case without marked points, with Hamiltonians given in \eqref{Hitchin Hamiltonians} in terms of invariant polynomials $P_i$ on $\g^\ast$ and points ${\mathsf q}_i$. This ensures in particular their invariance property as in Proposition \ref{prop: Hi invariance Hitchin type A}. We will specify the number of points ${\mathsf q}_i$ more precisely when we look at specific examples in section \ref{sec: examples}.

Finally, we proceed with generalising the construction of \S\ref{sec: Lag for Hitchin mod G} to the case with marked points.
Recall that we introduced the $\g$-valued gauge field ${\cal A}$ along $\RR^n$ and $\langle \Sigma^\ast\mu,\mathcal A \rangle$ was added to the pullback $\Sigma^\ast \Lag$ of the Lagrangian $1$-form \eqref{L_Phi_A} to eventually produce the gauged action \eqref{gauged Hitchin action}.
In the present case we find the gauged action
\begin{multline}\label{gauged Hitchin action marked points}
S_\Gamma[\Sigma,\mathcal{A},t] = 
\int_0^1 \Bigg( -\frac{1}{2 \pi i}
\int_C \left\langle B(u), \frac{\partial A''}{\partial u^j}- \bar{\partial}^{A''}\mathcal{A}_j \right\rangle+\sum_{\alpha=1}^N \Big\langle \Lambda_\alpha, \varphi_\alpha^{-1} \big(\partial_{u^j} + \mathcal A_j(\mathsf p_\alpha) \big) \varphi_\alpha \Big\rangle\\
- H_i\big( B(u) \big)\frac{\partial t^i}{\partial u^j} \Bigg) \frac{\d u^j}{\d s} \d s \,.
\end{multline}
Rewriting this in terms of the connection $A=A''+\mathcal{A}$ over $C\times\RR^n$ yields the following analogue of Theorem \ref{thm: BF Lagrangian}.

\begin{theorem} \label{thm: 3d BF with type A defects 1}
The gauged multiform action \eqref{gauged Hitchin action marked points} for Hitchin's system with marked points is given by the multiform action for $3$d mixed BF theory on $C \times \RR^n$ for the collection of fields $\big( B, A, (\varphi_\alpha) \big)$ with a type B line defect along each coordinate $t^i$ determined by the Hitchin Hamiltonian $H_i$ defined as in \eqref{Hitchin Hamiltonians} and a type A line defect at each marked point $\mathsf p_\alpha$, namely 
\begin{align}\label{S Gamma 2}
S_\Gamma[B, A, (\varphi_\alpha), t] &= \frac{1}{2 \pi i}\int_{C\times \Gamma} \big\langle B,F_A \big\rangle + \sum_{\alpha=1}^N \int_0^1\Big\langle\Lambda_\alpha,\varphi_\alpha^{-1}\big(\partial_{u^j} + \mathcal A_j(\mathsf p_\alpha) \big)\varphi_\alpha \Big\rangle \frac{\d u^j}{\d s}\d s \notag\\
&\qquad\qquad\qquad\qquad\qquad\qquad - \int_0^1H_i\big(B(u(s)) \big) \frac{\d t^i}{\d s} \d s \,,
\end{align}
for an arbitrary curve $\Gamma : (0,1) \to \RR^n$, $s \mapsto u(s)$, where $F_A$ is the curvature introduced in \eqref{FA def}.
\end{theorem}

Finally, we also have the following analogue of Theorem \ref{th_one_form_Hitchin}.

\begin{theorem} \label{thm: 3d BF with type A defects 2}
The gauged univariational principle applied to the $3$d mixed BF multiform action $S_\Gamma[B,A, (\varphi_\alpha), t]$ in \eqref{S Gamma 2} yields a set of equations for the fields $B$, $A$ and  $\varphi_\alpha$ for $\alpha =1,\ldots, N$. Working in any chart $(U_I, z_I)$ of $C$ and in terms of the components $B_{z_I}(z_I, \bar z_I, u) $, $A''_{\bar z_I}(z_I, \bar z_I, u)$, $\mathcal A^I_{j}(z_I, \bar z_I, u)$ of the various fields, these equations take the following form:
\begin{subequations}
\label{eqs_gauged_action}
\begin{align}
\partial_{\bar z_I} \widetilde{\mathcal A}^I_{i} - \partial_{t^i} A''_{\bar z_I} + [A''_{\bar z_I}, \widetilde{\mathcal A}^I_{i}] &= 2 \pi i\,\nabla P_i\big( B_{z_I}(\mathsf q_i) \big) \delta\big( z_I-z_I(\mathsf q_i) \big)\,, \label{F_B}\\
\partial_{\bar z_I} B_{z_I} + {\rm ad}^*_{A''_{\bar z_I}} B_{z_I} &= 2 \pi i\sum_{\substack{\alpha=1\\ \mathsf p_\alpha \in U_I}}^N \Ad^\ast_{\varphi^I_\alpha} \Lambda_\alpha \delta\big( z_I - z_I(\mathsf p_\alpha) \big) \label{A_defect_modif}\,,\\
\partial_{t^j} B_{z_I}+{\rm ad}^*_{\widetilde{\mathcal A}^I_{j}} B_{z_I} &= 0 \label{eq:B_z_punctures}\,,\\
\partial_{t^j} H_i(B)&=0 \,,\\
\partial_{t^j} \big( \Ad^\ast_{\varphi^I_\alpha} \Lambda_\alpha \big) + {\rm ad}^*_{\widetilde{\mathcal A}^I_{ j}(\mathsf p_\alpha)} \Ad^\ast_{\varphi^I_\alpha} \Lambda_\alpha &=0\label{eq:L_beta}\,,
\end{align}
\end{subequations}
where we have used the invertibility of the map $(u^j) \mapsto \big( t^i(u) \big)$ to define $\widetilde{\mathcal A}^I_{i} = \frac{\partial u^j}{\partial t^i} \mathcal A^I_{j}$.
The equations associated to any pair of overlapping charts $(U_I, z_I)$ and $(U_J, z_J)$ are compatible on $U_I \cap U_J \neq \emptyset$.
The following zero-curvature equations also hold
\begin{equation}\label{eq:zce}
\partial_{t^i} \widetilde{\mathcal A}^I_{j} - \partial_{t^j} \widetilde{\mathcal A}^I_{i} + \big[ \widetilde{\mathcal A}^I_{i}, \widetilde{\mathcal A}^I_{j} \big] =0\,.
\end{equation}
\end{theorem}
\begin{proof}
The computation for the variation of the action with respect to the fields $B^I$, $A^I$,  and $t^i$ is exactly as in the proof of Theorem \ref{th_one_form_Hitchin} and so is the derivation of \eqref{eq:zce}. The new terms in the action only modify the equation obtained by varying ${\cal A}^I$ and this gives the right-hand side of \eqref{A_defect_modif}. Finally, the variation with respect to $\varphi_\beta$ gives
\begin{align*}
\delta_{\varphi_\beta} S_\Gamma[B, A, (\varphi_\alpha), t] &= \delta_{\varphi_\beta} \int_0^1\sum_{\alpha=1}^N \Big\langle\Lambda_\alpha,\varphi_\alpha^{-1}\big(\partial_{u^j} + \mathcal A_{j}(\mathsf p_\alpha) \big)\varphi_\alpha \Big\rangle \frac{\d u^j}{\d s} \d s\\
&= \int_0^1\Big\langle\Lambda_\beta,-(\varphi^I_\beta)^{-1}\delta \varphi^I_\beta(\varphi^I_\beta)^{-1}\big( \partial_{u^j} + \mathcal A^I_{j}(\mathsf p_{\beta}) \big) \varphi^I_\beta\\
&\qquad\qquad\qquad\qquad\qquad +(\varphi^I_\beta)^{-1}\big( \partial_{u^j}+ \mathcal A^I_{j}(\mathsf p_{\beta}) \big) \delta\varphi^I_\beta \Big\rangle \frac{\d u^j}{\d s} \d s\\
&= \int_0^1\Big\langle\Ad^\ast_{\varphi^I_\beta} \Lambda_\beta, \big[ \partial_{u^j} + \mathcal A^I_{j}(\mathsf p_\beta),\delta\varphi^I_\beta(\varphi^I_\beta)^{-1}\big] \Big\rangle \frac{\d u^j}{\d s} \d s\\
&= - \int_0^1\Big\langle \big( \partial_{u^j} + {\rm ad}^\ast_{ \mathcal A^I_{j}(\mathsf p_\beta)} \big)\big( \Ad^\ast_{\varphi^I_\beta} \Lambda_\beta \big), \delta\varphi^I_\beta(\varphi^I_\beta)^{-1} \Big\rangle \frac{\d u^j}{\d s}\d s
\end{align*}
so that we obtain \eqref{eq:L_beta}.
\end{proof}

\section{Hitchin system in Lax form and its variational formulation} \label{sec: Hitchin holomorphic}

In \S\ref{sec: 3d BF} we showed that Hitchin's completely integrable system on the symplectic quotient $\mu^{-1}(0) / \mathcal G$ is described variationally by the action for $3$d mixed BF theory with type B line defects associated with each of the Hamiltonians $H_i$ and corresponding times $t^i$ for $i=1,\ldots, n$, see Theorem \ref{thm: BF Lagrangian}. We further extended this result to Hitchin's integrable system with marked points in \S\ref{sec: adding punctures} by adding type A line defects at each marked point in $C$ to the action, see Theorem \ref{thm: 3d BF with type A defects 1}.

The purpose of this section is twofold. First, in \S\ref{sec: Lax matrices} we will show that the set of variational equations derived in Theorem \ref{thm: 3d BF with type A defects 2} encode the hierarchy of equations of the Hitchin system in the standard Lax form, namely $\partial_{t^i}L=[M_i,L]$ for a meromorphic Lax matrix $L$ and a collection of meromorphic Lax matrices $M_i$ associated with each time $t^i$ in the hierarchy for $i=1, \ldots, n$. To our knowledge, the earliest connection between the Hitchin system and the Lax formalism appeared in \cite{K,LeOZ1}. Second, in \S\ref{sec:unifyingmultiform} we will show that this hierarchy of Lax equations for the Hitchin system is itself variational and obtain the corresponding Lagrangian $1$-form/action, which we refer to as the \emph{unifying action}, directly from the $3$d mixed BF multiform action \eqref{S Gamma 2}. The key is to make use of the isomorphism \eqref{Hitchin phase space intro} mentioned in the introduction in order to pass from the description of the Hitchin system on $\mu^{-1}(0) / \mathcal G$ in terms of smooth $\g$-valued $(0,1)$-connections $A'' \in \mathcal M$ and smooth Higgs fields $B \in T_{A''}^\ast \mathcal M$, as in \S\ref{sec: Lag for Hitchin mod G} and \S\ref{sec: adding punctures}, to a description of the Hitchin system on $T^\ast \text{Bun}_G(C)$ in terms of holomorphic transition functions and holomorphic Higgs fields. We will thus establish the right-hand side of the following commutative diagram
\begin{equation*}
\begin{tikzcd}[row sep=30, column sep=60]
S_\Gamma[B,A, (\varphi_\alpha), t] \arrow[r, "\text{holomorphic}", "\text{description}"'] \arrow[d, "\text{univariational}"' pos=0.35, "\text{principle}"' pos=0.65] & \text{Unifying action} \arrow[d, "\text{univariational}" pos=0.35, "\text{principle}" pos=0.65] \\
\text{Variational equations \eqref{eqs_gauged_action}} \arrow[r, "\text{holomorphic}", "\text{description}"'] & \partial_{t^i}L=[M_i,L] 
\end{tikzcd}
\end{equation*}

\subsection{Hitchin system in Lax form} \label{sec: Lax matrices}

The key observation, following \cite{VW}, is that \eqref{eq:B_z_punctures} takes the form of a hierarchy of Lax equations for Lax matrices $B_{z_I}$ and $\widetilde{\mathcal A}^I_i$ defined on the coordinate patch $(U_I, z_I)$ that become meromorphic \textit{when solving \eqref{F_B}-\eqref{A_defect_modif} in a local trivialisation such that $A''_{\bar z_I} = 0$}. The compatibility of this hierarchy of Lax equations is ensured by \eqref{eq:zce}. We now spell this out in detail.

As usual, in order to bring \eqref{eq:B_z_punctures} to the standard Lax form, we will identify $\g^\ast$ with $\g$ using a nondegenerate invariant bilinear form on $\g$. We still denote the latter by $\langle~,~\rangle$ (like the pairing between $\g^\ast$ and $\g$) hoping that it will not confuse the reader. This allows us to identify adjoint and coadjoint actions. Moreover, from now on we only consider matrix Lie groups and Lie algebras.

Next, we pick a point $\mathsf p \in C$ distinct from all the marked points $\mathsf p_\alpha$ with $\alpha = 1,\ldots, N$ and $\mathsf q_i$ with $i = 1,\ldots, n$. Pick a neighbourhood $U_0$ of $\mathsf p$ not containing any of these other points and equipped with a local coordinate $z_0 : U_0 \to C$. We also let $U_1 \coloneqq C \setminus \{ \mathsf p \}$ so that $\{ U_0, U_1 \}$ forms an open cover of $C$. Let us stress that for $g \geq 1$ the open $U_1 \subset C$ is not a coordinate chart since in general it cannot be equipped with a holomorphic coordinate $U_1 \to \CC$. However, we can further refine the cover by choosing a holomorphic atlas $\{ (U_I, z_I) \}_{I \in \mathcal I}$ for $U_1$, where $\mathcal I$ is any indexing set not containing $0$ and $1$ so that we can also use the notation $U_I$ for $I \in \{ 0, 1 \}$. In other words, $U_I$ for $I \in \mathcal I$ are open subsets of $U_1$ such that $\bigcup_{I \in \mathcal I} U_I = U_1$ and equipped with coordinates $z_I : U_I \to \CC$. 

For the purpose of describing the smooth principal $G$-bundle $\pi : {\mathcal P} \to C$ with a holomorphic structure specified by $A'' \in \mathcal M$, the open cover $\{U_0, U_1\}$ of $C$ suffices. Indeed, since $U_1$ is a non-compact Riemann surface and $G$ is semi-simple, we can trivialise the holomorphic principal $G$-bundle $({\mathcal P}, A'')$ over $U_1$. Since we can also trivialise $({\mathcal P}, A'')$ over $U_0$, we obtain local trivialisations for the holomorphic principal $G$-bundle $({\mathcal P}, A'')$ relative to the open cover $\{ U_0, U_1 \}$ of $C$. More importantly for us, we also have local trivialisations of the pullback bundle $\pi_C^\ast {\mathcal P} \cong {\mathcal P} \times \RR^n$, see \S\ref{sec: Lag for Hitchin mod G}, which we write as
\begin{subequations} \label{local triv examples}
\begin{align}
\label{loc triv 0} \psi_0 : \pi^{-1}(U_0 \times \RR^n) &\overset{\cong}\longrightarrow U_0 \times \RR^n \times G \,, \quad p \longmapsto (\pi(p), f_0(p)) \,,\\
\label{loc triv 1} \psi_1 : \pi^{-1}(U_1 \times \RR^n) &\overset{\cong}\longrightarrow U_1 \times \RR^n \times G \,, \quad p \longmapsto (\pi(p), f_1(p))
\end{align}
\end{subequations}
relative to the open cover $\{ U_0 \times \RR^n, U_1 \times \RR^n \}$ of $C \times \RR^n$. Let $g_{01} : (U_0 \cap U_1) \times \RR^n \to G$ be the smooth transition function on the overlap $(U_0 \cap U_1) \times \RR^n \neq \emptyset$ so that (with $(x,t)=\pi(p)\in (U_0 \cap U_1) \times \RR^n$)
\begin{equation*}
\psi_0 \circ \psi_1^{-1} : (U_0 \cap U_1) \times \RR^n \times G \longrightarrow (U_0 \cap U_1) \times \RR^n \times G \,, \quad (x, t, g) \, \longmapsto\, \big( x, t, g_{01}(x,t) g \big) \,.
\end{equation*}
Since $U_I \subset U_1$ for each $I \in \mathcal I$, we will use the restriction of the local trivialisation \eqref{loc triv 1} on each $\pi^{-1}(U_I \times \RR^n)$ so that the transition functions $g_{IJ} : (U_I \cap U_J) \times \RR^n \to G$ and $g_{1I} : (U_1 \cap U_I) \times \RR^n \to G$ are trivial for all $I,J \in \mathcal I$ and $g_{0I} : (U_0 \cap U_I) \times \RR^n \to G$ is given by $g_{0I} = g_{01}$ for all $I \in \mathcal I$.

Consider the equations of motion \eqref{eq:B_z_punctures} derived in Theorem \ref{thm: 3d BF with type A defects 2}. Since we are now identifying $\g^\ast$ with $\g$ and the coadjoint action with the adjoint action, we can rewrite this set of equations as
\begin{equation}\label{eq:eom-B_zA_k-def-hitchin}
\partial_{t^i} B_{z_I} = \big[ \! - \widetilde{\mathcal A}^I_i, B_{z_I} \big]
\end{equation}
on the open chart $U_I \subset C$ for each $I \in \mathcal I \cup \{0\}$. We note following \cite{VW} that these clearly resemble a hierarchy of Lax equations. We can rewrite them without using coordinates by expressing them in terms of the $\g$-valued $(1,0)$-forms $B^I = B_{z_I}(z_I, \bar z_I, t) \d z_I \in \Omega^{1,0}(U_I \times \RR^n, \g)$ for $I \in \mathcal I \cup \{0\}$ as
\begin{subequations} \label{global Lax}
\begin{equation}\label{global Lax a}
\partial_{t^i} B^I = \big[ \! - \widetilde{\mathcal A}^I_{i}, B^I \big] \,.
\end{equation}
Since these equations for $I \in \mathcal I$ are compatible on overlaps $U_I \cap U_J \neq \emptyset$ and in fact $B^I = B^J$ and $\widetilde{\mathcal A}^I_{i} = \widetilde{\mathcal A}^J_{i}$ since the transition functions on these overlaps are trivial, we may consider the equations \eqref{global Lax a} simply for $I \in \{ 0, 1 \}$. In that case, $B^I \in \Omega^{1,0}(U_I \times \RR^n, \g)$ and $\widetilde{\mathcal A}^I_{i}$ is a $\g$-valued function on $U_I \times \RR^n$ for $I\in \{0, 1 \}$ which on the overlap $U_0 \cap U_1$ are related by
\begin{equation} \label{global Lax b}
B^0 = g_{01} B^1 g_{01}^{-1} \,, \qquad \widetilde{\mathcal A}^0_{i} = g_{01} \widetilde{\mathcal A}^1_{i} g_{01}^{-1} - \partial_{t^i} g_{01} g_{01}^{-1}\,.
\end{equation}
\end{subequations}
We could also further rewrite \eqref{global Lax a} more compactly as
\begin{equation}\label{eq:eom-B_zA_k-def-hitchin U1 geom}
\d_{\RR^n} B^I = \big[\! - \mathcal A^I, B^I \big]
\end{equation}
recalling that $\mathcal A^I = \mathcal A^I_{j} \d u^j = \widetilde{\mathcal A}^I_{i} \d t^i$, see Theorem \ref{thm: 3d BF with type A defects 2}, and correspondingly the second equation in \eqref{global Lax b} would read 
\begin{equation}
\label{transfo_A}
    \mathcal A^0 = g_{01} \mathcal A^1 g_{01}^{-1} - \d_{\RR^n} g_{01} g_{01}^{-1}\,.
\end{equation}

In order for \eqref{global Lax} to describe Lax equations on the Riemann surface $C$ for the hierarchy of the Hitchin system, however, we need $B^I$ and $\widetilde{\mathcal A}^I_{i}$ to be solutions of \eqref{F_B} and \eqref{A_defect_modif} (recall that the latter represent the moment map condition $\mu=0$).
In addition, to ensure that the resulting Lax matrices are meromorphic in the spectral parameter, we make a change of local trivialisation on $\mathcal P$, moving from \eqref{local triv examples} to new local trivialisations $\tilde \psi_0$ and $\tilde \psi_1$ with respect to which $A''^I = 0$ for $I \in \{ 0, 1\}$, or equivalently $A''_{\bar z_I} = 0$ in each local chart $(U_I, z_I)$ for $I \in \mathcal I \cup \{ 0 \}$, see \S\ref{sec: M and G details} for details. Let $h = (h_I) \in \check C^0(C,G)$ be the \v{C}ech $0$-cochain implementing this change of trivialisation and let
\begin{equation} \label{Lax definition}
L^I \coloneqq h_I B^I h_I^{-1} \,, \qquad - M^I_{i} \coloneqq h_I \widetilde{\mathcal A}^I_{i} h_I - \partial_{t^i} h_I h_I^{-1} 
\end{equation}
be the local expressions for $B^I$ and $\widetilde{\mathcal A}^I_{i}$ with respect to this new local trivialisation. The additional minus sign in the definition of $M^I_{i}$ is introduced so that the equations \eqref{global Lax a} take on the standard Lax form in this local trivialisation, namely
\begin{subequations} \label{global Lax def}
\begin{equation}\label{global Lax def a}
\partial_{t^i} L^I = [ M^I_{i}, L^I ]
\end{equation}
on $U_I$ for $I \in \{ 0, 1 \}$. On the overlap $(U_0 \cap U_1) \times \RR^n$ we have the same relations as in \eqref{global Lax b}, namely
\begin{equation} \label{global Lax def b}
L^0 = \gamma L^1 \gamma^{-1} \,, \qquad M^0_{i} = \gamma M^1_{i} \gamma^{-1} + \partial_{t^i} \gamma \gamma^{-1}\,,
\end{equation}
\end{subequations}
where the new transition function 
\begin{equation}
\label{def_gamma}
\gamma \coloneqq h_0 g_{01} h_1^{-1} : (U_0 \cap U_1) \times \RR^n \to G    
\end{equation} 
is now holomorphic on $U_0 \cap U_1$, see \S\ref{sec: M and G details}. We note here that the plus sign in the last term of the second equation in \eqref{global Lax def b} stems from the minus sign introduced in the second definition in \eqref{Lax definition}. In other words, $-M^0_{i}$ and $-M^1_{i}$ are related by an ordinary gauge transformation by $\gamma$. If we introduce also the $\g$-valued $1$-forms $M^I \coloneqq M^I_i \d_{\RR^n} t^i \in \Omega^1(U_I \times \RR^n, \g)$ then the Lax hierarchy \eqref{global Lax def a} can be rewritten more compactly as in \eqref{eq:eom-B_zA_k-def-hitchin U1 geom}, namely
\begin{subequations} \label{geometric Lax}
\begin{equation}\label{geometric Lax eq}
\d_{\RR^n} L^I = [ M^I, L^I ]
\end{equation}
and also the relations \eqref{global Lax def b} become
\begin{equation} \label{geometric Lax overlap}
L^0 = \gamma L^1 \gamma^{-1} \,, \qquad M^0 = \gamma M^1 \gamma^{-1} + \d_{\RR^n} \gamma \gamma^{-1}\,.
\end{equation}
\end{subequations}
In this adapted local trivialisation, 
the equations of motion \eqref{F_B} and \eqref{A_defect_modif} become 
\begin{subequations}
\begin{align}
\partial_{\bar z_I} M^I_i &=- 2 \pi i \nabla P_i\big( L_{z_I}(\mathsf q_i) \big) \delta\big( z_I - z_I(\mathsf q_i) \big)\,,\\
\label{moment map constraint L} \partial_{\bar z_I} L_{z_I} &=  2 \pi i \sum_{\substack{\alpha=1\\ \mathsf p_\alpha \in U_I}}^N \varphi_\alpha \Lambda_\alpha \varphi_\alpha^{-1} \delta\big( z_I - z_I(\mathsf p_\alpha) \big)
\end{align}
\end{subequations}
for each $I \in \mathcal I \cup \{ 0\}$, where $L^I = L_{z_I}(z_I, \bar z_I, t) \d z_I$ is the local expression of the Lax matrix $L^I$ in the coordinate chart $(U_I, z_I)$.
Equivalently, using the identity \eqref{eq:del-Dirac-rel}, we have
\begin{subequations}
\begin{align}
&\partial_{\bar z_I} \bigg( M^I_i - \frac{\nabla P_i \big( L_{z_I}(\mathsf q_i)\big)}{z_I - z_I(\mathsf q_i)} \bigg) = 0,\label{eq:a_k-gauged}\\
&\partial_{\bar z_I} \Bigg( L_{z_I} + \sum_{\substack{\alpha=1\\ \mathsf p_\alpha \in U_I}}^N \frac{\varphi_\alpha \Lambda_\alpha \varphi_\alpha^{-1}}{z_I - z_I(\mathsf p_\alpha)}\Bigg) = 0\,.\label{eq:b_z-gauged}
\end{align}
\end{subequations}
So $M^I_i$ and $L_{z_I}$ are meromorphic in the chosen local trivialisation, with a specific pole structure.

First, equation \eqref{eq:b_z-gauged} tells us that $L^1$ is a $\g$-valued meromorphic $(1,0)$-form on $U_1$ with simple poles at each $\mathsf p_\alpha$ with residue $-\varphi_\alpha \Lambda_\alpha \varphi_\alpha^{-1}$ there. Indeed, note that the residue of a meromorphic $(1,0)$-form is independent of the local coordinate used. Locally around the point $\mathsf p_\alpha$ for any $\alpha = 1, \ldots, N$, if $\mathsf p_\alpha \in U_I$ for some $I \in \mathcal I$ then in the coordinate chart $(U_I, z_I)$ we can write
\begin{subequations} \label{Lax behaviour}
\begin{equation}\label{Lax behaviour a}
L^1 = L^I = \bigg( -\frac{\varphi_\alpha \Lambda_\alpha \varphi_\alpha^{-1}}{z_I - z_I(\mathsf p_\alpha)} + J^I_\alpha \bigg) \d z_I
\end{equation}
where the first equality follows from the fact that the transition functions of the principal $G$-bundle ${\mathcal P}$ between the opens $U_1$ and $U_I \subset U_1$ were, by definition, trivial. In the last expression, $J^I_\alpha$ denotes the holomorphic part of $L_{z_I}$ in the neighbourhood of $\mathsf p_\alpha$. Also, since $\mathsf p_\alpha \not \in U_0$ by assumption, it follows that $L^0 = L_{z_0} \d z_0$ is holomorphic on $U_0$. Therefore, using the relation \eqref{global Lax def b} we have
\begin{equation} \label{Lax behaviour b}
\gamma L^1 \gamma^{-1} = L^0
\end{equation}
\end{subequations}
on $U_0 \cap U_1$, which expresses the fact that the $\g$-valued $(1,0)$-form $\gamma L^1 \gamma^{-1} \in \Omega^{1,0} \big( (U_1 \cap U_0) \times \RR^n, \g \big)$ extends holomorphically to $U_0 \times \RR^n$. We will refer to $L$ as the \emph{Lax matrix} of the Hitchin system.

Next, equation \eqref{eq:a_k-gauged} tells us that $M^1_i$ is a $\g$-valued meromorphic function on $U_1$ with a simple pole at the point $\mathsf q_i$, with the local expression
\begin{equation} \label{MIi explicit}
M^1_i = M^I_i =  \frac{\nabla P_i \big( L_{z_I}(\mathsf q_i)\big)}{z_I - z_I(\mathsf q_i)} + K^I_i
\end{equation}
in the chart $(U_I, z_I)$ which is such that $\mathsf q_i \in U_I$. Here $K^I_i$ denotes the holomorphic part of $M^I_i$ in the neighbourhood of $\mathsf q_i$. Moreover, since we have $\mathsf q_i \not\in U_0$ by assumption, $M^0_i$ is a holomorphic $\g$-valued function on $U_0$ which is related to $M^1_i$ on the overlap $(U_0 \cap U_1) \times \RR^n$ by the second relation in \eqref{global Lax def b}. That is, $\gamma M^1_i \gamma^{-1} + \partial_{t^i} \gamma \gamma^{-1} = M^0_i$ extends holomorphically from $U_0 \cap U_1$ to $U_0$.

To summarise the above discussion, the Hitchin system can indeed be presented as the hierarchy of Lax equations \eqref{global Lax def a}, as a direct consequence expressing our variational equations \eqref{eqs_gauged_action} in a local trivialisation with $A''^I = 0$, in terms of:
\begin{enumerate}
    \item the Lax matrix $L$ which is a meromorphic section of $\pi_C^\ast \bigwedge^{1,0} C \otimes \g_{\pi_C^\ast {\mathcal P}}$ fixed in terms of the following degrees of freedom:
\begin{itemize}
  \item[$(a)$] Maps $\varphi_\alpha : \RR^n \to G / G_{\Lambda_\alpha}$ into the coadjoint orbits $\mathcal O_\alpha \cong G / G_{\Lambda_\alpha}$ for each $\alpha =1,\ldots, N$,
  \item[$(b)$] The transition function $\gamma : (U_0 \cap U_1) \times \RR^n \to G$ holomorphic in $U_0 \cap U_1$ and encoding the holomorphic structure on the principal $G$-bundle ${\mathcal P}$. In particular $L$ satisfies \eqref{Lax behaviour b}.
\end{itemize}

\item The $\g$-valued meromorphic functions $M^0_i$ and $M^1_i$ for each $i =1, \ldots, n$ with pole structure dictated by \eqref{MIi explicit} and satisfying the second equation in \eqref{geometric Lax overlap}.
\end{enumerate}

\subsection{Unifying Lagrangian 1-form} \label{sec:unifyingmultiform}

The description of the Hitchin system as a hierarchy of Lax equations \eqref{global Lax def a} was obtained in \S\ref{sec: Lax matrices} from the set of equations of Theorem \ref{thm: 3d BF with type A defects 2} by moving to a local trivialisation in which $A''^I = 0$ and solving the moment map condition \eqref{A_defect_modif}. The main point is that these equations are \textit{variational}, deriving from the  multiform 3d BF action \eqref{S Gamma 2}, so it is natural to ask whether we can also derive \eqref{global Lax def a} directly as Euler-Lagrange equations of an appropriate action.

In Theorem \ref{thm: Lag 1-form for Hitchin} below we construct the \textit{unifying action} and associated Lagrangian $1$-form by writing the $3$d mixed BF action \eqref{S Gamma 2} in a local trivialisation with $A''^I = 0$ and explicitly solving the moment map condition \eqref{A_defect_modif}. In Theorem \ref{thm: eom for 1d action} we then prove that this unifying action does indeed reproduce the Lax equations \eqref{global Lax def a} variationally.

\begin{theorem} \label{thm: Lag 1-form for Hitchin}
The $3$d mixed BF action \eqref{S Gamma 2} written in a local trivialisation where $A''^0 = A''^1 = 0$ and with the moment map condition \eqref{A_defect_modif} explicitly solved, leads to the unifying action
\begin{subequations} \label{eq:unifyingaction}
\begin{equation} \label{eq:unifyingaction a}
S_{{\rm H}, \Gamma}\big[ L, \gamma, (\varphi_\alpha), t \big] = \int_\Gamma \Lag_{\rm H} \,,
\end{equation}
for any parametrised curve $\Gamma : (0,1) \to \RR^n$, where $\Lag_{\rm H} \in \Omega^1(\RR^n)$ is the Hitchin Lagrangian $1$-form defined using a small counter-clockwise oriented loop $c_{\mathsf p}$ in $U_0 \cap U_1$ around the point $\mathsf p \in C$ by
\begin{equation} \label{eq:unifyingaction b}
\Lag_{\rm H} \coloneqq  \frac{1}{2 \pi i} \int_{c_{\mathsf p}} \big\langle L^0, \d_{\RR^n} \gamma \gamma^{-1} \big\rangle + \sum_{\alpha=1}^N \big\langle \Lambda_\alpha, \varphi_\alpha^{-1} \d_{\RR^n} \varphi_\alpha \big\rangle - H_i \d_{\RR^n} t^i,\qquad H_i=P_i\big( L^1_{z_I}(\mathsf{q}_i) \big) \,.
\end{equation}
\end{subequations}

In particular, under a change of `residual' local trivialisation $h = (h_0, h_1) \in \check C^0(C, G)$ with $h_I$ holomorphic on $U_I$ for $I \in \{0,1\}$ so that the condition $A''^0 = A''^1 = 0$ is preserved, the action \eqref{eq:unifyingaction a} is invariant in the sense that
\begin{equation} \label{1d action invariance}
S_{{\rm H}, \Gamma}\big[ hLh^{-1}, h_0 \gamma h_1^{-1}, (h_1 \varphi_\alpha) \big] = S_{{\rm H}, \Gamma}\big[ L^0, \gamma, (\varphi_\alpha) \big] \,,
\end{equation}
where $hLh^{-1}$ stands for the Lax matrix given by $h_0 L^0 h_0^{-1}$ on $U_0$ and $h_1 L^1 h_1^{-1}$ on $U_1$.
\end{theorem}
\begin{proof}
We choose local trivialisations over $U_0\times \RR^n$ and $U_1\times \RR^n$ in which the components $A''_{\bar z_0}$, $A''_{\bar z_1}$ of the partial connection $A$ vanish.  Then the local expression for the curvature \eqref{FA def} becomes
\begin{equation}
F^I_A = \partial_{\bar{z}_I}\A^I_j\d\bar{z}_I\wedge \d u^j+\frac12 (\partial_{u^i}\A^I_j-\partial_{u^j}\A^I_i+[\A^I_i,\A^I_j])\d u^i\wedge \d u^j.
\end{equation}
Unlike earlier trivialisations used in \S\ref{sec: Lag for Hitchin mod G} for which we had \eqref{dgdu}, these trivialisations will have a transition function $\gamma = g_{01} : U_0\cap U_1\times\mathbb{R}^n\to G$ that depends on $u\in\mathbb{R}^n$ (because our choice of trivialisation depends on $A''$, and $A''$ depends on $u$).  On the overlaps of coordinate charts we have
\begin{equation}\label{eq:4.1 A transition}
    \A_j^0=\gamma \A_j^1\gamma^{-1}-\partial_{u^j}\gamma\gamma^{-1},\qquad \bar{\partial}\gamma=0.
\end{equation}

Recall from \S\ref{sec: Lax matrices} that to emphasise that $B$ solves the moment map condition \eqref{A_defect_modif} we denote it by $L$ and refer to it as a Lax matrix. Now the moment map condition \eqref{eq:b_z-gauged} can be written concisely as
\begin{equation} \label{moment map condition L}
\bar\partial L^0 =  0 \,, \qquad \bar\partial L^1 = -2\pi i \sum_{\alpha=1}^N \varphi_\alpha \Lambda_\alpha \varphi_\alpha^{-1} \delta_{\mathsf p_\alpha} \,.
\end{equation}

In order to evaluate the action \eqref{S Gamma 2} in these trivialisations, we choose a circle $c_{\mathsf p}$ as described in the statement of the theorem, and let $R_0\subset U_0$, $R_1\subset U_1$ be the closure of the interior and exterior of this circle.  Then $R_0\cap R_1=c_{\mathsf p}$ and $R_0\cup R_1=C$.  The first term in the action \eqref{S Gamma 2} is
\begin{equation}\label{eq:4.1 integral}
\frac{1}{2\pi i}\int_{C\times \Gamma}\langle B,F_A\rangle =\sum_{I \in \{0,1\}} \frac{1}{2\pi i}\int_{R_I\times \Gamma}\langle L^I,\bar{\partial}\A^I_j\rangle\wedge \d u^j.
\end{equation}
We evaluate these terms separately.  For $I=1$ we find
\begin{align*}
\frac{1}{2 \pi i} \int_{R_1} \langle L^1, \bar{\partial}\A_j^1\rangle  &=  \frac{1}{2 \pi i}\int_{R_1}\left(-\bar{\partial}\langle L^1,\A_j^1\rangle + \langle\bar{\partial}L^1,\A_j^1\rangle\right)\\
&=- \frac{1}{2 \pi i} \int_{\partial R_1} \langle L^1, \A_j^1 \rangle - \sum_{\alpha=1}^N \big\langle \varphi_\alpha \Lambda_\alpha \varphi_\alpha^{-1}, \A_j^1(\mathsf p_\alpha) \big\rangle\\
&= \frac{1}{2 \pi i} \int_{c_\mathsf{p}} \big\langle L^0, \gamma \A_j^1 \gamma^{-1} \big\rangle - \sum_{\alpha=1}^N \big\langle \varphi_\alpha \Lambda_\alpha \varphi_\alpha^{-1}, \A_j^1(\mathsf p_\alpha) \big\rangle\,,
\end{align*}
where in the second equality we used the second relation in \eqref{moment map condition L} and in the last step we used the first relation in \eqref{geometric Lax overlap} and the fact that $\partial R_1$ is equal to $c_{\mathsf{p}}$ with reversed orientation. For $I=0$,
\begin{align*}
\frac{1}{2 \pi i} \int_{R_0} \langle L^0, \bar\partial \A_j^0\rangle &= - \frac{1}{2 \pi i} \int_{R_0}\bar{\partial} \langle L^0, \A_j^0 \rangle = - \frac{1}{2 \pi i} \int_{\partial R_0} \langle L^0, \A_j^0 \rangle\\
&= -\frac{1}{2 \pi i} \int_{c_{\mathsf{p}}} \big\langle L^0, \gamma \A_j^1 \gamma^{-1} \big\rangle + \frac{1}{2 \pi i} \int_{c_{\mathsf{p}}} \big\langle L^0, \partial_{u^j} \gamma \gamma^{-1} \big\rangle \,,
\end{align*}
where the first equality uses the first relation in \eqref{moment map condition L} and in the third step we used \eqref{eq:4.1 A transition} and the fact that $\partial R_0$ coincides with $c_{\mathsf{p}}$ and has the same orientation. It now follows from the above computation of both integrals on the right-hand side of \eqref{eq:4.1 integral} that
\begin{equation}
\frac{1}{2 \pi i}\int_{C\times \Gamma}\langle B,F_A\rangle =\frac{1}{2\pi i}\int_{c_{\mathsf p}\times\Gamma}\langle L^0,\d_{\RR^n}\gamma \gamma^{-1}\rangle- \sum_{\alpha=1}^N \int_\Gamma\big\langle \varphi_\alpha \Lambda_\alpha \varphi_\alpha^{-1}, \A_j^1(\mathsf p_\alpha) \big\rangle \d u^j\,.
\end{equation}
Substituting this into \eqref{S Gamma 2} and setting $B=L$ in the remaining terms gives the desired result.

To prove the last `in particular' statement, consider a change of `residual' local trivialisation $h = (h_0, h_1) \in \check C^0(C, G)$, with $h_I$ holomorphic on $U_I$ so that the condition $A''_{\bar z_I} = 0$ is preserved. We have $\gamma \mapsto \tilde\gamma = h_0 \gamma h_1^{-1}$, $\tilde L^0 = h_0 L^0 h_0^{-1}$ and $\tilde L^1 = h_1 L^1 h_1^{-1}$ so that $\tilde L^0 = \tilde \gamma \tilde L^1 \tilde \gamma^{-1}$. Then the integral in the first term on the right-hand side of \eqref{eq:unifyingaction b} transforms to
\begin{equation} \label{pre Lag 1-form change loc triv 1}
\int_{c_{\mathsf p}} \big\langle \tilde L^0, \d_{\RR^n} \tilde \gamma \tilde \gamma^{-1} \big\rangle = \int_{c_{\mathsf p}} \big\langle L^0, \d_{\RR^n} \gamma \gamma^{-1} \big\rangle + \int_{c_{\mathsf p}} \big\langle L^0, h_0^{-1} \d_{\RR^n} h_0 \big\rangle - \int_{c_{\mathsf p}} \big\langle L^1, h_1^{-1} \d_{\RR^n} h_1 \big\rangle \,.
\end{equation}
Note that the second term on the right-hand side vanishes by Cauchy's theorem since $h_0$ and $L_0$ are both holomorphic on $U_0$ and $c_{\mathsf p}$ is a small contour around $\mathsf p \in U_0$. To evaluate the last term on the right-hand side of \eqref{pre Lag 1-form change loc triv 1}, recall that the section $L^1$ has poles at the marked points $\mathsf p_\alpha$ given by \eqref{Lax behaviour a}. Since $h_1 : U_1 \to G$ is holomorphic on $U_1 = C \setminus \{ \mathsf p\}$ we then deduce using the residue theorem that
\begin{equation*}
\frac{1}{2 \pi i} \int_{c_{\mathsf p}} \big\langle L^1, h_1^{-1} \d_{\RR^n} h_1 \big\rangle = - \sum_{\alpha = 1}^N \big\langle \! -\varphi_\alpha \Lambda_\alpha \varphi_\alpha^{-1}, h_1(\mathsf p_\alpha)^{-1} \d_{\RR^n} h_1(\mathsf p_\alpha) \big\rangle
\end{equation*}
where the first sign comes from noting that $c_{\mathsf p}$ can be contracted down to a sum of small \emph{clockwise} circles around each point $\mathsf p_\alpha$. Substituting this into the right-hand side of \eqref{pre Lag 1-form change loc triv 1} we deduce the transformation property
\begin{equation} \label{Lag 1-form change loc triv 1}
 \frac{1}{2 \pi i} \int_{c_{\mathsf p}} \big\langle \tilde L^0, \d_{\RR^n} \tilde \gamma \tilde \gamma^{-1} \big\rangle = \frac{1}{2 \pi i} \int_{c_{\mathsf p}} \big\langle L^0, \d_{\RR^n} \gamma \gamma^{-1} \big\rangle - \sum_{\alpha = 1}^N \Big\langle \Lambda_\alpha, \varphi_\alpha^{-1} \big( h_1(\mathsf p_\alpha)^{-1} \d_{\RR^n} h_1(\mathsf p_\alpha) \big) \varphi_\alpha \Big\rangle \,.
\end{equation}
Consider now the second term in the Lagrangian $1$-form \eqref{eq:unifyingaction b}. Since $\varphi_\alpha$ transforms under the change of local trivialisation as $\varphi_\alpha \mapsto \tilde \varphi_\alpha = h_1(\mathsf p_\alpha) \varphi_\alpha$, see the discussion at the start of \S\ref{sec: adding punctures}, this term in the Lagrangian transforms as
\begin{equation} \label{Lag 1-form change loc triv 2}
\sum_{\alpha=1}^N \big\langle \Lambda_\alpha, \tilde \varphi_\alpha^{-1} \d_{\RR^n} \tilde \varphi_\alpha \big\rangle
= \sum_{\alpha=1}^N \big\langle \Lambda_\alpha, \varphi_\alpha^{-1} \d_{\RR^n} \varphi_\alpha \big\rangle + \sum_{\alpha = 1}^N \Big\langle \Lambda_\alpha, \varphi_\alpha^{-1} \big( h_1(\mathsf p_\alpha)^{-1} \d_{\RR^n} h_1(\mathsf p_\alpha) \big) \varphi_\alpha \Big\rangle \,.
\end{equation}
We see now from \eqref{Lag 1-form change loc triv 1} and \eqref{Lag 1-form change loc triv 2} that the Lagrangian $1$-form \eqref{eq:unifyingaction b} is invariant under residual changes of local trivialisations $h = (h_0, h_1) \in \check C^0(C, G)$ with $h_I$ holomorphic on $U_I$.
\end{proof}

\begin{remark} \label{rem: integration by parts}
The term $\int_{c_{\mathsf p}} \big\langle L^0, \d_{\RR^n} \gamma \gamma^{-1} \big\rangle$ in the kinetic part of the unifying Lagrangian \eqref{eq:unifyingaction b} arises from a subtle mechanism. 
In the proof of Theorem \ref{thm: Lag 1-form for Hitchin}, we work on the pullback bundle $\pi_C^\ast \mathcal P$ but in a local trivialisation in which $A''^I = 0$ for each $I \in \mathcal I$. This implies, in particular, that the transition function $g_{IJ} : (U_I \cap U_J) \times \RR^n \to G$ is holomorphic in $U_I \cap U_J$ and dependent on the coordinates $u^j$ of $\RR^n$. Thus, specialising the $\mathcal A$ component of the identity \eqref{partial connection compat} to this setting, we find that
\begin{equation*}
\bar\partial \mathcal A^I = g_{IJ} \, \bar\partial \mathcal A^J \, g_{IJ}^{-1} \,.
\end{equation*}
In particular, it follows that $\langle B^I, \bar\partial \mathcal A^I\rangle = \langle B^J, \bar\partial \mathcal A^J\rangle$ on overlaps $U_I \cap U_J \neq \emptyset$ so that we obtain a $\g$-valued $3$-form $\langle B, \bar\partial \mathcal A\rangle \in  \Omega^3(C \times \RR^n)$ which, in particular, is a global $(1,1)$-form along $C$ so that it can be integrated over $C$ (strictly speaking this is a fibre integration along the fibres of the projection $\pi_{\RR^n} : C \times \RR^n \to \RR^n$). Now, by contrast, in this local trivialisation where the transition functions of $\pi_C^\ast {\mathcal P}$ explicitly depend on $\RR^n$, we find using the transformation property \eqref{partial connection compat} that
\begin{equation} \label{del B A bad transformation}
\langle \bar\partial B^I, \mathcal A^I\rangle = \langle \bar\partial B^J, \mathcal A^J\rangle - \big\langle \bar\partial B^J, g_{IJ}^{-1} \d_{\RR^n} g_{IJ} \big\rangle \,.
\end{equation}
In other words, the local expressions $\langle \bar\partial B^I, \mathcal A^I\rangle$ do not define a global $\g$-valued $3$-form on $C \times \RR^n$. Likewise, the local expressions $\langle B^I, \mathcal A^I\rangle$ transform in a similar way to \eqref{del B A bad transformation} and therefore do not define a global $\g$-valued $2$-form on $C \times \RR^n$ either. This means that in the present context, ``naive'' integration by parts in an expression like $\int_C \langle B, \bar\partial \mathcal A\rangle$ is not possible, cf. Remark \ref{rem: mu well defined}. This is the source of the term $\int_{c_{\mathsf p}} \big\langle L^0, \d_{\RR^n} \gamma \gamma^{-1} \big\rangle$.
\end{remark}

We now derive the univariational equations for the unifying $1$d action in \eqref{eq:unifyingaction}.  To do so, it is helpful to recall the constraints on the various degrees of freedom.  We recall that $L^0$ and $L^1$ are holomorphic $\mathfrak{g}$-valued (1,0)-forms on $U_0$ and $U_1\setminus\{\mathsf{p}_\alpha\}$, $\gamma$ is a holomorphic $G$-valued function on $U_0\cap U_1$, and $\varphi_\alpha$ are elements of $G$.  All of these depend on $t\in\RR^n$, and are constrained by
\begin{equation}\label{constraint}
L^1=\gamma^{-1}L^0\gamma,\qquad \Res_{\mathsf p_\alpha}L^1_{z_I} = -\varphi_\alpha\Lambda_\alpha\varphi_\alpha^{-1}.
\end{equation}
We continue to assume that the Higgs bundle determined by $\gamma$ and $L$ satisfies the stability condition \eqref{inf_freeness}.  This means that there are no non-zero holomorphic functions $X^I:U_I\to\g$ for $I \in \{0,1\}$ with the property that $[L^I,X^I]=0$ on $U_I$ and $X^1=\gamma^{-1}X^0\gamma$ on $U_0\cap U_1$.

\begin{theorem}\label{thm: eom for 1d action}
The equations of motion of the unifying $1$d action $S_{{\rm H}, \Gamma}\big[ L, \gamma, (\varphi_\alpha), t \big]$ in \eqref{eq:unifyingaction} with respect to the variables $L$, $\gamma$ and $\varphi_\alpha$ take the form
\begin{equation}\label{Lax equations}
\partial_{t^i} L^0=[M_i^0,L^0],\qquad \partial_{t^i} L^1=[M_i^1,L^1]\,.
\end{equation}
Here $M_i^0,M_i^1$ are holomorphic $\mathfrak{g}$-valued functions on $U_0$ and $U_1\setminus\{\mathsf{q}_i\}$ that are uniquely determined by the constraints
\begin{equation}\label{M constraints}
\gamma^{-1}\partial_{t^i} \gamma = \gamma^{-1}M_i^0\gamma-M_i^1,\qquad \Res_{\mathsf{q}_i} M_i^1 = \nabla P_i\big( L_{z_I}^1(\mathsf{q}_i) \big)\,.
\end{equation}
Moreover, they satisfy the zero-curvature equations
\begin{equation}\label{eq:zce Mi}
\partial_{t^i} M^0_j - \partial_{t^j} M^0_i - [ M^0_i, M^0_j ] =0, \qquad
\partial_{t^i} M^1_j - \partial_{t^j} M^1_i - [ M^1_i, M^1_j ] =0\,.
\end{equation}
The equation of motion with respect to the set of times $t^i$ for $i=1,\ldots, n$ describe the conservation equations $\partial_{t^i} H_j = 0$ for all $i,j =1,\ldots, n$.
\end{theorem}
\begin{proof} The last statement for the variation with respect to $t^i$ for $i=1,\ldots, n$ is straightforward, so we focus on deriving the equations of motions associated to variations of $L$, $\gamma$ and $\varphi_\alpha$.

It follows from the condition \eqref{constraint} that a variation of $L,\gamma$ is described by a meromorphic $\g$-valued $(1,0)$-form $\delta L^1$ on $U_1$, a holomorphic $\g$-valued $(1,0)$-form $\delta L^0$ on $U_0$, and a holomorphic $\g$-valued function $\gamma^{-1}\delta\gamma$ satisfying
\begin{equation}
\delta L^1 = \gamma^{-1} \delta L^0    \gamma+[L^1,\gamma^{-1}\delta \gamma] 
\end{equation}
on the intersection $U_0 \cap U_1$.
The action should be stationary with respect to all such variations of $L$ and to arbitrary variations of $\varphi_\alpha$. In particular, it should be stationary when $\gamma^{-1}\delta \gamma=0$, so that
\begin{equation}\label{delta L constraint}
\delta L^1 = \gamma^{-1}\delta L^0\gamma\,,
\end{equation}
and when $\delta \varphi_\alpha=0$, in which case the second constraint in \eqref{constraint} gives $\Res_{\mathsf p_\alpha} \delta L^1_{z_I} = 0$. Thus, in that case $\delta L^1$ is holomorphic on $U_1$.
The corresponding variation of $S_{{\rm H}, \Gamma}$ is
\begin{equation}
\delta S_{{\rm H}, \Gamma}= \int_0^1\left(\frac{1}{2\pi i}\oint_{c_{\mathsf p}} \langle \delta L^1,\gamma^{-1}\partial_{u^j}\gamma\rangle - \big\langle \delta L^1_{z_I}(\mathsf{q}_i),\nabla P_i(L^1_{z_I}(\mathsf{q}_i))\big\rangle \frac{\partial t^i}{\partial u^j} \right)\frac{\d u^j}{\d s}\d s \,.
\end{equation}
This must vanish for all curves $\Gamma : (0,1) \to \RR^n$ so that for $i =1,\ldots, n$ we have
\begin{equation}\label{L variation}
0=\frac{1}{2\pi i}\oint_{c_{\mathsf p}} \langle \delta L^1,\gamma^{-1}\partial_{t^i}\gamma\rangle - \big\langle \delta L^1_{z_I}(\mathsf{q}_i),\nabla P_i(L^1_{z_I}(\mathsf{q}_i))\big\rangle \,.
\end{equation}

Consider, to begin with, a variation such that $\delta L^1_{z_I}(\mathsf{q}_i)=0$ for some fixed $i$.  We find that
\begin{equation}\label{Serre}
0 = \frac{1}{2\pi i}\oint_{c_{\mathsf p}} \langle \delta L^1,\gamma^{-1}\partial_{t^i}\gamma\rangle = \frac{1}{2\pi i}\oint_{c_{\mathsf p}} \langle \delta L^0, \partial_{t^i}\gamma \gamma^{-1} \rangle \,,
\end{equation}
where the second equality just follows from the invariance of the bilinear form $\langle \cdot, \cdot \rangle : \g \otimes \g \to \CC$ and the constraint \eqref{delta L constraint}. Let us interpret the equation \eqref{Serre} using sheaf cohomology. For each $t \in \RR^n$ we denote by $E_t \to C$ the holomorphic vector bundle with transition function $\Ad_{\gamma(t)}$ between the two local trivialisations over $U_1$ and $U_0$.
The pair $\delta L = (\delta L^0, \delta L^1)$ of $\g$-valud $(1,0)$-forms on $U_0$ and $U_1$ related by the constraint \eqref{delta L constraint} describes a holomorphic section of $\bigwedge^{1,0} C \otimes E_t$ that vanishes at $\mathsf{q}_i$ and hence it determines an element of the \v{C}ech cohomology group $H^0\big( C,\Lambda^{1,0}C\otimes E_t(-\mathsf{q}_i) \big)$, in which $(-\mathsf{q}_i)$ indicates that the section vanishes at $\mathsf{q}_i$. On the other hand, the pair $(\partial_{t^i}\gamma \gamma^{-1}, \gamma^{-1}\partial_{t^i}\gamma)$ of $\g$-valued functions on $U_0 \cap U_1$ describe a holomorphic section of $E_t$ over $U_0\cap U_1$ defined relative to the local trivialisations over $U_0$ and $U_1$. It thus determines an element of the \v{C}ech cohomology group $H^1\big( C,E_t(\mathsf{q}_i) \big)$. The two integral expressions in \eqref{Serre} then correspond to the Serre duality pairing between $H^0\big( C,\Lambda^{1,0}C\otimes E_t(-\mathsf{q}_i) \big)$ and $H^1\big( C,E_t(\mathsf{q}_i) \big)$ described using the local trivialisation of $E_t$ over $U_1$ and $U_0$, respectively. This pairing is nondegenerate, so if \eqref{Serre} holds for all variation $\delta L$ then $(\partial_{t^i}\gamma \gamma^{-1}, \gamma^{-1}\partial_{t^i}\gamma)$ is zero in cohomology.  This means that
\begin{equation}\label{M}
\gamma^{-1}\partial_{t^i}\gamma = \gamma^{-1}M_i^0\gamma-M_i^1 \,,
\end{equation}
where $M_i^0$ and $M_i^1$ are sections of $E_t(\mathsf{q}_i)$ over $U_0$ and $U_1$, respectively, and the right-hand side of \eqref{M} is the coboundary map in sheaf cohomology, expressed in the local trivialisation over $U_1$. More precisely, $M_i^0$ is a $\g$-valued holomorphic function on $U_0$ and $M_i^1$ is a $\g$-valued holomorphic on $U_1\setminus\{\mathsf{q}_i\}$ with a simple pole at $\mathsf{q}_i$. There is some freedom in the choice of $M_i^I$ solving \eqref{M}: if $N_i^0$ and $N_i^1$ are sections of $E_t$ over $U_0$ and $U_1\setminus\{\mathsf{q}_i\}$, respectively, such that $N_i^1$ has a simple pole at $\mathsf{q}_i$ and $\gamma^{-1}N_i^0\gamma=N_i^1$, then adding $N_i^I$ to $M_i^I$ produces a new solution of \eqref{M}.  Thus $M_i^0,M_i^1$ are unique up to the addition of elements of $H^0\big( C,E_t(\mathsf{q}_i) \big)$.

Now we insert \eqref{M} into \eqref{L variation}, in which the variation $\delta L^1$ is no longer constrained to vanish at $\mathsf{q}_i$.  We find that
\begin{align}
0 &= \frac{1}{2\pi i}\oint_{c_{\mathsf p}} \langle \delta L^1,\gamma^{-1}M_i^0\gamma-M_i^1\rangle - \big\langle \delta L^1_{z_I}(\mathsf{q}_i),\nabla P_i\big( L^1_{z_I}(\mathsf{q}_i) \big) \big\rangle \notag\\
&= \frac{1}{2\pi i}\oint_{c_{\mathsf p}} \langle \delta L^0,M_i^0\rangle -  \frac{1}{2\pi i}\oint_{c_{\mathsf p}} \langle \delta L^1,M_i^1\rangle- \big\langle \delta L^1_{z_I}(\mathsf{q}_i),\nabla P_i\big( L^1_{z_I}(\mathsf{q}_i) \big) \big\rangle \,.
\label{L variation 2}
\end{align}
The contour integral of $\langle \delta L^0,M_i^0\rangle$ on the right-hand side vanishes because its integrand is holomorphic on $U_0$. Deforming the contour integral of $\langle\delta L^1,M_i^1\rangle$ from the anticlockwise contour $c_{\mathsf p}$ to a clockwise small circle around $\mathsf{q}_i$ yields 
\begin{equation*}
\frac{1}{2\pi i}\oint_{c_{\mathsf p}} \langle \delta L^1,M_i^1\rangle=-\Big\langle \delta L^1_{z_I}(\mathsf{q}_i),\Res_{\mathsf{q}_i} M_i^1\Big\rangle \,.
\end{equation*}
We may then rewrite \eqref{L variation 2} as
\begin{equation}\label{M residue constraint}
0=\Big\langle \delta L^1_{z_I}(\mathsf{q}_i),\Res_{\mathsf{q}_i} M_i^1 - \nabla P_i\big( L_{z_I}^1(\mathsf{q}_i) \big)\Big\rangle \,.
\end{equation}
We would like to conclude from this that
\begin{equation}\label{M residue}
\Res_{\mathsf{q}_i} M_i^1 = \nabla P_i\big( L_{z_I}^1(\mathsf{q}_i) \big) \,.
\end{equation}
However, this does not follow immediately.  The variation $\delta L^1$ is constrained by \eqref{delta L constraint}, and this constraint may mean that $\delta L^1_{z_I}(\mathsf{q}_i)$ takes values in a proper subspace of the fibre $E_t|_{\mathsf{q}_i}$ of $E_t$ at $\mathsf{q}_i$. If so, \eqref{M residue constraint} does not constrain all of the components of $\Res_{\mathsf{q}_i}M_i^1$. So to find a solution $M_i^0,M_i^1$ of \eqref{M residue constraint} we use more sophisticated methods.

Consider the exact sequence of sheaves on $C$ given by
\begin{equation} \label{short exact}
0 \longrightarrow E_t \longrightarrow E_t(\mathsf{q}_i) \longrightarrow E_t|_{\mathsf{q}_i} \longrightarrow 0 \,.
\end{equation}
Here, by abuse of notation, $E_t$ denotes the sheaf of holomorphic sections of the vector bundle $E_t$, and $E_t(\mathsf{q}_i)$ denotes the sheaf of sections that are holomorphic on $C\setminus\{\mathsf{q}_i\}$ with a simple pole at $\mathsf{q}_i$. Moreover, $E_t|_{\mathsf{q}_i}$ is the fibre of $E_t$ at $\mathsf{q}_i$, which we regard as a skyscraper sheaf on $C$ supported at the point $\mathsf{q}_i$.  This exact sequence of sheaves induces a long exact sequence in \v{C}ech cohomology:
\begin{equation}\label{exact sequence}
\begin{tikzcd}
0 \arrow[r] & H^0(C,E_t) \arrow[r, "\iota"] & H^0(C,E_t(\mathsf{q}_i)) \arrow[r, "\Res_{\mathsf{q}_i}"] & E_t|_{\mathsf{q}_i} \arrow[r, "h"] & H^1(C,E_t) \arrow[r] & \ldots \,.
\end{tikzcd}
\end{equation}
The map $\iota$ is the inclusion of the space of holomorphic sections in the space of sections with a simple pole. The map $\Res_{\mathsf q_i}$ is given by evaluating the residue of a section at the point $\mathsf{q}_i$. The so-called connecting homomorphism $h$ acts as follows: given $X_i \in E_t|_{\mathsf{q}_i}$, we use the exactness of \eqref{short exact} at $E_t|_{\mathsf q_i}$ to choose a holomorphic section $h(X_i)$ of $E_t$ over $U_1 \setminus \{\mathsf q_i\}$ with a simple pole at $\mathsf{q}_i$ with residue $X_i$ there. The restriction of this holomorphic section to $U_0\cap U_1$ determines a cohomology class in $H^1(C,E_t)$ which is independent of the choice of holomorphic section.

With this notation and letting $X_i \coloneqq \Res_{\mathsf{q}_i}M_i^1-\nabla P_i\big( L^1_{z_I}(\mathsf{q}_i) \big)$, we can rewrite \eqref{M residue constraint} as follows:
\begin{equation} \label{Serre 2}
0 = \big\langle \delta L^1_{z_I}(\mathsf{q}_i),X_i \big\rangle = -\frac{1}{2\pi i}\oint_{c_{\mathsf p}} \big\langle \delta L^1,h(X_i) \big\rangle = -\frac{1}{2\pi i}\oint_{c_{\mathsf p}} \big\langle \delta L^0, \gamma h(X_i) \gamma^{-1} \big\rangle \,,
\end{equation}
where the last step is as in \eqref{Serre}.
These two integral expressions describe the Serre duality pairing between $\delta L\in H^0(C,\Lambda^{1,0}\otimes E^\ast_t)$ and $\big( \gamma h(X_i) \gamma^{-1}, h(X_i) \big) \in H^1(C,E_t)$ in the local trivialisation of $E_t$ over $U_1$ and $U_0$, respectively.
Since \eqref{Serre 2} holds for all $\delta L\in H^0(C,\Lambda^{1,0}\otimes E_t)$, and since the pairing is nondegenerate, we conclude that $h(X_i)=0\in H^1(C,E_t)$.  Since the sequence \eqref{exact sequence} is exact, there exists an $N_i \in H^0\big( C,E_t(\mathsf{q}_i) \big)$ such that $X_i = \Res_{\mathsf q_i} N_i$.  We then have that $\tilde{M}_i \coloneqq M_i - N_i$ solves both \eqref{M} and \eqref{M residue}, as required.

There is still some freedom in the choice of $M_i$ solving \eqref{M}, \eqref{M residue}. The solution of \eqref{M} is unique up to the addition of $\g$-valued functions $N_i^0,N_i^1$ on $U_0,U_1$ satisfying $N_i^1=\gamma^{-1}N_i^0\gamma$.  Since the pole of $M_i^1$ is fixed by \eqref{M residue}, these functions $N_i^I$ are holomorphic, so they determine a holomorphic section of $E_t$.  Therefore $M_i$ is unique up to addition of elements of $H^0(C,E_t)$.

The analysis so far was based on variations such that $\delta\gamma=0$ and $\delta\varphi_\alpha=0$, which produced necessary conditions for the stationarity of the action. We now consider general variations of $L$, $\varphi_\alpha$ and $\gamma$.  We assume that $M_i$ has been chosen solving the necessary conditions \eqref{M} and \eqref{M residue}. From \eqref{constraint}, $\delta L^1$ is now meromorphic and the variations satisfy
\begin{align}
\label{variation constraint}
[L^0,\delta \gamma\gamma^{-1}]&=\gamma \delta L^1\gamma^{-1}-\delta L^0 \,,\\
\label{delta L residue}
\Res_{\mathsf{p}_\alpha}\delta L^1_{z_I}&=-[\delta\varphi_\alpha\varphi_\alpha^{-1},\varphi_{\alpha}\Lambda_\alpha\varphi_\alpha^{-1}] \,.
\end{align}
The variation of $S_{{\rm H}, \Gamma}$ is
\begin{multline}
\delta S_{{\rm H}, \Gamma} = \int_0^1\bigg(
\frac{1}{2\pi i}\oint_{c_{\mathsf p}} \langle \delta L^1,\gamma^{-1}\partial_{u^j}\gamma\rangle
+\frac{1}{2\pi i}\oint_{c_{\mathsf p}} \big\langle L^1,\gamma^{-1}\partial_{u^j}(\delta \gamma\gamma^{-1})\gamma \big\rangle\\
+ \sum_{\alpha=1}^N\big\langle\Lambda_\alpha,\varphi_\alpha^{-1}\partial_{u^j}(\delta\varphi_\alpha\varphi_\alpha^{-1})\varphi_\alpha\big\rangle
- \big\langle\delta L^1_{z_I}(\mathsf{q}_i),\nabla P_i\big( L^1_{z_I}(\mathsf{q}_i) \big) \big\rangle \frac{\partial t^i}{\partial u^j}\bigg)\frac{\d u^j}{\d s}\d s \,.
\end{multline}
We will rewrite this by substituting \eqref{M} in the first term and by using the following contour integral, which is evaluated using \eqref{M residue} and \eqref{delta L residue}:
\begin{equation}
\frac{1}{2\pi i}\oint_{c_{\mathsf p}} \langle\delta L^1,M_i^1\rangle=-\big\langle \delta L^1_{z_I}(\mathsf{q}_i),\nabla P_i\big( L^1_{z_I}(\mathsf{q}_i) \big)\big\rangle+\sum_{\alpha=1}^N \big\langle [\delta\varphi_\alpha\varphi_\alpha^{-1},\varphi_\alpha\Lambda_\alpha\varphi_\alpha^{-1}],M_i^1(\mathsf{p}_\alpha)\big\rangle \,.
\end{equation}
We obtain 
\begin{multline}
\delta S_{{\rm H}, \Gamma} = \int_0^1\bigg(
\frac{1}{2\pi i}\oint_{c_{\mathsf p}} \langle \gamma\,\delta L^1\,\gamma^{-1},M_i^0\rangle\frac{\partial t^i}{\partial u^j}
+\frac{1}{2\pi i}\oint_{c_{\mathsf p}} \big\langle \gamma\,L^1\,\gamma^{-1},\partial_{u^j}(\delta \gamma\gamma^{-1}) \big\rangle\\
+ \sum_{\alpha=1}^N \big\langle\Lambda_\alpha,\varphi_\alpha^{-1}\partial_{u^j}(\delta\varphi_\alpha\varphi_\alpha^{-1})\varphi_\alpha \big\rangle
-
\sum_{\alpha=1}^N \big\langle [\delta\varphi_\alpha\varphi_\alpha^{-1},\varphi_\alpha\Lambda_\alpha\varphi_\alpha^{-1}],M_i^1(\mathsf{p}_\alpha)\big\rangle\frac{\partial t^i}{\partial u^j}\bigg)\frac{\d u^j}{\d s}\d s \,.
\end{multline}
We now rewrite this using \eqref{constraint} and \eqref{variation constraint}:
\begin{multline}
\delta S_{{\rm H}, \Gamma} = \int_0^1\bigg(
\frac{1}{2\pi i}\oint_{c_{\mathsf p}} \big\langle \delta L^0+[L^0,\delta\gamma\gamma^{-1}],M_i^0 \big\rangle\frac{\partial t^i}{\partial u^j}
+\frac{1}{2\pi i}\oint_{c_{\mathsf p}} \big\langle L^0,\partial_{u^j}(\delta \gamma\gamma^{-1}) \big\rangle\\
+ \sum_\alpha \big\langle \varphi_\alpha\,\Lambda_\alpha\,\varphi_\alpha^{-1},\partial_{u^j}(\delta\varphi_\alpha\varphi_\alpha^{-1})\big\rangle
-
\sum_\alpha \big\langle [\delta\varphi_\alpha\varphi_\alpha^{-1},\varphi_\alpha\Lambda_\alpha\varphi_\alpha^{-1}],M_i^1(\mathsf{p}_\alpha) \big\rangle\frac{\partial t^i}{\partial u^j}\bigg)\frac{\d u^j}{\d s}\d s \,.
\end{multline}
The contour integral of $\langle \delta L^0,M_i^0\rangle$ vanishes because it is holomorphic on $U_0$.  After integration by parts, the variation becomes
\begin{multline}\label{SH variation}
\delta S_{{\rm H}, \Gamma} = \int_0^1\bigg(
\frac{1}{2\pi i}\oint_{c_{\mathsf p}} \big\langle [M_i^0,L^0],\delta\gamma\gamma^{-1}\big\rangle\frac{\partial t^i}{\partial u^j}
-\frac{1}{2\pi i}\oint_{c_{\mathsf p}} \big\langle \partial_{u^j} L^0,\delta \gamma\gamma^{-1}\big\rangle\\
- \sum_\alpha \big\langle \partial_{u^j}(\varphi_\alpha\,\Lambda_\alpha\,\varphi_\alpha^{-1}),\delta\varphi_\alpha\varphi_\alpha^{-1} \big\rangle
+
\sum_\alpha \big\langle [M_i^1(\mathsf{p}_\alpha),\varphi_\alpha\Lambda_\alpha\varphi_\alpha^{-1}],\delta\varphi_\alpha\varphi_\alpha^{-1}\big\rangle\frac{\partial t^i}{\partial u^j}\bigg)\frac{\d u^j}{\d s}\d s \,.
\end{multline}

Let us restrict attention to variations for which $\delta\varphi_\alpha=0$.  Then $\delta S_{{\rm H}, \Gamma}$ vanishes for all curves $\Gamma$ if and only if
\begin{equation}\label{gamma variation}
0 = \oint_{c_{\mathsf p}} \big\langle \partial_{t^i} L^0-[M_i^0,L^0],\delta \gamma\gamma^{-1} \big\rangle = \oint_{c_{\mathsf p}} \big\langle \partial_{t^i} L^1-[M_i^1,L^1], \gamma^{-1}\delta \gamma \big\rangle
\end{equation}
for all variations $\delta\gamma\gamma^{-1}$ satisfying \eqref{variation constraint}, where the second equality follows from noting that
\begin{equation}
\gamma^{-1}\big( \partial_{t^i} L^0-[M_i^0,L^0] \big) \gamma=\partial_{t^i} L^1-[M_i^1,L^1]
\end{equation}
as a consequence of \eqref{M}. In particular, it follows that $\partial_{t^i}L^I-[M_i^I,L^I]$ can be interpreted as a section of $E_t$. It could have poles at the point $\mathsf{q}_i$ and any of the points $\mathsf{p}_\alpha$, because $M_i^1$ and $L^1$ have a simple poles there. However, from \eqref{M residue} the residue at $\mathsf{q}_i$ is given by 
\begin{equation}
-\Big[\Res_{\mathsf{q}_i}M_i^1,L_{z_I}^1(\mathsf{q}_i) \Big]=-\Big[\nabla P_i \big( L^1_{z_I}(\mathsf{q}_i) \big), L_{z_I}^1(\mathsf{q}_i) \Big]
\end{equation}
and this vanishes because $P_i$ is an invariant polynomial.  So $\partial_{t^i} L^1-[M_i^1,L^1]$ only has poles at the points $\mathsf{p}_\alpha$, and it determines a holomorphic section in  $H^0\big( C,\Lambda^{1,0}C\otimes E_t^\ast(\sum_\alpha \mathsf{p}_\alpha) \big)$.

To interpret \eqref{gamma variation} using Serre duality, we note using the constraint \eqref{variation constraint} that the components of the holomorphic section $Y \coloneqq (\delta\gamma\gamma^{-1}, \gamma^{-1} \delta \gamma)$ of $E_t$ over $U_0 \cap U_1$ defined relative to the local trivialisations over $U_0$ and $U_1$, respectively, satisfy
\begin{equation*}
[L^0,\delta \gamma\gamma^{-1}] = \gamma \delta L^1\gamma^{-1}-\delta L^0, \qquad
[L^1,\gamma^{-1}\delta \gamma] = \delta L^1 - \gamma^{-1} \delta L^0 \gamma\,.
\end{equation*}
Since $\delta L^0,\delta L^1$ are holomorphic on $U_0,U_1$ (because we are assuming that $\delta\varphi_\alpha=0$), these equations say that $[L, Y]$ lies in the image of the \v{C}ech coboundary map, so determines a trivial element of $H^1\big( C,\Lambda^{1,0}C\otimes E_t \big)$. In other words, $Y$ is in the kernel of the map
\begin{equation}\label{ad L}
[L,\cdot]:H^1\big(C,E_t ( -\textstyle\sum_\alpha \mathsf{p}_\alpha )\big) \longrightarrow H^1\big( C,\Lambda^{1,0}C\otimes E_t \big)\,.
\end{equation}
(The notation $(-\sum_\alpha \mathsf{p}_\alpha)$ in the domain of this map denotes the sheaf of sections vanishing at the points $\mathsf{p}_\alpha$, and appears because $L$ has poles at these points). Thus we may interpret the integrals in \eqref{gamma variation} as the Serre duality pairing between $\partial_{t^i} L - [M_i, L] \in H^0\big( C,\Lambda^{1,0}C\otimes E_t^\ast(\sum_\alpha \mathsf{p}_\alpha) \big)$ and $Y = (\delta\gamma\gamma^{-1}, \gamma^{-1} \delta \gamma) \in H^1\big(C,E_t(-\sum_\alpha \mathsf{p}_\alpha )\big)$. The pairing is required to vanish for all $Y$ in the kernel of the map \eqref{ad L}.  This happens if and only if $\partial_{t^i}L-[M_i,L]$ lies in the image of the adjoint of \eqref{ad L}.  The adjoint of \eqref{ad L} is the map
\begin{equation}\label{ad L adjoint}
-[L,\cdot]:H^0\left(C,E_t\right) \longrightarrow H^0\big( C,\Lambda^{1,0}C\otimes E_t(\textstyle\sum_\alpha \mathsf{p}_\alpha) \big)
\end{equation}
because, by Serre duality, the dual space of $H^1(C,V)$ is $H^0(C,\Lambda^{1,0}C\otimes V^\ast)$, and 
\begin{equation}
\oint_{c_{\mathsf p}}\langle [L,u],v\rangle =- \oint_{c_{\mathsf p}}\langle u,[L,v]\rangle
\end{equation}
for every $u\in H^1\big(C,E_t(-\textstyle\sum_\alpha \mathsf{p}_\alpha)\big)$ and $v\in H^0\left(C,E_t\right)$.
Since $\partial_{t^i}L-[M_i,L]$ is in the image of the map \eqref{ad L adjoint}, there exist holomorphic functions $N_i^I:U_I\to\g$ for $I=0,1$ satisfying
\begin{equation}
\partial_{t^i}L^I-[M^I_i,L^I]=[N^I_i,L^I],\qquad \gamma^{-1}N^0_i\gamma = N^1_i\,.
\end{equation}
Recall from earlier that the solution of \eqref{M} and \eqref{M residue} was unique up to addition of elements of $H^0(C,E_t)$.  Therefore we are free to define $\tilde{M}^I_i \coloneqq M^I_i+N^I_i$, and this choice of $M_i$ satisfies the Lax equations \eqref{Lax equations} and the constraints \eqref{M constraints}.  The choice of $M_i$ is now unique up to addition of elements of $\ker [L,\cdot]\subset H^0(C,E_t)$.  But the stability condition \eqref{inf_freeness} means that there are no holomorphic sections of $C$ that commute with $L$. So $M_i$ is unique, as claimed.

We have shown that $S_{{\rm H}, \Gamma}$ is critical with respect to variations of $L,\gamma$ with $\varphi_\alpha$ fixed if and only if $L$ solves the Lax equation \eqref{Lax equations}.  It remains to check that solutions of these equations are also critical points of $S_{{\rm H}, \Gamma}$ with respect to variations of $\varphi_{\alpha}$.  The variation of $S_{{\rm H}, \Gamma}$ is given in \eqref{SH variation}.  The term involving $\delta\gamma$ vanishes because $L$ satisfies the Lax equations \eqref{Lax equations}.  The term involving $\delta\varphi_\alpha$ also vanishes because taking residues of \eqref{Lax equations} gives
\begin{equation}
\partial_{t^i} (\varphi_\alpha\Lambda_\alpha\varphi_\alpha^{-1}) = \big[ M_i^1(\mathsf{p}_\alpha),\varphi_\alpha\Lambda_\alpha\varphi_\alpha^{-1} \big].
\end{equation}
So the variation vanishes and $S_{{\rm H}, \Gamma}$ is critical with respect to variations of $L,\gamma,\varphi_\alpha$ \,.

Finally, to show the zero-curvature equation \eqref{eq:zce Mi} we note that the Lax equations \eqref{Lax equations} imply
\begin{equation}\label{ZCE Mi com L}
\big[ \partial_{t^i} M^I_j - \partial_{t^j} M^I_i - [ M^I_i, M^I_j ], L^I \big] =0
\end{equation}
for $I \in \{ 0, 1\}$. On the other hand, it also follows from the first equation in \eqref{M constraints} that
\begin{equation}
\gamma^{-1} \big( \partial_{t^i} M^0_j - \partial_{t^j} M^0_i - [ M^0_i, M^0_j ] \big) \gamma = \partial_{t^i} M^1_j - \partial_{t^j} M^1_i - [ M^1_i, M^1_j ]
\end{equation}
on $U_0 \cap U_1$. The zero-curvature equation \eqref{eq:zce Mi} now follows by the stability condition.
\end{proof}

\section{Examples}\label{sec: examples}

The unifying Lagrangian $1$-form \eqref{eq:unifyingaction b} has the typical structure of a Lagrangian $1$-form in phase space coordinates as first used in \cite{CDS} for finite-dimensional integrable systems on coadjoint orbits, and further developed in \cite{CH} in relation to the univariational principle. It is the difference of a kinetic part and a potential part involving all the invariant Hamiltonians $H_i$.
The essential novelty compared to all Lagrangian $1$-forms considered so far is that the kinetic part involves not only the group coordinates $\varphi_\alpha$ of the coadjoint orbits at the marked points $\mathsf p_\alpha$ for $\alpha = 1,\ldots, N$, but also the transition function $\gamma : (U_0 \cap U_1) \times \RR^n \to G$ of the holomorphic principal $G$-bundle. This introduces a new type of kinetic term in the action which is crucial for generating the complete set of kinetic contributions. We will see explicit examples of this in the genus $1$ case considered in \S\ref{sec:ellipticgaudin}, where we derive a Lagrangian $1$-form for the elliptic Gaudin model and also the elliptic spin Calogero-Moser model as a special case. In the genus $0$ case considered in \S\ref{sec:rationalgaudin}, this new kinetic term is absent and we recover the familiar Lagrangian $1$-form first introduced in \cite{CDS,CSV}.

Since this section deals with explicit examples and we only work with a matrix Lie algebra $\g$, we fix the nondegenerate invariant bilinear pairing $\langle \cdot, \cdot \rangle : \g \otimes \g \to \CC$ to be the usual one given in terms of the trace by $\langle A,B \rangle = \Tr(AB)$. Also, throughout this section, we identify $\g^\ast$ with $\g$ and the coadjoint action with the adjoint action. Further, we fix a Cartan subalgebra $\h \subset \g$ with basis $\{\sfH_\mu\}_{\mu = 1}^{\text{rk}\, \g}$ and the corresponding root generators $\{\sfE_\varrho\}$ of $\g$.

\subsection{Rational Gaudin hierarchy}\label{sec:rationalgaudin}

We begin with the simplest case of genus $0$ so that $C = \CP$ is the Riemann sphere. We will show that when the unifying Lagrangian $1$-form \eqref{eq:unifyingaction b} is specialised to the case $\gamma = 1$ we obtain the rational Gaudin hierarchy. Gaudin models are a broad class of integrable systems associated with Lie algebras equipped with a nondegenerate invariant bilinear form. They were first introduced for the Lie algebra $\mathfrak{sl}_2(\mathbb{C})$ by Gaudin in \cite{Ga1} as quantum integrable spin chains with long-range interactions, and later extended to arbitrary semisimple Lie algebras in \cite{Ga2}. In the rational case, the classical Gaudin model provides the simplest example of a Hitchin system with marked points and has been extensively studied from this perspective; see, for instance, \cite{Ne, ER}.

We choose $N \in \ZZ_{\geq 1}$ distinct marked points $\{ \sfp_\alpha \}_{\alpha=1}^N$. We also need sufficiently many points $\{\sfq_i\}_{i=1}^n$ which are distinct from the $\{ \sfp_\alpha \}_{\alpha=1}^N$, where $n$ will be half the dimension of the phase space, to obtain a sufficient number of Hamiltonians. Let $U_0$ be a neighbourhood of a point $\sfp \in C$ distinct from $\{ \sfp_\alpha \}_{\alpha=1}^N$ and $\{\sfq_i\}_{i=1}^n$, such that $U_0$ does not contain any of these other points. Further, let $U_1 \coloneqq C \setminus \{\sfp\}$ which we equip with a holomorphic coordinate $z \colon U_1 \rightarrow \mathbb{C}$, where, in contrast to the general case of \S \ref{sec: Hitchin holomorphic}, we have dropped the index on $z$ for notational simplicity.

\paragraph{Lax matrix:} Following the general setup in \S \ref{sec: Lax matrices}, we obtain a Lax matrix for the rational Gaudin model as the solution of the moment map condition $\mu=0$ \eqref{eq:b_z-gauged} in a local trivialisation where $A''^I = 0$. Using \eqref{Lax behaviour a}, we can write
\begin{equation}\label{eq:L1-genuszero}
L^1 = L_{z} \d z = \left( \frac{L_\alpha}{z - z(\mathsf p_\alpha)} + J_\alpha^1 \right) \d z\,.
\end{equation}
locally near the point $\mathsf p_\alpha$. Here $L_\alpha = - \varphi_\alpha \Lambda_\alpha \varphi_\alpha^{-1}$ denotes the residue at $\mathsf p_\alpha$ and $J^1_\alpha$ the holomorphic part of $L_z$ in the neighbourhood of $\mathsf p_\alpha$. Since $\gamma = 1$ we have $L^0 = L^1$ on $U_0 \cap U_1$ so that $L^1$ extends holomorphically to the point $\mathsf p \in C = \CP$ at infinity in the coordinate $z$.
We can therefore write the Lax matrix $L^1$ as a global $\g$-valued meromorphic $(1,0)$-form
\begin{equation}\label{eq:B-Gaudin}
L^1 = L_z \d z = \sum_{\alpha=1}^N \frac{L_\alpha}{z-z(\mathsf p_\alpha)} \d z
\end{equation}
together with the constraint $\sum_\alpha L_\alpha = 0$ coming from the fact that $L^1$ is regular at $z(\mathsf p) = \infty$.

\paragraph{Lagrangian 1-form:} Let us look at the unifying Lagrangian 1-form \eqref{eq:unifyingaction b} in the present case. Since $\gamma = 1$, the kinetic term in the Lagrangian $1$-form \eqref{eq:unifyingaction b} involving the transition function drops out, and we are left with
\begin{equation}\label{eq:rgmultiform}
\Lag_{\rm RG} = \sum_{\alpha=1}^N \Tr \left( \Lambda_\alpha \varphi_\alpha^{-1} \d_{\RR^n} \varphi_\alpha \right) - H_i\, \d_{\RR^n} t^i
\end{equation}
where
\begin{equation}
H_i = P_i( L_{z}(\sfq_i))
\end{equation}
with $G$-invariant polynomials $P_i$ acting on $L_{z}$ given by \eqref{eq:B-Gaudin}. The Lagrangian 1-form \eqref{eq:rgmultiform} describes the rational Gaudin hierarchy. A Lagrangian 1-form for the rational Gaudin hierarchy was first obtained in \cite[Section 7]{CDS} using an algebraic approach. The Lagrangian $1$-form describing the dynamics for a single Hamiltonian of the Gaudin model was also obtained previously in \cite{VW}. By the same derivation as in those references, one shows that the Euler-Lagrange equations associated to \eqref{eq:rgmultiform} give the following collection of equations 
\begin{equation}
  \partial_{t^i} L_\alpha = \left[ \frac{\nabla P_i(L_{z}(\sfq_i))}{z(\mathsf p_\alpha) - z(\sfq_i)}, L_\alpha \right],
\end{equation}
for all $i$. One can check that one recovers a hierarchy of Lax equations from the above collection of equations using \eqref{MIi explicit}.

\subsection{Elliptic Gaudin and elliptic spin Calogero-Moser hierarchies}\label{sec:ellipticgaudin}

Let us now consider the genus $1$ case so that the Riemann surface $C$ is a torus $\mathbb{C}/(\mathbb{Z} + \tau \mathbb{Z})$ with $\text{Im}(\tau)>0$. Following the general setup from \S\ref{sec: Hitchin holomorphic} we fix a point $\sfp \in C$ and take $U_0$ to be a small neighbourhood of $\sfp$, and we let $U_1 \coloneqq C \setminus \{\sfp\}$. We fix a coordinate $z$ on $C$ with the identifications $z \sim z + 1$ and $z \sim z + \tau$ such that $z(\sfp) = 0$.
Our starting point is the Lagrangian $1$-form \eqref{eq:unifyingaction b} but now specialised to a holomorphic transition function $\gamma$ on the annulus $U_0 \cap U_1$ given by
\begin{equation}\label{eq:tranfun-genusone}
    \gamma = \exp\left(\frac{\Q}{z} \right) \quad \text{with}\quad \Q = q^\mu \sfH_\mu
\end{equation}
where $q^\mu \in \mathbb{C}$, for each $\mu \in \{1, \ldots, \text{rk}\, \g\}$, is constant in $z$. We will show that the resulting Lagrangian $1$-form describes the elliptic Gaudin hierarchy, and in the case with $N=1$ marked point the elliptic spin Calogero-Moser hierarchy. 

The elliptic Gaudin model we construct was obtained from a Hamiltonian reduction procedure in \cite{ER, Ne}. To avoid possible confusion, it is worth noting in passing that for $\g = \sl_m$ there is another integrable system that also goes by the name elliptic Gaudin model which was originally derived as a limit of the XYZ spin chain in \cite{ST1, ST2}. It is not clear if and how these two models are related, especially since the elliptic Gaudin model we consider is known to have a dynamical $r$-matrix \cite{ER} while the one of \cite{ST1, ST2} can be built from Belavin's elliptic solution \cite{Be} of the classical Yang-Baxter equation by following the procedure of \cite{J}. A spin generalisation of the Calogero-Moser model was first defined in \cite{GiH}, and its case with elliptic potential was realised as a Hitchin system in \cite{Ma, GoNe, ER, Ne}. The main goal of this section is to use the unifying Lagrangian 1-form \eqref{eq:unifyingaction b} to obtain variational descriptions of these two hierarchies.\footnote{Recently, the Lagrangian 1-form structure of certain Calogero-Moser-type systems was studied in \cite{KNY}. Starting from a general ansatz for the kinetic term of the Lagrangian 1-form, the authors used the resulting generalised Euler-Lagrange equations to derive the integrable cases of these systems.}

We choose $N \in \ZZ_{\geq 1}$ distinct marked points $\{ \sfp_\alpha \}_{\alpha=1}^N$ in $C \setminus \overline{U_0}$. As in the general case, we need sufficiently many additional points $\sfq_i \in C \setminus \overline{U_0}$ distinct from $\{\sfp_\alpha\}_{\alpha=1}^N$ to obtain a sufficient number of Hamiltonians.
We will later specialise to the case of a single marked point $\sfp_1$, i.e. taking $N=1$.

\paragraph{Lax matrix:} From \eqref{Lax behaviour a}, we have
\begin{equation}\label{eq:L1-eg}
L^1 = L_z \d z = \left( \frac{L_\alpha}{z - z(\mathsf p_\alpha)} + J^1_\alpha \right) \d z
\end{equation}
locally around the point $\mathsf p_\alpha$ in our choice of local trivialisation. Here $L_\alpha = -\varphi_\alpha \Lambda_\alpha \varphi_\alpha^{-1}$ denotes the residue at $\mathsf p_\alpha$ and $J^1_\alpha$ the holomorphic part of $L_z$ in the neighbourhood of $\mathsf p_\alpha$. We can express $L_\alpha$ and $L^1 = L_z \d z$ in the basis $(\sfH_\mu, \sfE_\varrho)$ as
\begin{equation}\label{eq:L1-cartanbasis}
  L_\alpha = (L_\alpha)^\mu \sfH_\mu + (L_\alpha)^ \varrho \sfE_\varrho \quad \text{and} \quad L_z = L^\mu \sfH_\mu + L^\varrho \sfE_\varrho.
\end{equation}
Since $L^1 = L_z \d z$ is a meromorphic $1$-form on the punctured torus $U_1$, it satisfies
\begin{equation}
    L^1(z+1) = L^1(z+\tau) = L^1(z)\,. 
\end{equation}
Also, recall from \S \ref{sec: Lax matrices} that $L^0$ is holomorphic on $U_0$ and we have the relation $L^0 = \gamma L^1 \gamma^{-1}$ on $U_0 \cap U_1$. Therefore, we can write
\begin{equation}\label{eq:L0-cartanbasis}
    L^0 = \bigg( L^\mu \sfH_\mu +  L^\varrho \exp\Big(\frac{\varrho(\Q)}{z} \Big) \sfE_\varrho \bigg) \d z\,.
\end{equation}
We can express $L^\mu$ and $L^\varrho$ in \eqref{eq:L1-cartanbasis} and \eqref{eq:L0-cartanbasis} in terms of the Weierstrass $\zeta$-function and the Weierstrass $\sgm$-function which are respectively defined by the relations
\begin{equation}\label{eq:zetafunction}
  \frac{\d \zeta(z)}{\d z} = - \wp(z), \quad \lim_{z \rightarrow 0} \left( \zeta(z) - \frac{1}{z} \right) = 0
\end{equation}
and
\begin{equation}\label{eq:sigmafunction}
  \frac{\d \log (\sgm(z))}{\d z} = \zeta(z), \quad \lim_{z \rightarrow 0} \frac{\sgm(z)}{z} = 1
\end{equation}
where $\wp$ is the Weierstrass $\wp$-function. While the Weierstrass $\wp$-function is doubly periodic, the $\sigma$-function and the $\zeta$-function satisfy
\begin{align}
\label{quasi_period1}  &\zeta(z + 2 \omg_l) = \zeta(z) + 2 \eta_l \,,\\
\label{quasi_period2}  &\sgm(z + 2 \omg_l) = - \sgm(z) e^{2 \eta_l (z + 2 \omg_l)}\,,
\end{align} 
where $2\omg_1$ and $2\omg_2$ are the periods of a generic torus and $\eta_l = \zeta(\omg_l)$, for $l = 1, 2$. Here we have $2\omg_1 = 1$ and $2\omg_2 = \tau$.

It follows from \eqref{eq:L1-eg} that the elliptic function $L^\mu$ is meromorphic with simple poles at $\{\sfp_\alpha\}_{\alpha=1}^N$ with residues $(L_\alpha)^\mu$, so it can be expressed as
\begin{subequations} \label{L mu and constraint}
\begin{equation}\label{eq:Lmu-eg}
  L^\mu =  \pi^\mu + \sum_{\alpha=1}^N (L_\alpha)^\mu \zeta(z-z(\sfp_\alpha)) 
\end{equation}
where $\pi^\mu$ is constant in $z$, and we have
\begin{equation}\label{eq:sum_residue}
    \sum_{\alpha=1}^N (L_\alpha)^\mu = 0
\end{equation}
\end{subequations}
since the sum of residues over an irreducible set of poles of an elliptic function vanishes. Next, for the function $L^\varrho$, \eqref{eq:L1-eg} tells us that it has poles at $\{\sfp_\alpha\}_{\alpha=1}^N$ with residues $(L_\alpha)^\varrho$ and \eqref{eq:L0-cartanbasis} implies that the function $L^\varrho e^{\frac{\varrho(\Q)}{z}}$ is holomorphic near $z = 0$. It then follows that
\begin{equation}\label{eq:Lvarrho-eg}
    L^\varrho = \sum_{\alpha=1}^{N} (L_\alpha)^\varrho \frac{\sgm(\varrho(\Q) + z - z(\sfp_\alpha))}{\sgm(\varrho(\Q)) \sgm(z-z(\sfp_\alpha)) } e^{- \varrho(\Q)\zeta(z)}\,.
\end{equation}
Then, we get $L^1 = (L^\mu \sfH_\mu + L^\varrho \sfE_\varrho) \d z$ with $L^\mu$ and $L^\varrho$ given by \eqref{L mu and constraint} and \eqref{eq:Lvarrho-eg} as the Lax matrix of the elliptic Gaudin model.

In the special case when $N = 1$ and $\sfp_1 = 0$, the components \eqref{L mu and constraint} simply reduce to $L^\mu =  \pi^\mu$, while the components \eqref{eq:Lvarrho-eg} read
\begin{equation}
L^\varrho = \sum_{\alpha=1}^{N} (L_\alpha)^\varrho \frac{\sgm(\varrho(\Q) + z)}{\sgm(\varrho(\Q)) \sgm(z) } e^{- \varrho(\Q)\zeta(z)}\,.
\end{equation} 
When $\g = \mathfrak{sl}_m(\mathbb{C})$, coincides with the Lax matrix of the elliptic spin Calogero-Moser model, see, e.g. \cite[Chapter 7]{BBT}.

\paragraph{Lagrangian 1-form:} We now specialise the unifying Lagrangian 1-form \eqref{eq:unifyingaction b} to the present case. Recalling that $\g$ is taken to be a matrix Lie algebra we take the nondegenerate bilinear form on $\g$ is given by the trace, this takes the form
\begin{equation}\label{eq:egmultiform-unsimp}
  \Lag_{\rm EG} = \frac{1}{2 \pi i} \int_{c_{\mathsf p}} \Tr \left( L^1 \wedge \gamma^{-1} \d_{\RR^n} \gamma \right) + \sum_{\alpha=1}^N \Tr \left( \Lambda_\alpha \varphi_\alpha^{-1} \d_{\RR^n} \varphi_\alpha \right) - H_i \, \d_{\RR^n} t^i
\end{equation}
where the transition function $\gamma$ is given by \eqref{eq:tranfun-genusone} and
\begin{equation}\label{eq:Ham-genusone}
H_i = P_i( L_{z}(\mathsf q_i))\,.
\end{equation}
Let us evaluate the integral in the first term on the right-hand side of \eqref{eq:egmultiform-unsimp}. We have
\begin{equation}\label{eq:integraleval}
  \begin{split}
    \frac{1}{2 \pi i} \int_{c_{\mathsf p}} \Tr \left( L_z \d z \wedge \gamma^{-1} \d_{\RR^n} \gamma \right) = \frac{1}{2 \pi i} \int_{c_{\mathsf p}} \Tr \left( L_z \d z \wedge \frac{\d_{\RR^n} \Q}{z}  \right) = p_\mu \d_{\RR^n} q^\mu
  \end{split}
\end{equation}
where the coordinates $q^\mu$ were introduced in \eqref{eq:tranfun-genusone} and we defined
\begin{equation}
p_\mu \coloneqq \Tr \big( L_{z}(z(\sfp)) \sfH_\mu \big)\,,~~\mu \in \{1, \ldots, \text{rk}\, \g\}\,.
\end{equation}
Plugging back \eqref{eq:integraleval} into \eqref{eq:egmultiform-unsimp}, we get a Lagrangian 1-form for the elliptic Gaudin hierarchy
\begin{equation}\label{eq:egmultiform-simp}
  \Lag_{\rm EG} = p_\mu \d_{\RR^n} q^\mu + \sum_{\alpha=1}^N \Tr \left( \Lambda_\alpha \varphi_\alpha^{-1} \d_{\RR^n} \varphi_\alpha \right) - H_i\, \d_{\RR^n} t^i\,,
\end{equation}
with $H_i$ defined by \eqref{eq:Ham-genusone} depending on the elliptic Gaudin Lax matrix given by \eqref{L mu and constraint} and \eqref{eq:Lvarrho-eg}. Note that in contrast with the Lagrangian 1-form \eqref{eq:rgmultiform} obtained in the genus $0$ case, the Lagrangian 1-form above has an additional kinetic term corresponding to the cotangent bundle degrees of freedom $p_\mu,q^\mu$ arising from the non-triviality of the principal $G$-bundle in genus 1.

Writing a Lagrangian 1-form for the elliptic spin Calogero-Moser hierarchy only requires us to set $N=1$ in the Lagrangian 1-form for the elliptic Gaudin hierarchy. We then get
\begin{equation}\label{eq:escmmultiform}
  \Lag_{\rm ESCM} = p_\mu \d_{\RR^n} q^\mu + \Tr \left( \Lambda_1 \varphi_1^{-1} \d_{\RR^n} \varphi_1 \right) - H_i \, \d_{\RR^n} t^i
\end{equation}
for this special case.

\paragraph{Alternative description:} Lax matrices for the elliptic Gaudin and elliptic spin Calogero-Moser models have appeared before in the literature and we now explain how they are related to the ones we used above to formulate our Lagrangian $1$-form, hence justifying that we have indeed built a variational description for the elliptic Gaudin hierarchy \cite{ER, Ne}. The two descriptions are related by a change of local trivialisation, as we now show.

Recall that under a change of local trivialisation implemented through smooth functions $f_0$ on $U_0$ and $f_1$ on $U_1$, the $(0, 1)$-connection $A''^I$ and the $(1, 0)$-form $B^I$ transform as
\begin{subequations}
\begin{align}
A''^I &\longmapsto \widetilde{A}''^I = f_I A''^I f_I^{-1} - \dol f_I f_I^{-1}\,,\\
L^I &\longmapsto \widetilde{L}^I = f_I L^I f_I^{-1}
\end{align}
\end{subequations}
locally in $U_I$ for each $I \in \{0, 1\}$, while the new transition function $\widetilde{\gamma}$ reads
\begin{equation}
    \widetilde{\gamma} = f_0 \gamma f_1^{-1} \,,
\end{equation}
and $\widetilde{\varphi}_\alpha=f_1(z(p_\alpha))\,\varphi_\alpha$ so that 
\begin{equation}
    \widetilde{L}_\alpha=f_1(z(p_\alpha))\,L_\alpha\,f_1(z(p_\alpha))^{-1}\,.
\end{equation}

In our choice of local trivialisation with respect to which $A''^I = 0$ for $I \in \{ 0, 1\}$, we obtain a Lax matrix $L$ for the elliptic Gaudin model described by \eqref{L mu and constraint} and \eqref{eq:Lvarrho-eg}. The alternative description of the elliptic Gaudin model, which can be found in \cite{Ne}, for instance, corresponds to a local trivialisation with respect to which
\begin{equation}\label{eq:cartanlocaltriv}
    \widetilde{A}''^I = \mathcal{V} = v^\mu \sfH_\mu \,,
\end{equation}
where $v^\mu$ is constant in $(z, \zbar)$ for each $I \in \{0, 1\}$, and with trivial transition function. To implement the change of local trivialisation from the one with respect to which $A''^I = 0$ to the one defined by \eqref{eq:cartanlocaltriv}, we need functions $f_0, f_1$ to satisfy
\begin{equation}\label{eq:changetrivcond1}
    - \dol f_I f_I^{-1} = \widetilde{A}''^I = \mathcal{V}\,, \quad I \in \{0, 1\}\,.
\end{equation}
The new transition function is given by 
\begin{equation}\label{eq:changetrivcond3}
    \widetilde{\gamma} = f_0 \gamma f_1^{-1} \,.
\end{equation}
Then choosing $f_0, f_1$ such that $f_0  = f_1 \exp\left(-\frac{\Q}{z}\right)$ we obtain $\widetilde{\gamma} = 1$.
Finally, since $f_1$ is a smooth function on the punctured torus $U_1$, it must satisfy:
\begin{equation}\label{eq:changetrivcond2}
    f_1(z+1) = f_1(z+\tau) = f_1(z)\,.
\end{equation}
From \eqref{eq:changetrivcond1} and \eqref{eq:changetrivcond2}, we get
\begin{equation}\label{eq:changetrivfunc}
    f_1 = \exp\left(\Q \left( \zeta(z) +  2 \zeta(\tfrac{1}{2})\frac{z \taubar - \zbar \tau}{\tau - \taubar} -  2 \zeta(\tfrac{\tau}{2})\frac{z - \zbar}{\tau - \taubar} \right)\right)
\end{equation}
together with the relation
\begin{equation}
    \mathcal{V} = - \mathcal{Q} \frac{2}{\tau - \taubar} \left( \zeta(\tfrac{\tau}{2}) - \tau \zeta(\tfrac{1}{2}) \right)=\frac{2\pi i}{\tau - \taubar}\mathcal{Q}\,,
\end{equation}
where in the final step we used Legendre's relation between the periods and quasiperiods of $\zeta(z)$.
Expressing $\widetilde{L}^1 =f_1L^1f_1^{-1} = \widetilde{L}_z \d z$ in terms of the basis $(\sfH_\mu, \sfE_\varrho)$ of $\g$ as
\begin{equation}
    \widetilde{L}_z = \widetilde{L}^\mu \sfH_\mu + \widetilde{L}^\varrho \sfE_\varrho=L^\mu \sfH_\mu + L^\varrho f_1\sfE_\varrho f_1^{-1}\,,
\end{equation}   
we find  
\begin{equation}\label{eq:Lmu-eg-nek}
    \widetilde{L}^\mu=L^\mu= \pi^\mu + \sum_{\alpha=1}^N (L_\alpha)^\mu \zeta(z-z(\sfp_\alpha))= \pi^\mu + \sum_{\alpha=1}^N (\widetilde{L_\alpha})^\mu \zeta(z-z(\sfp_\alpha))
\end{equation}
and 
\begin{align}
    \widetilde{L}^\varrho&=\sum_{\alpha=1}^{N} (L_\alpha)^\varrho \frac{\sgm(\varrho(\Q) + z - z(\sfp_\alpha))}{\sgm(\varrho(\Q)) \sgm(z-z(\sfp_\alpha)) } e^{\varrho(\Q) \left(   2 \zeta(\frac{1}{2})\frac{z \taubar - \zbar \tau}{\tau - \taubar} -  2 \zeta(\frac{\tau}{2})\frac{z - \zbar}{\tau - \taubar} \right)}\nonumber\\
    &=\sum_{\alpha=1}^{N} (\widetilde{L}_\alpha)^\varrho \frac{\sgm(\varrho(\Q) + z - z(\sfp_\alpha))}{\sgm(\varrho(\Q)) \sgm(z-z(\sfp_\alpha)) } e^{\varrho(\Q) \left(   2 \zeta(\frac{1}{2})\frac{(z-z(\sfp_\alpha)) \taubar - (\zbar -\zbar(\sfp_\alpha)) \tau}{\tau - \taubar} -  2 \zeta(\frac{\tau}{2})\frac{(z-z(\sfp_\alpha)) - (\zbar-\zbar(\sfp_\alpha))}{\tau - \taubar} \right)} \,.
\end{align}
Note that, using again Legendre's relation,
\begin{equation}
   2 \zeta(\tfrac{1}{2})\frac{z \taubar - \zbar \tau}{\tau - \taubar} -  2 \zeta(\tfrac{\tau}{2})\frac{z - \zbar}{\tau - \taubar}= -2 \zeta(\tfrac{1}{2})z+\frac{2\pi i}{\tau-\taubar}(z - \zbar)
\end{equation}
so that we can write
\begin{align}\label{eq:Lvarrho-eg-nek}
    \widetilde{L}^\varrho
    &=\sum_{\alpha=1}^{N} (\widetilde{L}_\alpha)^\varrho \frac{\sgm(\varrho(\Q) + z - z(\sfp_\alpha))}{\sgm(\varrho(\Q)) \sgm(z-z(\sfp_\alpha)) } e^{\varrho(\Q) \left(   \frac{2\pi i}{\tau-\taubar}(z - \zbar) \right)}e^{-\varrho(\Q) \left(    \frac{2\pi i}{\tau-\taubar}(z(\sfp_\alpha) - \zbar(\sfp_\alpha)) \right)}e^{ -2 \zeta(\frac{1}{2})\varrho(\Q) (z-z(\sfp_\alpha))} \,.
\end{align}
This last expression allows us to compare the Lax matrix $\widetilde{L}^1 = (\widetilde{L}^\mu \sfH_\mu + \widetilde{L}^\varrho \sfE_\varrho) \d z$ with $\widetilde{L}^\mu$ and $\widetilde{L}^\varrho$ given by \eqref{eq:Lmu-eg-nek} and \eqref{eq:Lvarrho-eg-nek}, with the elliptic Gaudin Lax matrix obtained in \cite[(4.5)]{Ne} (when specialising to the case $\g = \mathfrak{gl}_m(\CC)$ so that writing \eqref{eq:tranfun-genusone} as $\mathcal Q = \sum_{k=1}^m x_k \mathsf E_{kk}$ in the standard basis $\{ \mathsf E_{ij} \}_{i,j=1}^m$ for $\mathfrak{gl}_m(\CC)$, we have $\varrho(\mathcal Q) = x_i-x_j$ for the $\mathfrak{gl}_m(\CC)$ root $\varrho = \epsilon_i - \epsilon_j$). They coincide up to the factor 
$e^{-2\varrho(\Q)\zeta(\frac{1}{2})(z-z(\sfp_\alpha))}$ 
which possibly arises from the use of a different convention for the $\sigma$-function and the $\zeta$-function in \cite{Ne} relative to the standard \eqref{quasi_period1}-\eqref{quasi_period2}.

\section{Conclusion and outlook}

In this paper, we derived a variational formulation of Hitchin's completely integrable system associated with a compact Riemann surface of arbitrary genus by extending the theory of Lagrangian multiforms to the setting of gauge theories. As an application, in the genus one case we obtained an explicit phase-space Lagrangian multiform for the elliptic Gaudin hierarchy (and also the elliptic spin Calogero-Moser hierarchy as a special case) for the first time. The present work thus extends the results of \cite{CDS, CSV}, where explicit Lagrangian multiforms were constructed for both non-cyclotomic and cyclotomic rational Gaudin models using the algebraic framework of Lie dialgebras. In that respect, a natural question is how to generalise the well-understood connection between classical $r$-matrices and Lagrangian multiforms in the non-dynamical case \cite{CStV2, CDS, CSV} to the present setting. Since the Hitchin system associated with an elliptic curve is known to have a dynamical $r$-matrix \cite{ER}, the work \cite{BDOZ} appears particularly promising in bridging this gap.

One of the virtues of our construction is that it makes manifest the explicit connection between classical $3$d mixed holomorphic–topological BF theory and Hitchin's completely integrable system. The relevance of this $3$d holomorphic–topological gauge theory for the study of the Hitchin system has been noted previously in the literature \cite{GW, Z}, see also \cite{GT}, in fact, already at the quantum level. In those works, however, its appearance is less direct: $3$d mixed BF theory can be obtained as a certain level-zero limit of $3$d Chern-Simons theory which upon quantisation undergoes the familiar shift to the critical level. The latter is in turn deeply related to the representation theory of affine Kac–Moody algebras at the critical level, via conformal blocks of the WZW model living on the boundary, which underlies the geometric Langlands correspondence and the quantisation of the Hitchin system through the work of Beilinson-Drinfel'd \cite{BD}.

As emphasised in the introduction, the fruitful merging of the framework of Lagrangian $1$-forms with that of $3$d mixed holomorphic-topological BF theory was achieved through the study of Hitchin's completely integrable system in the Lagrangian framework. A natural generalisation of the present work will be to similarly bring together the framework of Lagrangian $2$-forms with that of the celebrated $4$d semi-holomorphic Chern-Simons theory \cite{CoY} with the aim of obtaining a gauge-theoretic origin for hierarchies of $2$d integrable field theories. In connection with these ideas, it is important to note also the work \cite{LeOZ2} on the construction of a $2$d field-theoretic generalisation of Hitchin's integrable system based on affine Higgs bundles, which was shown to be closely related to the $4$d Chern-Simons setup. Examples of affine Hitchin systems constructed in this way include $2$d field theoretic analogues of the elliptic Gaudin and Calogero-Moser models, both of which we have cast in the Lagrangian multiform framework in the present work (see also, for instance, \cite{AZ} and references therein concerning the so-called field analogue of elliptic (spin) Calogero-Moser model). Naturally, a related goal is then to use a higher-dimensional generalisation of our present work to obtain a variational analogue of the affine Higgs bundle setup.

Finally, a central motivation behind the Lagrangian multiform programme is to provide a path integral quantisation framework for integrable hierarchies. By constructing Lagrangian multiforms for Hitchin systems and demonstrating their connection to $3$d classical mixed BF theory in the present work, we have opened the way for studying the quantisation of Hitchin systems via path integral methods. This approach also offers the potential to uncover new insights into their relations with gauge theories at the quantum level.

\appendix

\section{Stability for \texorpdfstring{$G=SL_m(\CC)$}{G=SLm}} \label{sec: Stability}

In this appendix we show that stable bundles with $G=SL_m(\mathbb{C})$ satisfy our differential stability condition \eqref{inf_freeness}.  A stable $SL_m(\CC)$ principal Higgs bundle $\mathcal P$ is equivalent to a stable vector Higgs bundle $E = \mathcal P \times_{SL_m(\CC)}\CC^n$ for which the holomorphic line bundle $\det(E)$ is trivial.  We recall that a vector Higgs bundle consists of a holomorphic vector bundle $E\to C$ and a holomorphic section $B$ of $\Lambda^{1,0} C \otimes\mathrm{End}(E)$. Such a bundle is called \emph{stable} if, for all holomorphic sub-bundles $F\subset E$ satisfying $B(F)\subset \Lambda^{1,0}C\otimes F$, we have
\begin{equation}
\frac{\mathrm{deg}(F)}{\mathrm{rank}(F)}<\frac{\mathrm{deg}(E)}{\mathrm{rank}(E)} \,.
\end{equation}

Let $h$ be any Hermitian metric on the holomorphic vector bundle $E$.  Then $h$, together with $E$, determines a Chern connection $A$.  This is the unique unitary connection such that the holomorphic structure $\bar{\partial}$ of $E$ coincides with the $(0,1)$-part $\bar\partial^A$ of the covariant derivative $\d^A$.  We let $B^\dagger$ denote the Hermitian-adjoint of $B$; this is a smooth section of $\Lambda^{0,1}C\otimes \mathrm{End}(E)$. Then we can form two operators
\begin{align}
D' &= \bar\partial^A + B \,, \\
D''&= \partial^A + B^\dagger \,.
\end{align}
Simpson showed \cite{Sim} that every stable bundle admits a unique Hermitian metric such that
\begin{equation}\label{HE}
    D'D''+D''D' = \lambda\,\omega\otimes\mathrm{id}_E\,,
\end{equation}
in which $\omega$ is a chosen area form for the Riemann surface and $\lambda\in\mathbb{C}$ is a constant.

We will prove that a stable vector Higgs bundle satisfies the property \eqref{inf_freeness}, following a standard argument.  To do so, suppose for contradiction that $X$ is a non-zero traceless section of $\mathrm{End}(E)$ satisfying $\bar\partial X=0$ and $[B,X]=0$.  Let $h$ be the Hermitian metric on $E$ satisfying \eqref{HE}.  We note that
\begin{equation}\label{DD1}
    i\int_C\mathrm{tr}(D''X^\dagger\wedge D'X - D'X^\dagger\wedge D''X)=\int_C\mathrm{tr}(\star \d^AX^\dagger\wedge \d^AX+\star[\Phi,X]^\dagger\wedge[\Phi,X]) \,,
\end{equation}
in which $\star$ is the Hodge star, $\d^A=\partial^A+\bar\partial^A$ and $\Phi=B-B^\dagger$.  On the other hand, \eqref{HE} and integration by parts show that
\begin{equation}\label{DD2}
\int_C\mathrm{tr}(D''X^\dagger\wedge D'X + D'X^\dagger\wedge D''X)=-\int_C\lambda\mathrm{tr}(X^\dagger[\mathrm{id}_E,X])\omega=0 \,.
\end{equation}
Our assumptions on $X$ mean that $D''X=0$.  So, \eqref{DD2} implies that the left-hand side of \eqref{DD1} vanishes.  Since the right-hand side of \eqref{DD1} is the integral of a non-negative function, its integrand must vanish, meaning that $\d^AX=0$ and $[\Phi,X]=0$. In other words, $X$ is parallel and commutes with $\Phi$.

Since $X$ is parallel, its eigenvalues are constant.  Let $F_1$ be the bundle of eigenspaces for one of its eigenvalues.  Since $X$ is traceless, $F_1\neq E$.  Let $F_2$ be the orthogonal complement of $F_1$ obtained using the Hermitian metric $h$.  Since $\d^AX=0$ and $A$ is unitary, $A$ restricts to unitary connections on $F_1$ and $F_2$.  Similarly, since $[\Phi,X]=0$ and $\Phi$ is anti-Hermitian, $F_1$ and $F_2$ are $\Phi$-invariant, \ie $\Phi(F_i)\subset \Lambda^{1,0}C\otimes F_i$.  Since $A$ determines the holomorphic structure, $F_1$ and $F_2$ are both holomorphic sub-bundles of $E$.  Since $B$ is the $(1,0)$-part of $\Phi$, $B(F_i)\subset \Lambda^{1,0}\otimes F_i$ for $i=1,2$.  Since $E=F_1\oplus F_2$,
\begin{equation}
    \mathrm{deg}(E)=\mathrm{rank}(F_1)\mathrm{deg}(F_2)+\mathrm{rank}(F_2)\mathrm{deg}(F_1) \,.
\end{equation}
In terms of the slope $\mu$, this means that
\begin{equation}
    \mu(F_1)+\mu(F_2)=\frac{\mathrm{rank}(F_1)+\mathrm{rank}(F_2)}{\mathrm{rank}(F_1)\mathrm{rank}(F_2)}\mu(E)\geq2\mu(E) \,.
\end{equation}
This means that either $\mu(F_1)\geq \mu(E)$ or $\mu(F_2)\geq \mu(E)$.  Either way, we have a contradiction to the stability of $E$.

\paragraph{Data Availability} There is no data associated to this work.

\paragraph{Declarations}

\paragraph{Conflict of interest} The authors have no conflict of interest to disclose.

\end{document}